\documentclass[12pt,draftcls,onecolumn]{IEEEtran}

\usepackage{graphicx}   
\usepackage{amssymb}
\usepackage{amsmath}
\usepackage{epstopdf} 
\usepackage{xspace}
\usepackage{cancel}
\usepackage{url}

\usepackage[stretch=16,shrink=16,step=4]{microtype}

\usepackage{amssymb}
\usepackage{bbm}
\usepackage{mathrsfs}
\usepackage{xspace}

\newenvironment{textbmatrix}{	\setlength{\arraycolsep}{2.5pt}%
								\big[\begin{matrix}}{\end{matrix}\big]%
								\raisebox{0.08ex}{\vphantom{M}}}

\def\be{\begin{equation}}
\def\ee{\end{equation}}
\def\benn{\begin{equation*}}
\def\eenn{\end{equation*}}
\def\een{\nonumber \end{equation}}
\def\mat{\begin{bmatrix}}
\def\emat{\end{bmatrix}}
\def\btm{\begin{textbmatrix}}
\def\etm{\end{textbmatrix}}

\def\ba#1\ea{\begin{align}#1\end{align}}
\def\bann#1\eann{\begin{align*}#1\end{align*}}
\def\bs#1\es{\begin{split}#1\end{split}} 
\def\bg#1\eg{\begin{gather}#1\end{gather}} 
\def\bi#1\ei{\begin{itemize}#1\end{itemize}} 
\def\bsa#1\esa{\begin{IEEEeqnarray}{rCl}#1\end{IEEEeqnarray}}
\def\bsann#1\esann{\begin{IEEEeqnarray*}{rCl}#1\end{IEEEeqnarray*}}

\newcommand{\safemath}[2]{\newcommand{#1}{\ensuremath{#2}\xspace}}

\newcommand{\lefto}{\mathopen{}\left}

\DeclareMathOperator*{\argmin}{arg\;min}		
\DeclareMathOperator{\Exop}{\mathbb{E}}		

\safemath{\normal}{\mathcal{N}}				
\safemath{\circnorm}{\mathcal{CN}}			
\safemath{\uniform}{\mathcal{U}}				

\safemath{\interior}{\mathrm{Int}}			 

\safemath{\dfn}{:=}							
\newcommand{\dotgeq}{\dot{\geq}} 
\newcommand{\dotleq}{\dot{\leq}} 
\newcommand{\hbcomment}[1]{}

\newcommand{\cdind}[2]{\imath_{#1}^{#2}}
\newcommand{\cdindt}[2]{\hat{\imath}_{#1}^{#2}}

\newtheorem{theorem}{Theorem}
\newtheorem{remark}{Remark}
\newtheorem{definer}{Definition}

\newcommand{\prob}[1]{\ensuremath{\mathbb{P}\lefto[#1\right]}}

\safemath{\SNR}{\text{\sc snr}} 				
\safemath{\No}{N_0}							
\safemath{\Es}{E_s}							
\safemath{\Eb}{E_b}							
\safemath{\EbNo}{\frac{\Eb}{\No}}
\safemath{\EsNo}{\frac{\Es}{\No}}

\providecommand{\Hop}{\ensuremath{\mathbb{H}}} 
\safemath{\LH}{L_{\Hop}}						
\safemath{\SH}{S_\Hop}						
\safemath{\HH}{H_{\Hop}}						
\safemath{\CH}{C_\Hop}						
\safemath{\RH}{R_\Hop}						
\safemath{\Rh}{R_h}							

\safemath{\dB}{\,\mathrm{dB}}
\safemath{\dBm}{\,\mathrm{dBm}}
\safemath{\Hz}{\,\mathrm{Hz}}
\safemath{\kHz}{\,\mathrm{kHz}}
\safemath{\MHz}{\,\mathrm{MHz}}
\safemath{\GHz}{\,\mathrm{GHz}}
\safemath{\s}{\,\mathrm{s}}
\safemath{\ms}{\,\mathrm{ms}}
\safemath{\mus}{\,\mathrm{\mu s}}
\safemath{\ns}{\,\mathrm{ns}}
\safemath{\meter}{\,\mathrm{m}}
\safemath{\mm}{\,\mathrm{mm}}
\safemath{\cm}{\,\mathrm{cm}}
\safemath{\m}{\,\mathrm{m}}
\safemath{\W}{\,\mathrm{W}}
\safemath{\J}{\,\mathrm{J}}
\safemath{\K}{\,\mathrm{K}}
\safemath{\bit}{\,\mathrm{bit}}

\safemath{\define}{\triangleq}			

\safemath{\equivalent}{\sim}
\safemath{\distas}{\sim}					

\safemath{\reals}{\mathbb{R}}
\safemath{\positivereals}{\mathbb{R}^{+}}
\safemath{\integers}{\mathbb{Z}}
\safemath{\posint}{\mathbb{Z}_{+}}
\safemath{\naturals}{\mathbb{N}}
\safemath{\complexset}{\mathbb{C}}

\safemath{\setA}{\mathcal{A}}
\safemath{\setB}{\mathcal{B}}
\safemath{\setC}{\mathcal{C}}
\safemath{\setD}{\mathcal{D}}
\safemath{\setE}{\mathcal{E}}
\safemath{\setF}{\mathcal{F}}
\safemath{\setG}{\mathcal{G}}
\safemath{\setH}{\mathcal{H}}
\safemath{\setI}{\mathcal{I}}
\safemath{\setJ}{\mathcal{J}}
\safemath{\setK}{\mathcal{K}}
\safemath{\setL}{\mathcal{L}}
\safemath{\setM}{\mathcal{M}}
\safemath{\setN}{\mathcal{N}}
\safemath{\setO}{\mathcal{O}}
\safemath{\setP}{\mathcal{P}}
\safemath{\setQ}{\mathcal{Q}}
\safemath{\setR}{\mathcal{R}}
\safemath{\setS}{\mathcal{S}}
\safemath{\setT}{\mathcal{T}}
\safemath{\setU}{\mathcal{U}}
\safemath{\setV}{\mathcal{V}}
\safemath{\setW}{\mathcal{W}}
\safemath{\setX}{\mathcal{X}}
\safemath{\setY}{\mathcal{Y}}
\safemath{\setZ}{\mathcal{Z}}
\safemath{\emptySet}{\varnothing}

\safemath{\bma}{\mathbf{a}}
\safemath{\bmb}{\mathbf{b}}
\safemath{\bmc}{\mathbf{c}}
\safemath{\bmd}{\mathbf{d}}
\safemath{\bme}{\mathbf{e}}
\safemath{\bmf}{\mathbf{f}}
\safemath{\bmg}{\mathbf{g}}
\safemath{\bmh}{\mathbf{h}}
\safemath{\bmi}{\mathbf{i}}
\safemath{\bmj}{\mathbf{j}}
\safemath{\bmk}{\mathbf{k}}
\safemath{\bml}{\mathbf{l}}
\safemath{\bmm}{\mathbf{m}}
\safemath{\bmn}{\mathbf{n}}
\safemath{\bmo}{\mathbf{o}}
\safemath{\bmp}{\mathbf{p}}
\safemath{\bmq}{\mathbf{q}}
\safemath{\bmr}{\mathbf{r}}
\safemath{\bms}{\mathbf{s}}
\safemath{\bmt}{\mathbf{t}}
\safemath{\bmu}{\mathbf{u}}
\safemath{\bmv}{\mathbf{v}}
\safemath{\bmw}{\mathbf{w}}
\safemath{\bmx}{\mathbf{x}}
\safemath{\bmy}{\mathbf{y}}
\safemath{\bmz}{\mathbf{z}}

\safemath{\bmxi}{\boldsymbol{\xi}}
\safemath{\bmlambda}{\mathbf{\lambda}}
\safemath{\bmmu}{\mathbf{\mu}}
\safemath{\bmtheta}{\boldsymbol{\theta}}
\safemath{\bmphi}{\boldsymbol{\phi}}

\safemath{\bA}{\mathbf{A}}
\safemath{\bB}{\mathbf{B}}
\safemath{\bC}{\mathbf{C}}
\safemath{\bD}{\mathbf{D}}
\safemath{\bE}{\mathbf{E}}
\safemath{\bF}{\mathbf{F}}
\safemath{\bG}{\mathbf{G}}
\safemath{\bH}{\mathbf{H}}
\safemath{\bI}{\mathbf{I}}
\safemath{\bJ}{\mathbf{J}}
\safemath{\bK}{\mathbf{K}}
\safemath{\bL}{\mathbf{L}}
\safemath{\bM}{\mathbf{M}}
\safemath{\bN}{\mathbf{N}}
\safemath{\bO}{\mathbf{O}}
\safemath{\bP}{\mathbf{P}}
\safemath{\bQ}{\mathbf{Q}}
\safemath{\bR}{\mathbf{R}}
\safemath{\bS}{\mathbf{S}}
\safemath{\bT}{\mathbf{T}}
\safemath{\bU}{\mathbf{U}}
\safemath{\bV}{\mathbf{V}}
\safemath{\bW}{\mathbf{W}}
\safemath{\bX}{\mathbf{X}}
\safemath{\bY}{\mathbf{Y}}
\safemath{\bZ}{\mathbf{Z}}

\safemath{\bDelta}{\mathbf{\Delta}}
\safemath{\bLambda}{\mathbf{\Lambda}}
\safemath{\bPhi}{\mathbf{\Phi}}
\safemath{\bSigma}{\mathbf{\Sigma}}
\safemath{\bOmega}{\mathbf{\Omega}}
\safemath{\bTheta}{\mathbf{\Theta}}

\safemath{\bZero}{\mathbf{0}}

\safemath{\veca}{\bma}
\safemath{\vecb}{\bmb}
\safemath{\vecc}{\bmc}
\safemath{\vecd}{\bmd}
\safemath{\vece}{\bme}
\safemath{\vecf}{\bmf}
\safemath{\vecg}{\bmg}
\safemath{\vech}{\bmh}
\safemath{\veci}{\bmi}
\safemath{\vecj}{\bmj}
\safemath{\veck}{\bmk}
\safemath{\vecl}{\bml}
\safemath{\vecm}{\bmm}
\safemath{\vecn}{\bmn}
\safemath{\veco}{\bmo}
\safemath{\vecp}{\bmp}
\safemath{\vecq}{\bmq}
\safemath{\vecr}{\bmr}
\safemath{\vecs}{\bms}
\safemath{\vect}{\bmt}
\safemath{\vecu}{\bmu}
\safemath{\vecv}{\bmv}
\safemath{\vecw}{\bmw}
\safemath{\vecx}{\bmx}
\safemath{\vecy}{\bmy}
\safemath{\vecz}{\bmz}
\safemath{\vecZero}{\bZero}

\safemath{\vecxi}{\bmxi}
\safemath{\veclambda}{\bmlambda}
\safemath{\vecmu}{\bmmu}
\safemath{\vectheta}{\bmtheta}
\safemath{\vecphi}{\bmphi}

\safemath{\matA}{\bA}
\safemath{\matB}{\bB}
\safemath{\matC}{\bC}
\safemath{\matD}{\bD}
\safemath{\matE}{\bE}
\safemath{\matF}{\bF}
\safemath{\matG}{\bG}
\safemath{\matH}{\bH}
\safemath{\matI}{\bI}
\safemath{\matJ}{\bJ}
\safemath{\matK}{\bK}
\safemath{\matL}{\bL}
\safemath{\matM}{\bM}
\safemath{\matN}{\bN}
\safemath{\matO}{\bO}
\safemath{\matP}{\bP}
\safemath{\matQ}{\bQ}
\safemath{\matR}{\bR}
\safemath{\matS}{\bS}
\safemath{\matT}{\bT}
\safemath{\matU}{\bU}
\safemath{\matV}{\bV}
\safemath{\matW}{\bW}
\safemath{\matX}{\bX}
\safemath{\matY}{\bY}
\safemath{\matZ}{\bZ}
\safemath{\matZero}{\bZero}

\safemath{\matDelta}{\bDelta}
\safemath{\matLambda}{\bLambda}
\safemath{\matPhi}{\bPhi}
\safemath{\matSigma}{\bSigma}
\safemath{\matOmega}{\bOmega}
\safemath{\matTheta}{\bTheta}

\newcommand{\maxo}[1]{\ensuremath{(#1)^{+}}}
\newcommand{\ajml}[1]{\emph{joint ML decoder for #1}}
\newcommand{\ajmlong}[1]{\emph{joint ML decoder for #1}}

\makeatletter
\def\@IEEEinterspaceratioM{0.265}
\def\@IEEEinterspaceMINratioM{0.1651}
\def\@IEEEinterspaceMAXratioM{0.38}

\def\@IEEEinterspaceratioB{0.31}
\def\@IEEEinterspaceMINratioB{0.19}
\def\@IEEEinterspaceMAXratioB{0.38}
\@IEEEtunefonts
\makeatother

\IEEEoverridecommandlockouts

\begin{document}

\author{Cemal Ak\c{c}aba and Helmut B\"{o}lcskei
\thanks{The authors are with the Communication Technology Laboratory, ETH Zurich, 8092 Zurich, Switzerland (email:\{cakcaba {\textbar} boelcskei\}@nari.ee.ethz.ch).}
\thanks{Part of this paper will appear in the Proc. of the IEEE Int. Symposium on Information Theory (ISIT), Seoul, Korea, Jun. 2009, available:  http://arxiv.org/abs/0903.2226v2}
}

\title{Diversity-Multiplexing Tradeoff in Fading Interference Channels}

\maketitle

\begin{abstract}
We analyze two-user single-antenna fading interference channels with perfect receive channel state information (CSI) and no transmit CSI.     We compute the diversity-multiplexing tradeoff (DMT) region of a fixed-power-split Han and Kobayashi (HK)-type superposition coding scheme and provide design criteria for the corresponding superposition codes.  We demonstrate that this scheme is DMT-optimal under moderate, strong, and very strong interference by showing that it achieves a DMT outer bound that we derive.   Further, under very strong interference, we show that a joint decoder is DMT-optimal and ``decouples'' the fading interference channel, i.e., from a DMT perspective, it is possible to transmit as if the interfering user were not present.  In addition, we show that, under very strong interference,  decoding interference while treating the intended signal as noise, subtracting the result out, and then decoding the desired signal, a process known as ``stripping'', achieves the optimal DMT region.   Our proofs are constructive in the sense that code design criteria for achieving DMT-optimality (in the cases where we can demonstrate it) are provided. \end{abstract}

\section{Introduction}
The interference channel (IC) models the situation where $M$ unrelated transmitters communicate their separate messages to $M$ independent receivers, each of which is assigned to a single transmitter.   Apart from a few special cases \cite{Sato81,Carleial75,Sason04}, the capacity region of the IC remains unknown.  Recently, for the interference-limited regime, Etkin \emph{et al.} \cite{Etkin06, Etkin08} characterized the capacity region of the IC to within one bit.  Later, Telatar and Tse \cite{Telatar07} generalized this result to a wider class of ICs.  The techniques used in \cite{Etkin06, Etkin08, Telatar07} rely on perfect channel state information (CSI) at the transmitter. Shang \emph{et al.} derived the noisy-interference sum-rate capacity for Gaussian ICs  in \cite{shangkram08}, while Raja \emph{et al.} \cite{rajavis08} characterized the capacity region of the two-user finite-state compound Gaussian IC to within one bit.  Annapureddy and Veeravalli \cite{veeravalli}  showed that the sum capacity of the two-user Gaussian IC, under weak interference, is achieved by treating interference as noise.

In \cite{Akuiyibo08},  Akuiyibo and L\'{e}v\^{e}que derived an outer bound on the diversity-multiplexing tradeoff (DMT) region for the two-user IC based on the results of Etkin \emph{et al.} \cite{Etkin08}.    In this paper, we investigate the achievability of this outer bound and we analyze the DMT region realized by a fixed-power-split Han and Kobayashi (HK)-type superposition coding scheme.  For the sake of simplicity of exposition,  we restrict our attention to the two-user case throughout the paper. Furthermore, we assume that  the receivers have perfect CSI  whereas the transmitters only know the channel statistics.  We would like to point out that the schemes used in \cite{Etkin08} make explicit use of transmit CSI and so does the scheme in \cite{Akuiyibo08}, which immediately implies that the results reported in \cite{Akuiyibo08} serve as an outer bound on the DMT region achievable in the absence of transmit CSI, the case considered here.  The contributions in this paper can be summarized as follows:
\begin{itemize}
\item{For general interference levels, we compute the DMT region of a two-message, fixed-power-split HK-type superposition coding scheme and provide design criteria for the corresponding superposition codes.   For the case where the multiplexing rates of the two transmitters are equal, we demonstrate that the two-message, fixed-power-split HK-type superposition coding scheme achieves the optimal DMT  of the two-user IC under \emph{moderate}, \emph{strong}, and \emph{very strong interference}.  For asymmetric rates, i.e., when the multiplexing rates of the two transmitters are not equal, we prove that the two message, fixed-power-split HK scheme is also DMT-optimal in the \emph{strong} and \emph{very strong interference} regimes.}
\item{Under \emph{very strong interference}, a \emph{joint decoder}, i.e., a decoder that jointly decodes the transmitted messages of both transmitters at each receiver, ``decouples'' the fading IC, i.e.,  from a DMT perspective, the achievable performance is equivalent to that of a system with two isolated single-user links.}
\item{For \emph{very strong interference}, we show that a \emph{stripping decoder}, which decodes interference while treating the intended signal as noise, subtracts the result out, and then decodes the intended signal is DMT-optimal. We furthermore show that the optimal DMT can be achieved if each of the two transmitters employs a code that is DMT-optimal on a single-input single-output (SISO) channel.}
\end{itemize}

\subsubsection*{Notation}
The superscripts $^{T}$ and $^{H}$ stand for transpose and conjugate transpose, respectively.
$x_{i}$ represents the $i$th element of the column vector $\mathbf{x}$, and $\lambda_{\min}(\mathbf{X}) $ denotes the smallest eigenvalue of the matrix $\mathbf{X}$.   $\matI_{N}$ is the $N\times N$ identity matrix, and $\mathbf{0}$ denotes the all zeros matrix of appropriate size.  All logarithms are to the base $2$ and $(a)^{+}=\max(a,0)$.  
$X\sim \mathcal{CN}(0,\sigma^2)$ stands for a circularly symmetric complex Gaussian random variable (RV) with variance $\sigma^{2}$. $f(\rho) \doteq g(\rho)$ denotes exponential equality of the functions $f(\cdot)$ and $g(\cdot)$, i.e.,  
\[\lim_{\rho\rightarrow\infty} \frac{\log f(\rho)}{\log \rho}  = \lim_{\rho\rightarrow\infty} \frac{\log g(\rho)}{\log \rho} . \]
The symbols $\dotgeq$, $\dotleq$, $\dot{>}$, and $\dot{<}$ are defined analogously.

\subsubsection*{System model}
We consider a two-user fading IC where two transmitters communicate information to two receivers via a common channel. The fading coefficient between transmitter $i$  $(i=1,2)$ and receiver $j$ $(j=1,2)$ is denoted by $h_{ij}$ and is assumed to be $\mathcal{CN}(0,1)$.  Transmitter $i$ ($ \mathcal{T}_{i}$) chooses an $N$-dimensional codeword $\vecx_{i} \in \mathbb{C}^{N}$, $\|\vecx_{i}\|^{2} \leq N$,  from its codebook, and transmits $\check{\vecx}_{i}= \sqrt{P_{i}}\vecx_{i}$ in accordance with its transmit power constraint $\|\check{\vecx}_{i}\|^{2} \leq NP_{i}$.  In addition, we account for the attenuation of transmit signal $i$ at receiver $j$ ($\mathcal{R}_{j}$) through the real-valued coefficients $\eta_{ij} > 0$.   Defining $\vecy_{i}$ and $\vecz_{i}\sim \mathcal{CN}(\mathbf{0},\matI_{N})$ as the $N$-dimensional received signal vector and noise vector, respectively, at $\mathcal{R}_{i}$, the input-output relations are given by  
 \begin{align}
\label{intc1}
	\vecy_{1} &= \eta_{11}h_{11}\check{\vecx}_{1}+\eta_{21}h_{21}\check{\vecx}_{2}+\vecz_{1} \\
\label{intc2}	\vecy_{2} &= \eta_{12}h_{12}\check{\vecx}_{1}+\eta_{22}h_{22}\check{\vecx}_{2}+\vecz_{2}.
\end{align}
Setting $\eta_{11}^{2}P_{1} = \eta_{22}^{2}P_{2} = \SNR$ and $\eta_{21}^{2}P_{2} = \eta_{12}^{2}P_{1} = \SNR^{\alpha}$ with $\alpha \in [0,\infty]$ simplifies the exposition of our results, and the comparison to  \cite{Etkin08} and \cite{Akuiyibo08}. The resulting equivalent set of input-output relations is 
\begin{align}
\label{intc1s}    \vecy_{1} &= \sqrt{\SNR}h_{11}\vecx_{1}+\sqrt{\SNR^{\alpha}}h_{21}\vecx_{2}+ \vecz_{1} \\
\label{intc2s}	\vecy_{2} &= \sqrt{\SNR^{\alpha}}h_{12}\vecx_{1}+\sqrt{\SNR}h_{22}\vecx_{2}+\vecz_{2}.
\end{align} 
We assume that both receivers know the signal-to-noise ratio (SNR) value $\SNR$ and the parameter $\alpha$, and $\mathcal{R}_{i}$ $(i=1,2)$ knows $\vech_{i} = [h_{1i} \  h_{2i}]^{T}$ perfectly,  whereas the transmitters only know the channel statistics for the channels $h_{ij}$ $\left(i,j=1,2\right)$, the SNR value, and the interference parameter $\alpha$.  The data rate of $ \mathcal{T}_{i}$ scales with SNR according to $R_{i}=r_{i}\log\SNR$, where the multiplexing rate $r_{i}$ obeys $0 \leq r_{i} \leq 1$.   As a result, for $ \mathcal{T}_{i}$ to operate at multiplexing rate $r_{i}$,  we need a sequence of codebooks $\mathcal{C}_{i}(\SNR,r_{i})$, one for each $\SNR$, with $| \mathcal{C}_{i}(\SNR,r_{i})|=2^{NR_{i}}$ codewords $\{\vecx^{1}_{i},\vecx^{2}_{i},\ldots, \vecx_{i}^{2^{NR_{i}}}\}$.  In the following, we will need the multiplexing rate vector $\vecr=[r_{1} \ r_{2}]^{T}$.

\subsubsection*{Performance metric} The  error probability corresponding to maximum-likelihood (ML) decoding of $ \mathcal{T}_{i}$ at $\mathcal{R}_{i}$ under the assumption that the correctly decoded interference $\mathcal{T}_{j}$ has been removed is denoted by $\prob{E_{ii}|\vech_{i}}$ for $i,j=1,2$ and $i\neq j$. The corresponding average (with respect to (w.r.t.) the random channel)  error probability is  $P(E_{ii}) \triangleq \Exop_{\vech_{i}}\!\!\left\{\prob{E_{ii}|\vech_{i}}\right\}$. The notation $\vecx_{i}^{j}\rightarrow \vecx_{i}^{k}$ designates the event of mistakenly decoding the transmitted codeword $\vecx_{i}^{j}$ for the codeword $\vecx_{i}^{k}$.   

The average (w.r.t. the random channel) error probability corresponding to decoding of $\mathcal{T}_{i}$ at $\mathcal{R}_{i}$ incurred by a particular communication scheme $\chi$ is denoted by $P(E^{\chi}_{i})$ for $i=1,2$ .  Throughout the paper,  as done in \cite{Akuiyibo08}, we use the performance metric $P(E^{\chi})= \max\{P(E_{1}^{\chi}), P(E_{2}^{\chi})\}$.   The DMT realized by a communication scheme $\chi$ is then characterized by
\begin{align}
	d^{\chi}(\vecr) = -\lim_{\SNR\rightarrow \infty}\frac{\log P(E^{\chi})}{\log\SNR}.
\end{align}
As discussed in \cite{tse_mu,verdu98}, the receiver that minimizes the error probability for each $\mathcal{T}_{i}$ is the \emph{individual ML receiver} at $\mathcal{R}_{i}$ for $i=1,2$, which we define next. 
\begin{definer}
An \emph{individual ML receiver} for $\mathcal{T}_{i}$ at $\mathcal{R}_{j}$ for $i,j=1,2$ treats the signal from $\mathcal{T}_{k}$ for $k=1,2$, $k\neq i$, as discrete noise with known structure (i.e., codebooks) and carries out an ML detection of the message of $\mathcal{T}_{i}$ \cite{tse_mu, verdu98}. In the following,  we denote the error probability of an individual ML receiver for $\mathcal{T}_{i}$ at $\mathcal{R}_{j}$ by $\prob{\mathcal{E}_{ij}^{IML}}$ for $i,j=1,2$.  The corresponding average (w.r.t. the random channel) error probability is denoted by $P(E^{IML}_{ij}) \triangleq \mathbb{E}_{\vech_{j}}\lefto\{\prob{\mathcal{E}_{ij}^{IML}}\right\}$.
\end{definer}

The DMT realized by the strategy of employing an \emph{individual ML receiver} for $\mathcal{T}_{i}$ at each receiver $\mathcal{R}_{i}$ for $i=1,2$ is given by 
\begin{align}
	d^{IML}(\vecr) = -\lim_{\SNR\rightarrow \infty}\frac{\log \max\lefto\{P(E^{IML}_{11}), P(E^{IML}_{22}) \right\}}{\log\SNR}.
\end{align}
Since the \emph{individual ML receiver} minimizes the error probability for each $\mathcal{T}_{i}$ at $\mathcal{R}_{i}$ for $i=1,2$, we have that the DMT $d^{IML}(\vecr)$ is an outer bound on the DMT realized by any communication scheme $\chi$, i.e.,
\begin{align}
	d^{IML}(\vecr) \geq 	d^{\chi}(\vecr).
\end{align}


\section{Achievable DMT for Joint Decoding}
\label{jointML}
A simple achievable rate region for the IC is obtained by having each receiver perform joint decoding of the messages from both transmitters.  Hence, there are no \emph{private} messages, i.e., there are no messages that should only be decoded at one receiver, and the messages of both transmitters are said to be \emph{public}.  We formally define the \emph{joint decoder} or \ajmlong{IC} next. 
\begin{definer}
A \ajmlong{IC} at $\mathcal{R}_{j }$ ($j=1,2$)  carries out  joint ML detection on the messages from both transmitters ($\mathcal{T}_{i}$ for $i=1,2$).  For the \ajml{IC} at $\mathcal{R}_{j}$, one does not declare an error if the estimate of the signal from $\mathcal{T}_{i}$ does not match the transmitted signal from $\mathcal{T}_{i}$ for $i,j=1,2$ and $i\neq j$.  The error probability of this receiver is denoted by $\prob{\mathcal{E}^{JD}_{j}}$. Then, $\prob{\mathcal{E}^{JD}_{j}}$ is the probability that only $\mathcal{T}_{j}$ or both $T_{i}$ and $T_{j}$ for $i,j=1,2$ and $i\neq j$  are decoded incorrectly.  The corresponding average (w.r.t. the random channel) error probability is denoted by $P\lefto(E^{JD}_{j}\right)\triangleq \mathbb{E}_{\vech_{j}}\lefto\{\prob{\mathcal{E}^{JD}_{j}}\right\}$.
\end{definer}
	
	The achievable DMT of the \ajmlong{IC} is characterized next.

\begin{theorem}
\label{theoremJD}
The DMT corresponding to joint decoding at each receiver is given by
\begin{align}
d^{JD}(\vecr) =   \min\limits_{i=1,2,3}\left(d^{JD}_{i}(\vecr)\right) \label{jfafjfffffsgh}
\end{align}
where
 \begin{align}
	&d^{JD}_{i}(\vecr)= (1-r_{i})^{+}, \ \ \ \ {\rm for} \ \ \  i=1,2 \\
	& d^{JD}_{3}(\vecr) = \lefto(1-r_{1}-r_{2}\right)^{+}+ \lefto(\alpha-r_{1}-r_{2}\right)^{+}. \notag
\end{align}
Denote $j^{*}= \arg \min_{i=1,2,3} d^{JD}_{i}(\vecr)$.  Let $\Gamma_{i}(\vecr) =[\gamma_{i}^{1}(\vecr) \ \gamma_{i}^{2}(\vecr)]^{T}$ be functions\footnote{We note that the functions $\Gamma_{i}(\vecr)$ might not be unique.} such that $d^{JD}_{j^{*}}(\vecr) = d^{JD}_{i}(\Gamma_{i}(\vecr))$  for $i=1,2,3$. 
If a sequence (in SNR) of codebooks with block length $N \geq 2$ satisfies
\begin{align}
	\| \Delta \vecx_{i}\|^{2} \ & \dotgeq \ \SNR^{-\gamma_{i}^{i}(\vecr)+\epsilon}, \label{xxczzxvz}\\
	\lambda_{\min}\lefto(\Delta\matX_{ij} (\Delta\matX_{ij})^{H}\right) \ & \dotgeq \ \SNR^{-\gamma_{3}^{1}(\vecr)-\gamma_{3}^{2}(\vecr)+\epsilon} \label{xxczzxvz2}
\end{align}
for all pairs of codewords $\vecx_{i}^{n_{i}},  \vecx_{i}^{\tilde{n}_{i}} \in \mathcal{C}_{i}(\SNR,r_{i})$ s.t. $\vecx_{i}^{n_{i}} \neq \vecx^{\tilde{n}_{i}}_{i}$, $\vecx_{j}^{n_{j}},  \vecx_{j}^{\tilde{n}_{j}} \in \mathcal{C}_{j}(\SNR,r_{j})$  s.t. $\vecx_{j}^{n_{j}} \neq \vecx^{\tilde{n}_{j}}_{j}$  for $i,j = 1,2$ and $i \neq j$,  where  $\Delta \vecx_{i} = \vecx_{i}^{n_{i}} - \vecx^{\tilde{n}_{i}}_{i}$, $\Delta \vecx_{j} = \vecx_{j}^{n_{j}} - \vecx^{\tilde{n}_{j}}_{j}$,   and $\Delta\matX_{ij} = [\Delta\vecx_{i} \ \Delta\vecx_{j}]$, and $\lambda_{\min}(\Delta\matX_{ij} (\Delta\matX_{ij})^{H})$ denotes the smallest nonzero eigenvalue of $\Delta\matX_{ij} (\Delta\matX_{ij})^{H}$, for some\footnote{We note that $\epsilon$ is allowed to be different in \eqref{xxczzxvz} and \eqref{xxczzxvz2}. } $\epsilon >  0$, then $P\lefto(E^{JD}\right)$ obeys
\begin{align}
P\lefto(E^{JD}\right) \doteq \SNR^{-d^{JD}(\vecr)}.
\end{align}
\end{theorem}



\begin{proof}
We first identify a lower bound on $P\lefto(E^{JD}\right)$, which constitutes an upper bound on the DMT of the \ajmlong{IC}, and then show, using an appropriate upper bound on $P(E^{JD})$, that the SNR exponents of the upper and lower bounds on $P(E^{JD})$ match at high SNR.  Hence, the upper bound on the DMT of the \ajmlong{IC} is achievable.  We define the outage events corresponding to decoding $\mathcal{T}_{i}$ at $\mathcal{R}_{i}$ (in the absence of a signal from $\mathcal{T}_{j}$) and to jointly decoding $\mathcal{T}_{i}$ \emph{and} $\mathcal{T}_{j}$ at $\mathcal{R}_{i}$ for $i,j=1,2$ and $i\neq j$ by 
\begin{align}
	\mathcal{O}_{i1}^{JD} &\triangleq \lefto\{ \vech_{i} : I(\vecx_{i}; \vecy_{i} | \vecx_{j}, \vech_{i}) < R_{i}  \right\} \label{charnomac1} \\
	\mathcal{O}_{i2}^{JD} &\triangleq \lefto\{ \vech_{i} : I(\vecx_{i},\vecx_{j} ; \vecy_{i} | \vech_{i}) < R_{1}+R_{2} \right\}  \label{charnomac2}.
\end{align}
We define an outage event at $\mathcal{R}_{i}$ for the IC as
\begin{align}
	\mathcal{O}^{JD}_{i} \triangleq \bigcup_{k=1}^{2} \mathcal{O}^{JD}_{ik} \label{outeven}
\end{align}
for $i=1,2$.  We would like to point out that the definition of the outage event in \eqref{outeven} is different from the corresponding outage event definition in multiple access channels (MACs) \cite{tse_mu, isit08_cgb} as the outage event corresponding to  decoding of $\mathcal{T}_{i}$ at $\mathcal{R}_{j}$ is absent in \eqref{outeven}.   We note that only $\mathcal{T}_{j}$ being decoded in error at $\mathcal{R}_{i}$ for $i\neq j$, although being a standard error event for the MAC,  is not (and \emph{should not}) be defined as an error event for the IC.  As long as the decision on $\mathcal{T}_{i}$ at $\mathcal{R}_{i}$ is correct, from the point of view of the IC, there is no error to be declared.  The probability of outage yields a lower bound on the error probability of the \ajml{IC}. As in \cite{Akuiyibo08}, we define the total outage probability of the IC as 
\begin{align}
\prob{\mathcal{O}^{JD}} \triangleq  \max\lefto\{\prob{	\mathcal{O}^{JD}_{1}}\! , 	\prob{\mathcal{O}^{JD}_{2} }\right\}.
\end{align} 
Using a standard argument along the lines of \cite{tse_mu, isit08_cgb}, we can see that assuming that both transmitters employ i.i.d. Gaussian codebooks results in no loss of optimality in terms of DMT performance.  We can therefore evaluate \eqref{charnomac1} and \eqref{charnomac2} as 
\begin{align}
	\mathcal{O}_{i1}^{JD}(\vecr) &\triangleq \lefto\{ \vech_{i}:  \log\left(1+\SNR|h_{ii}|^2\right) < R_{i}  \right\} \notag \\
	\mathcal{O}_{i2}^{JD}(\vecr) &\triangleq \lefto\{ \vech_{i} :\log\left(1+\SNR^{\alpha}|h_{ji}|^{2}+\SNR|h_{ii}|^{2}\right) < R_{1} + R_{2 }\right\}. \notag
\end{align}
In the following, we will also need the definitions of the no-outage events, according to
\begin{align}
	\bar{\mathcal{O}}_{i1}^{JD}(\vecr) &\triangleq \lefto\{ \vech_{i}:  \log\left(1+\SNR|h_{ii}|^2\right) \geq R_{i}  \right\} \notag \\
	\bar{\mathcal{O}}_{i2}^{JD}(\vecr) &\triangleq \lefto\{ \vech_{i} :\log\left(1+\SNR^{\alpha}|h_{ji}|^{2}+\SNR|h_{ii}|^{2}\right) \geq R_{1}+R_{2} \right\} \notag
\end{align}
with $i,j=1,2$ and $i\neq j$.  We can now establish the asymptotic behavior of $\mathcal{O}^{JD}_{i}$. By the union bound, we have 
\begin{align}
	\prob{\mathcal{O}^{JD}_{i}}  \leq  \sum_{k=1}^{2}  \prob{\mathcal{O}^{JD}_{ik}(\vecr)}.
\end{align}
Obviously, it holds that 
\begin{align}
	\prob{\mathcal{O}^{JD}_{i}} \doteq \max_{k=1,2} \prob{\mathcal{O}^{JD}_{ik}(\vecr)}.  \label{maxomizc}
\end{align}
It is shown in \cite{zheng_tradeoff} and \cite{Akuiyibo08} that
\begin{align}
\prob{\mathcal{O}_{i1}^{JD}(\vecr)} &\doteq \SNR^{-d_{i1}^{JD}(\vecr)} \\
\prob{\mathcal{O}_{i2}^{JD}(\vecr)} &\doteq \SNR^{-d_{i2}^{JD}(\vecr)} 
\end{align}
with 
\begin{align}
d_{i1}^{JD}(\vecr) &= (1-r_{i})^{+}  \label{jod1}\\
d_{i2}^{JD}(\vecr) &= (1-r_{1}-r_{2})^{+}+(\alpha-r_{1}-r_{2})^{+} \label{jod2}
\end{align}
for $i=1,2$.  We point out that \eqref{jod1} and \eqref{jod2} define four SNR exponents $d_{ij}^{JD}(\vecr)$ for $i,j=1,2$.  The outage event corresponding to jointly decoding the signals from both transmitters at $\mathcal{R}_{1}$ is identical to the outage event corresponding to jointly decoding the signals from both transmitters at $\mathcal{R}_{2}$. Hence,  the corresponding SNR exponents of the outage probabilities of these events, namely, $d_{12}^{JD}(\vecr)$ and $d_{22}^{JD}(\vecr)$, are exactly the same. The total outage probability of the IC then behaves according to 
\begin{align}
\prob{\mathcal{O}^{JD}} &=  \max\lefto\{\prob{	\mathcal{O}^{JD}_{1} }, 	\prob{\mathcal{O}^{JD}_{2} }\right\}.  \label{afasfasazzzzkmf}
\end{align}
From \eqref{maxomizc}, it follows that
\begin{align}
	\prob{\mathcal{O}^{JD}_{i}} \doteq \max_{k=1,2} \prob{\mathcal{O}^{JD}_{ik}(\vecr)} \doteq \SNR^{-\min\limits_{k=1,2} d_{ik}^{JD}(\vecr)} \label{xxazsfafgafg}.
\end{align}
Hence, combining \eqref{afasfasazzzzkmf} and \eqref{xxazsfafgafg}, we get 
 \begin{align}
\prob{\mathcal{O}^{JD}}	  &  \doteq   \max_{i=1,2}\SNR^{-\min_{k=1,2} d_{ik}^{JD}(\vecr)} \label{ffffazcvgn} \\
							  & \doteq  \SNR^{-d^{JD}(\vecr)} \label{twentyone}\end{align}
where 
\begin{align}
d^{JD}(\vecr) =   \min\limits_{i=1,2,3}\left(d^{JD}_{i}(\vecr)\right) \label{fffazcaf}
\end{align}
with
 \begin{align}
	&d^{JD}_{i}(\vecr)= (1-r_{i})^{+} \ \ \ \ \text{for} \ \ \  i=1,2  \label{hoz} \\ 
	& d^{JD}_{3}(\vecr) = \lefto(1-r_{1}-r_{2}\right)^{+}+ \lefto(\alpha-r_{1}-r_{2}\right)^{+}.  \label{vzas}\notag
\end{align}
We note that \eqref{ffffazcvgn} can be simplified by eliminating either $d_{12}^{JD}(\vecr)$ or $d_{22}^{JD}(\vecr)$ as explained earlier. This is precisely what we have done in going from \eqref{ffffazcvgn} to \eqref{twentyone}.

With \eqref{xxazsfafgafg} we arrived at a lower bound on the error probability of the \ajml{IC} at $\mathcal{R}_{i}$.  This lower bound, by definition, gives an upper bound on the DMT region.  We next try to find an upper bound on the error probability that has the same exponential behavior as this lower bound.  To this end, consider next the error probability corresponding to the \ajml{IC}. We first define the relevant error events.  Let $\vecx_{i}^{n_{i}}$ and $\vecx_{j}^{n_{j}}$ with $n_{i} \in \{1,2,\ldots, 2^{NR_{i}}\}$, $n_{j} \in \{1,2,\ldots, 2^{NR_{j}}\}$ ($i,j=1,2$ and $i\neq j$) be the  codewords transmitted by $\mathcal{T}_{i}$ and $\mathcal{T}_{j}$, respectively. The results of (joint ML) decoding of $\mathcal{T}_{i}$ and $\mathcal{T}_{j}$ at $\mathcal{R}_{i}$ are denoted by $\vecx_{i}^{\tilde{n}_{i}}$ and $\vecx_{j}^{\tilde{n}_{j}}$, respectively, with $\tilde{n}_{i} \in \{1,2,\ldots, 2^{NR_{i}}\}$, $\tilde{n}_{j} \in \{1,2,\ldots, 2^{NR_{j}}\}$ for $i,j=1,2$ and $i\neq j$.   We have the error events corresponding to $\mathcal{T}_{i}$ only and $\mathcal{T}_{i}$ and $\mathcal{T}_{j}$ being decoded in error at $\mathcal{R}_{i}$ as
\begin{align}
\mathcal{E}_{i1}^{JD} & \triangleq \left\{ \tilde{n}_{i}\neq n_{i}, \  \tilde{n}_{j} = n_{j}\right\}\\ 
\mathcal{E}_{i2}^{JD} & \triangleq \left\{ \tilde{n}_{i}\neq n_{i}, \  \tilde{n}_{j} \neq n_{j}\right\} 
\end{align}
for $i,j=1,2$ and $i\neq j$.  We will also need the total error probability defined as 
\begin{align}
	\mathcal{E}^{JD}_{i} \triangleq \bigcup_{k=1,2}\mathcal{E}_{ik}^{JD}.
\end{align}

We denote $j^{*}= \arg \min_{i=1,2,3} d^{JD}_{i}(\vecr)$.  Let $\Gamma_{i}(\vecr) =[\gamma_{i}^{1}(\vecr) \ \gamma_{i}^{2}(\vecr)]^{T}$ be functions\footnote{We note that the functions $\Gamma_{i}(\vecr)$ might not be unique.} such that $d^{JD}_{j^{*}}(\vecr) = d^{JD}_{i}(\Gamma_{i}(\vecr))$  for  $i=1,2,3$.  We recall that $d_{i2}^{JD}(\vecr) = d_{3}^{JD}(\vecr)$ for $i=1,2,$ by definition.  

We next find an upper bound on the probability of the events $\mathcal{E}_{i1}^{JD}$ as follows:
\begin{align}
&\prob{\mathcal{E}^{JD}_{i1}} = \prob{\mathcal{E}^{JD}_{i1}, \mathcal{O}^{JD}_{i1}(\Gamma_{i}(\vecr))} +\prob{\mathcal{E}^{JD}_{i1}, \bar{\mathcal{O}}^{JD}_{i1}(\Gamma_{i}(\vecr))}  \notag \\
&\leq   \prob{ \mathcal{O}^{JD}_{i1}(\Gamma_{i}(\vecr))}+ \prob{\mathcal{E}^{JD}_{i1}| \bar{\mathcal{O}}^{JD}_{i1}(\Gamma_{i}(\vecr))} \label{fadgk}
\end{align}
and for the events $\mathcal{E}_{i2}^{JD}$ according to:  
\begin{align}
&\prob{\mathcal{E}^{JD}_{i2}} = \prob{\mathcal{E}^{JD}_{i2}, \mathcal{O}^{JD}_{i2}(\Gamma_{3}(\vecr))} +\prob{\mathcal{E}^{JD}_{i2}, \bar{\mathcal{O}}^{JD}_{i2}(\Gamma_{3}(\vecr))}  \notag \\
&\leq   \prob{ \mathcal{O}^{JD}_{i2}(\Gamma_{3}(\vecr))}+ \prob{\mathcal{E}^{JD}_{i2}| \bar{\mathcal{O}}^{JD}_{i2}(\Gamma_{3}(\vecr))} \label{fadgfj}. 
\end{align}
We start by deriving an upper bound on the average (w.r.t. the random channel) pairwise error probability (PEP) of each error event $\mathcal{E}^{JD}_{ik}$ for $i=1,2$ and $k=1,2$.   Assuming, without loss of generality,  that we have an $\mathcal{E}^{JD}_{i2}$ type event, the probability of the ML decoder mistakenly deciding in favor of the codeword $\matX_{ij}^{\tilde{n}_{i}\tilde{n}_{j}}=[\vecx_{i}^{\tilde{n}_{i}} \ \vecx_{j}^{\tilde{n}_{j}}]$ when $\matX^{n_{i}n_{j}}_{ij} =[\vecx^{n_{i}}_{i} \ \vecx^{n_{j}}_{j}]$ (with $\vecx_{i}^{n_{i}}, \vecx^{\tilde{n}_{i}}_{i} \in \mathcal{C}_{i}(\SNR,r_{i})$ and $\vecx_{j}, \vecx^{\tilde{n}_{j}}_{j} \in \mathcal{C}_{j}(\SNR,r_{j})$, $i,j=1,2$ and $i\neq j$) was actually transmitted, can be upper-bounded according to
\begin{align}
&	 \mathbb{E}_{\vech_{i}}\! \lefto\{\prob{\matX_{ij}^{n_{i}n_{j}} \rightarrow \matX^{\tilde{n}_{i}\tilde{n}_{j}}_{ij}}\right\} \\ &  \leq  \mathbb{E}_{\vech_{i}}\lefto\{\exp\left[- \frac{\|\Delta\matX_{ij} \tilde{\vech}_{i} \|^2}{4} \right]\right\} \\ 
											  & \leq  \mathbb{E}_{\vech_{i}}\lefto\{\exp\left[- \frac{\lambda_{\min}\|\tilde{\vech}_{i} \|^2}{4} \right]\right\} \\ 													    
											  &  = \mathbb{E}_{\vech_{i}}\lefto\{\exp\left[- \lambda_{\min} \frac{\SNR|h_{ii}|^2+\SNR^{\alpha}|h_{ji}|^{2}}{4} \right]\right\} \label{stars}
\end{align}
where $\tilde{\vech}_{i} = [\sqrt{\SNR}h_{ii} \ \sqrt{\SNR^{\alpha}}h_{ji}]^T $  for $i,j=1,2$ and $i\neq j$ and $\lambda_{\min}$ is the smallest nonzero eigenvalue of $\Delta\matX_{ij}(\Delta\matX_{ij})^{H}$.

Noting that the no outage event  $\bar{\mathcal{O}}^{JD}_{i2}\left(\Gamma_{3}(\vecr)\right)$ entails $\SNR|h_{ii}|^2+\SNR^{\alpha}|h_{ji}|^{2} \geq \SNR^{\gamma_{3}^{1}(\vecr)+\gamma_{3}^{2}(\vecr)}-1$, \eqref{fadgk} implies  an upper bound on $\prob{\mathcal{E}^{JD}_{i2}}$ according to:
\begin{align}
&\mathbb{E}_{\vech_{i}} \! \lefto\{\prob{\mathcal{E}^{JD}_{i2}}\right\} \ \dotleq \ \label{xxvafa} \\&  \prob{ \mathcal{O}_{i2}^{JD}\left(\Gamma_{3}(\vecr)\right)}+ \SNR^{N(r_{1}+r_{2})}\exp\left[-\frac{\lambda_{\min}\SNR^{\gamma_{3}^{1}(\vecr)+\gamma_{3}^{2}(\vecr)}}{4}\right]   \notag.
\end{align}
Here, we used the definitions $R_{i}=r_{i}\log\SNR$ for $i=1,2$ and $\exp[-\frac{\lambda_{\min}}{4} (\SNR^{\gamma_{3}^{1}(\vecr)+\gamma_{3}^{2}(\vecr)}-1) ] \doteq \exp[-\frac{\lambda_{\min}}{4}\SNR^{\gamma_{3}^{1}(\vecr)+\gamma_{3}^{2}(\vecr)}]$. Given that $\lambda_{\min} \ \dotgeq \ \SNR^{-\gamma_{3}^{1}(\vecr)-\gamma_{3}^{2}(\vecr)+\epsilon}$ with $\epsilon > 0$, by assumption, we obtain
\begin{align}
&\mathbb{E}_{\vech_{i}}\!  \left\{\prob{\mathcal{E}^{JD}_{i2}}\right\}  \notag \\ 
& \ \dotleq \ \prob{ \mathcal{O}_{i2}^{JD}\left(\Gamma_{3}(\vecr)\right)}+ \SNR^{N(r_{1}+r_{2})}\exp\left[-\frac{\SNR^{\epsilon}}{4}\right]   \label{xafkasfkaskfsfk} \\								   & \doteq  \prob{ \mathcal{O}_{i2}^{JD}\left(\Gamma_{3}(\vecr)\right)} \\
&\doteq \SNR^{-d_{j^{*}}^{JD}(\vecr)}\label{finzcuoqw}
\end{align}
as the second term on the right-hand-side (RHS) of \eqref{xafkasfkaskfsfk} decays exponentially in SNR whereas the first term decays polynomially. Eq. \eqref{finzcuoqw} follows by the definition of the function $\Gamma_{3}(\vecr)$. 

A similar analysis for the $\mathcal{E}^{JD}_{i1}$-type error event results in  
 \begin{align}
 & \mathbb{E}_{\vech_{i}}\lefto\{\prob{\vecx_{i}^{n_{i}} \rightarrow \vecx^{\tilde{n}_{i}}_{i}}\right\}  \leq \notag \\ 
 & \hspace{2cm} \mathbb{E}_{\vech_{i}}\lefto\{\exp\left[-  \frac{\SNR|h_{ii}|^2\|\Delta\vecx_{i}\|^{2}}{4} \right]\right\} \label{gasgafas}
 \end{align} 
which, upon invoking \eqref{xxczzxvz} and using the fact that $\bar{\mathcal{O}}_{i1}^{JD}\left(\Gamma_{i}(\vecr)\right)$ entails $\SNR|h_{ii}|^2 \geq \SNR^{\gamma_{i}^{i}}-1$, yields
\begin{align}
&\mathbb{E}_{\vech_{i}}\lefto\{\prob{\mathcal{E}^{JD}_{i1}}\right\} \notag \\ 
& \ \ \ \dotleq \ \prob{ \mathcal{O}_{i1}^{JD}\left(\Gamma_{i}(\vecr)\right)}+ \SNR^{Nr_{i}}\exp\left[-\frac{\SNR^{\epsilon}}{4}\right]  \label{mosnzzcP} \\
& \ \ \doteq \prob{ \mathcal{O}_{i1}^{JD}\left(\Gamma_{i}(\vecr)\right)}
 \end{align}
for $i=1,2$. To complete the proof, we note that
\begin{align}
\mathbb{E}_{\vech_{i}}\lefto\{\prob{\mathcal{E}^{JD}_{i}} \right\} & \leq \sum_{k=1}^{2}\mathbb{E}_{\vech_{i}}\lefto\{ \prob{\mathcal{E}_{ik}^{JD}} \right\} \label{hoijd}\\
										&\ \dotleq \  \prob{ \mathcal{O}_{i1}^{JD}\left(\Gamma_{i}(\vecr)\right)} +\prob{ \mathcal{O}_{i2}^{JD}\left(\Gamma_{3}(\vecr)\right)} \\  
										& = 2\SNR^{-d_{j^{*}}^{JD}(\vecr)} \doteq \SNR^{-\min_{i=1,2,3} d_{i}^{JD}(\vecr)} \notag.
										\end{align}
Recalling that $P\lefto(E^{JD}\right)  = \max_{i=1,2} \mathbb{E}_{\vech_{i}}\lefto\{	\prob{\mathcal{E}^{JD}_{i}}\right\}$, we upper-bound $P\lefto(E^{JD}\right)$ according to    
\begin{align}
P\lefto(E^{JD}\right) & = \max_{i=1,2} \mathbb{E}_{\vech_{i}}\lefto\{	\prob{\mathcal{E}^{JD}_{i}} \right\} \\
	  & \ \dotleq \  \max_{i=1,2} \SNR^{-\min_{j=1,2,3} d_{j}^{JD}(\vecr)}\\
	  & \ \doteq \ \SNR^{-d^{JD}(\vecr)} \label{fortysizx}.
\end{align}
Since \eqref{fortysizx} gives an upper bound that matches the  lower bound in \eqref{twentyone}, the proof is complete.
\end{proof}

\subsubsection*{Discussion}
The strategy of the \ajmlong{IC} forces us to decode the message from the interfering user $\mathcal{T}_{j}$ at $\mathcal{R}_{i}$ for $i,j=1,2$ and $i\neq j$ together with the intended message from $
\mathcal{T}_{i}$ in its entirety.  We can relax this constraint and allow only part of the interfering signal $\mathcal{T}_{j}$ to be decoded at $\mathcal{R}_{i}$ for $i,j=1,2$ and $i\neq j$.  This is precisely the idea behind the Han-Kobayashi communication scheme, which we analyze in Section \ref{hankosection}.

\section{Achievable DMT of Two-Message Fixed-Power-Split Han-Kobayashi Schemes} 
\label{hankosection}
The Han-Kobayashi (HK) rate region \cite{hanko} remains the best known achievable rate region for the Gaussian IC \cite{Sason04, kramer06}.  The original HK strategy lets each transmitter split its message into two messages, allows each receiver to decode part of the interfering signal, and uses five auxiliary RVs $Q, U_{1}, U_{2}, W_{1},$ and $W_{2},$ all defined on arbitrary finite sets.  The auxiliary RV $U_{i}$ carries the private message of $ \mathcal{T}_{i}$, whereas the auxiliary RV $W_{i}$ carries the public message of $ \mathcal{T}_{i}$ destined for both receivers.  The RV $Q$ is for time-sharing.    The general HK rate region is usually prohibitively complex to describe \cite{Chong08}. 

In the following, we analyze the DMT of a two-message, fixed-power-split superposition HK scheme where $ \mathcal{T}_{i}$ transmits the $N$-dimensional ($N \geq 2$) vector $\vecx_{i} = \vecu_{i}+\vecw_{i}$ with $\vecu_{i}$ and $\vecw_{i}$ representing the private and the public message, respectively.   The power constraints for $\vecu_{i}$ and $\vecw_{i}$ are  
\begin{align}
\| \vecu_{i} \| \leq \sqrt{\frac{N}{\SNR^{1-p_{i}}}},  \ \  
\| \vecw_{i} \| \leq \sqrt{N}\lefto(1-\sqrt{\frac{1}{\SNR^{1-p_{i}}}}\right)  \notag
\end{align}
so that $\|\vecx_{i}\| \leq \|\vecu_{i}\| + \|\vecw_{i}\| = \sqrt{N}$.  Here, $0 \leq p_{i} < 1$ accounts for  the exponential order of the power allocated to the private message. The power split is assumed fixed and is independent of the channel realizations.  When both the private and the public messages are allocated maximum power, we have $\frac{\|\vecw_{i}\|^{2}}{\|\vecu_{i}\|^{2}} \doteq \SNR^{1-p_{i}}$.  We emphasize that any $p_{i} < 1$  constitutes a valid power split.  We will demonstrate later that schemes with $p_{i} < 0 $  yield zero diversity order, and, hence, do not contribute to the DMT region as the private message codebook is vanishing in size with increasing  SNR.  The case $p_{i}=-\infty$ corresponds to public messages only, and was treated in section \ref{jointML}.

We assume that $ \mathcal{T}_{i}$ transmits at rate $R_{i}=r_{i}\log \SNR$  where  the rates for the private and the public messages, respectively,  are  $S_{i}=s_{i} \log \SNR$ and $T_{i} = t_{i} \log \SNR$ with $r_{i}= s_{i}+t_{i}$, $s_{i}, t_{i}\geq 0$, and $0 \leq r_{i}\leq 1$.   The codebooks corresponding to the private and the public message parts are denoted as  $\mathcal{C}^{\vecu_{i}}(\SNR,s_{i})$ and $\mathcal{C}^{\vecw_{i}}(\SNR,t_{i})$, respectively, and satisfy $|\mathcal{C}^{\vecu_{i}}(\SNR,s_{i})|=\SNR^{Ns_{i}}$ and $|\mathcal{C}^{\vecw_{i}}(\SNR,t_{i})|=\SNR^{Nt_{i}}$. Clearly, $\mathcal{C}^{\vecx_{i}}(\SNR,r_{i}) = \mathcal{C}^{\vecu_{i}}(\SNR,s_{i})\times \mathcal{C}^{\vecw_{i}}(\SNR,t_{i}) $ with $|\mathcal{C}^{\vecx_{i}}(\SNR,r_{i})| = \SNR^{r_{i}}$.   In the following, we will need the private message multiplexing rate vector  $\vecs=[s_{1} \ s_{2}]^{T}$ and the SNR exponent vector  $\vecp= [p_{1} \ p_{2}]^{T}$ of the private messages.    

\begin{definer}
A \emph{joint ML decoder for the two-message, fixed-power-split HK scheme} at $\mathcal{R}_{j }$ ($j=1,2$)  carries out  joint ML detection on the public messages from both transmitters ($\mathcal{T}_{i}$ for $i=1,2$) and the private message from $\mathcal{T}_{j}$.  For the \emph{joint ML decoder for the two-message, fixed-power-split HK scheme} at $\mathcal{R}_{j}$, one does not declare an error if the estimate of the public message of $\mathcal{T}_{i}$ does not match the transmitted message for $i,j=1,2$ and $i\neq j$.  The error probability of this receiver at $\mathcal{R}_{j}$ is denoted by $\prob{\mathcal{E}^{HK}_{j}}$ for $j=1,2$.  The average error probability of this receiver is denoted by $P\lefto(E^{HK}_{j}\right) \triangleq \mathbb{E}_{\vech_{j}}\left\{\prob{\mathcal{E}_{j}}\right\}$ for $j=1,2$. 
\end{definer}

We employ a \emph{joint ML decoder for the two-message, fixed-power-split HK scheme} at each $\mathcal{R}_{j}$ ($j=1,2$).  The SNR exponent of  $P(E^{HK}) = \max\{P\lefto(E^{HK}_{1}\right)\! , P\lefto(E^{HK}_{2}\right)\}$ and the conditions on the superposition codes for achieving this SNR exponent are characterized next. 
 
\begin{theorem}	
\label{theoremHK}
The achievable DMT for the two-message, fixed-power-split HK scheme is given by  \begin{align}
\label{dhk}d^{HK}(\vecr) = \max_{\vecs,\vecp} d(\vecr,\vecs, \vecp) 
\end{align} with the optimization carried out subject to the constraints \begin{align} s_{i}+t_{i} &= r_{i},  \ \text{with} \ 	s_{i},t_{i} \geq 0  \notag \\
	0 \leq p_{i} & < 1, \ i =1,2   \notag
\end{align}
and
\begin{align} d(\vecr,\vecs, \vecp)  &= \min_{\substack{ k=1,2 \\ l=1,2,\ldots,6}}\left(d_{kl}(\vecr,\vecs, \vecp)\right)   \notag \\
d_{i1}(\vecr,\vecs,\vecp) &= 	\begin{cases} 
	(p_{i}-s_{i})^{+}, &   \text{if } \   p_{j} <  1-\alpha \\
	(1-\alpha-p_{j}+p_{i}-s_{i})^{+}, &  \text{if } \   p_{j} \geq  1 - \alpha
	\end{cases} \notag \\
d_{i2}(\vecr,\vecs,\vecp)  &= \begin{cases} 
	(1-r_{i}+s_{i})^{+}, & \  \text{if } \   p_{j} <  1-\alpha   \\
	(2-\alpha-p_{j}-r_{i}+s_{i})^{+}, & \ \text{if } \   p_{j} \geq  1-\alpha  
	\end{cases} \notag \\
d_{i3}(\vecr,\vecs,\vecp) &= 	\begin{cases} 
	(1-r_{i})^{+}, & \  \text{if } \   p_{j} <  1-\alpha \notag \\
	(2-\alpha-p_{j}-r_{i})^{+}, & \ \text{if } \   p_{j} \geq  1-\alpha \notag
	\end{cases} \\
d_{i4}(\vecr,\vecs,\vecp) &= 	\begin{cases} 
	(p_{i}-s_{i}-r_{j}+s_{j})^{+} \! + \! (\alpha-s_{i}-r_{j}+s_{j})^{+}, \\  \ \ \ \ \ \ \ \ \ \ \ \ \  \text{if } \   p_{j} < 1-s_{i}-r_{j}+s_{j} \notag \\
	(p_{i}-s_{i}-r_{j}+s_{j})^{+}, \\ \ \ \ \ \ \text{if } \   p_{j} \geq  1-s_{i}-r_{j}+s_{j}  \ \text{and} \ p_{j} < 1-\alpha \notag \\
	(1-\alpha-p_{j}+p_{i}-s_{i}-r_{j}+s_{j})^{+}, \\ \ \ \ \ \ \text{if } \   p_{j} \geq  1-s_{i}-r_{j}+s_{j} \ \text{and} \ p_{j} \geq 1-\alpha \notag 
	\end{cases} \\
d_{i5}(\vecr,\vecs,\vecp) &=  \begin{cases} 
	\lefto(\! 1- \! \! \sum\limits_{k=1}^{2}r_{k}+\! \sum\limits_{l=1}^{2} s_{l} \! \right)^{\! \! +}\! \! +\! \lefto(\! \alpha -\! \sum\limits_{k=1}^{2}r_{k}+\! \sum\limits_{l=1}^{2} s_{l} \! \right)^{\! \! +}\! ,  \\  \ \  \text{if } \   p_{j} <  1-\sum\limits_{k=1}^{2}r_{k}+\sum\limits_{l=1}^{2} s_{l}, \\
	\lefto(\! 1-\sum\limits_{k=1}^{2}r_{k}+\sum\limits_{l=1}^{2} s_{l} \! \right)^{\! \! +},  \\  \ \  \text{if } \   p_{j} \geq  1-\sum\limits_{k=1}^{2}r_{k}+\sum\limits_{l=1}^{2} s_{l} \ \text{and} \ p_{j} < 1-\alpha \notag \\
	\lefto(2-\alpha-p_{j}-\sum\limits_{k=1}^{2}r_{k}+\sum\limits_{l=1}^{2} s_{l}\right)^{\! \! +},  \\  \ \  \text{if } \   p_{j} \geq  1-\sum\limits_{k=1}^{2}r_{k}+\sum\limits_{l=1}^{2} s_{l} \ \text{and} \ p_{j} \geq 1-\alpha
	\end{cases} \\	
d_{i6}(\vecr,\vecs,\vecp) &=	\begin{cases} 
	(1-r_{i}-r_{j}+s_{j})^{+}+(\alpha-r_{i}-r_{j}+s_{j})^{+},  \\ \ \ \ \ \ \ \ \ \ \ \ \ \ \ \    \text{if } \   p_{j} <  1-r_{i}-r_{j}+s_{j} \notag \\
	(1-r_{i}-r_{j}+s_{j})^{+},  \\ \ \ \ \ \ \text{if } \   p_{j} \geq  1-r_{i}-r_{j}+s_{j} \ \text{and} \ p_{j} < 1-\alpha  \notag \\
	(2-\alpha-p_{j}-r_{i}-r_{j}+s_{j})^{+},  \\  \ \ \ \ \ \text{if } \   p_{j} \geq  1-r_{i}-r_{j}+s_{j} \ \text{and} \ p_{j} \geq 1-\alpha\notag
	\end{cases}
\end{align} with $i,j=1,2$ and $i\neq j$.  Define the codeword difference vectors $\Delta\vecu_{i}=\sqrt{\SNR^{1-p_{i}}}(\vecu^{\cdind{i}{u}}_{i}-\vecu^{\cdindt{i}{u}}_{i})$, $\Delta\vecw_{i}=\vecw^{\cdind{i}{w}}_{i}-\vecw^{\cdindt{i}{w}}_{i}$, and $\Delta\vecx_{i}=\vecx^{\cdind{i}{x}}_{i}-\vecx^{\cdindt{i}{x}}_{i}$ with $\vecu^{\cdind{i}{u}}_{i}, \vecu^{\cdindt{i}{u}}_{i}\in \mathcal{C}^{\vecu_{i}}(\SNR, s_{i})$, $\vecw^{\cdind{i}{w}}_{i},\vecw^{\cdindt{i}{w}}_{i}\in \mathcal{C}^{\vecw_{i}}(\SNR,t_{i})$ and $\vecx^{\cdind{i}{x}}_{i}, \vecx^{\cdindt{i}{x}}_{i} \in \mathcal{C}^{\vecx_{i}}(\SNR, r_{i})$,  for $i=1,2$. Further, define  $\Delta\matA_{ij} = [ \Delta\vecu_{i}  \ \Delta \vecw_{j}]$, $\Delta\matB_{ij} = [ \Delta\vecw_{i}  \ \Delta \vecw_{j}]$,  and $\Delta\matC_{ij} = [ \Delta\vecx_{i}  \ \Delta \vecw_{j}]$ for $i,j=1,2$ and $i\neq j$.  Denote the optimizing values of $\vecs$, $\vect$, and $\vecp$  obtained by solving (\ref{dhk}) as $\vecs^{*}, \vect^{*}$, and $\vecp^{*}$, respectively. 
We let 
\begin{align} [k^{*} \ l^{*}] = \argmin_{\substack{ k=1,2 \\ l=1,2,3,4,5,6}}\left(d_{kl}(\vecr,\vecs, \vecp)\right).\end{align} 
Further, let the functions\footnote{We note that the functions $\Upsilon_{nm}(\vecr)$ and $\Psi_{nm}(\vecs^{*})$ might not be unique.} $\Upsilon_{nm}(\vecr)=[\upsilon_{nm}^{1}(\vecr) \ \upsilon_{nm}^{2}(\vecr)]^{T}$ and $\Psi_{nm}(\vecs^{*})=[\psi_{nm}^{1}(\vecs^{*}) \ \psi_{nm}^{2}(\vecs^{*})]^{T}$ be such that
\begin{align}
	d_{k^{*}l^{*}}(\vecr,\vecs^{*},\vecp^{*}) = d_{nm}(\Upsilon_{nm}(\vecr), \Psi_{nm}(\vecs^{*}), \vecp^{*}) \notag
\end{align}
for all $n=1,2$ and $m=1,2,\ldots,6$.  If there exists a sequence (in SNR) of superposition codes satisfying 
\begin{align}
	\| \Delta\vecu_{i}\|^2 \ &\dotgeq \ \SNR^{-\psi_{i1}^{i}(\vecs^{*}) +\epsilon} \notag \\
      \| \Delta\vecw_{i}\|^2 \ &\dotgeq \ \SNR^{-\upsilon_{i2}^{i}(\vecr)+\psi_{i2}^{i}(\vecs^{*})+\epsilon}\notag\\
	\| \Delta\vecx_{i}\|^2 \ &\dotgeq \  \SNR^{-\upsilon_{i3}^{i}(\vecr)+\epsilon}\notag\\ 
	\lambda_{\min}(\Delta\matA_{ij}\lefto(\Delta\matA_{ij}\right)^{H}) \ &\dotgeq \ \SNR^{-\psi_{i4}^{i}(\vecs^{*}) -\upsilon_{j4}^{j}(\vecr)+\psi_{j4}^{j}(\vecs^{*})+\epsilon}\notag\\
	\lambda_{\min}(\Delta\matB_{ij}\lefto(\Delta\matB_{ij}\right)^{H}) \ &\dotgeq \ \SNR^{-\sum\limits_{k=1}^{2}\upsilon_{k5}^{k}(\vecr)+\sum\limits_{j=1}^{2}\psi_{j5}^{j}(\vecs^{*})+\epsilon} \notag \\
	\lambda_{\min}(\Delta\matC_{ij}\lefto(\Delta\matC_{ij}\right)^{H}) \ &\dotgeq \ \SNR^{-\upsilon_{i6}^{i}(\vecr) -\upsilon_{j6}^{j}(\vecr)+\psi_{j6}^{j}(\vecs^{*})+\epsilon}\label{cloon}
\end{align}
for every pair of codewords in each codebook for $i,j=1,2$, $i\neq j$, and for some\footnote{We note that the $\epsilon$'s in \eqref{cloon} are allowed to be different.} $\epsilon > 0$, then we have 
\begin{align}
	P(E^{HK}) \doteq \SNR^{-d_{HK}(\vecr)}.
\end{align}

\end{theorem}
\begin{proof} 
The public message is to be decoded at both receivers, whereas  the private message is to be decoded \emph{only} at the intended receiver.    As stated before and discussed in \cite{zheng_tradeoff}, there is no loss of optimality in assuming i.i.d. Gaussian inputs in obtaining an outer bound on the DMT.  Hence, we restrict ourselves to the case where all codebooks are i.i.d. Gaussian, i.e.,  
\begin{align} 
\vecu_{i}  &\sim  \mathcal{CN}(\mathbf{0}, \SNR^{p_{i}-1}\matI_{N}) \\
\vecw_{i} &\sim \mathcal{CN}(\mathbf{0}, \lefto(1-\sqrt{1/(\SNR^{1-p_{i}})}\right)^{2}\matI_{N}) \label{agagagagagazcazgv}
\end{align} with $ 0 \leq p_{i} < 1$.  Since we are interested in the high-SNR asymptotics, we can take $\lefto(1-\sqrt{\frac{1}{\SNR^{1-p_{i}}}}\right)^{2}\approx 1$ so that \eqref{agagagagagazcazgv} becomes 
\begin{align} 
\vecw_{i} &\sim \mathcal{CN}(\mathbf{0}, \matI_{N}).
\end{align}
The set of achievable rates $\{S_{i}, T_{i}, T_{j} \}$ for $i,j=1,2$, $i\neq j$ at $\mathcal{R}_{i}$,  given the channel realization $\vech_{i}$, can be characterized as
\begin{align}
&\mathcal{R}^{i}_{HK} \triangleq  \left\{S_{i}, T_{i}, T_{j}\right\} \ :  \notag \\
&\label{zzz} S_{i} \leq \log \lefto(1+ \frac{\SNR^{p_{i}}|h_{ii}|^{2}}{1+\SNR^{\alpha+p_{j}-1}|h_{ji}|^{2}}\right) \\ 
&\label{unsz2}T_{i} \leq \log \lefto(1+ \frac{\SNR |h_{ii}|^{2}}{1+\SNR^{\alpha+p_{j}-1}|h_{ji}|^{2}}\right) \\
&\label{unsz}T_{j} \leq \log \lefto(1+ \frac{\SNR^{\alpha} |h_{ji}|^{2}}{1+\SNR^{\alpha+p_{j}-1}|h_{ji}|^{2}}\right) \\
&S_{i}+T_{i} \leq \log \lefto(1+ \frac{\SNR|h_{ii}|^{2}}{1+\SNR^{\alpha+p_{j}-1}|h_{ji}|^{2}}\right) \\
&S_{i}+T_{j} \leq \log \lefto(1+ \frac{\SNR^{p_{i}}|h_{ii}|^{2}+\SNR^{\alpha}|h_{ji}|^2}{1+\SNR^{\alpha+p_{j}-1}|h_{ji}|^{2}}\right) \\
&\label{unsz3}T_{i}+T_{j} \leq \log \lefto(1+ \frac{\SNR|h_{ii}|^{2}+\SNR^{\alpha}|h_{ji}|^2}{1+\SNR^{\alpha+p_{j}-1}|h_{ji}|^{2}}\right) \\ 
&S_{i}+T_{i}+ T_{j} \leq \log \lefto(1+ \frac{\SNR|h_{ii}|^{2}+\SNR^{\alpha}|h_{ji}|^2}{1+\SNR^{\alpha+p_{j}-1}|h_{ji}|^{2}}\right) \\
&\label{kkk} S_{i},T_{i}, T_{j} \geq 0 
\end{align} for  $i,j=1,2$ and $i\neq j$. For a set $\mathcal{S}$ of quadruples $\left\{S_{1},T_{1},S_{2},T_{2}\right\}$, let $\prod(\mathcal{S})$ be the set of rate pairs $(R_{1},R_{2})$ such that $R_{1}=S_{1}+T_{1}$ and $R_{2}=S_{2}+T_{2}$. Then, the set
\begin{align} 
\mathcal{R}^{*} \triangleq \prod\left( \mathcal{R}_{HK}^{1}\bigcap \mathcal{R}_{HK}^{2}\right) \label{hkrater}
\end{align}
is an achievable rate region for the IC operating under a HK scheme with fixed power split $\vecp$. By definition, no decoding error is made at $\mathcal{R}_{i}$ if the private and the public message of $\mathcal{T}_{i}$ are decoded correctly but the public message of $\mathcal{T}_{j}$ is decoded incorrectly \cite{Chong08}.  Therefore, as  the receiver $R_{i}$ is not interested in the messages from $\mathcal{T}_{j}$, it does not make sense to declare an outage because the channel between the unintended transmitter $\mathcal{T}_{j}$ and the receiver $\mathcal{R}_{i}$ for $i,j=1,2$ $i\neq j$, is not good enough to support the transmission rate $T_{j}$. Hence, the outage event corresponding to decoding the public message of the unintended transmitter,  (\ref{unsz}), and its counterpart for $\mathcal{R}_j$ are unnecessary from the point of view of the respective receivers.   An outage event for $\mathcal{R}_{i}$ is therefore defined by  
\begin{align}
\mathcal{O}_{i}(\vecr,\vecs,\vecp) \triangleq \bigcup_{j=1}^{6} \mathcal{O}_{ij}(\vecr,\vecs,\vecp) \label{totaloutage}
\end{align}
where
\begin{align}
&\mathcal{O}_{i1}(\vecr,\vecs,\vecp) \triangleq \notag \\& \lefto\{ \vech_{i}:	   \log \lefto(1+ \frac{\SNR^{p_{i}}|h_{ii}|^{2}}{1+\SNR^{\alpha+p_{j}-1}|h_{ji}|^{2}}\right) < S_{i}  \right\} \label{outage1} \\
&\mathcal{O}_{i2}(\vecr,\vecs,\vecp) \triangleq \notag \\& \lefto\{\vech_{i}:   \log \lefto(1+ \frac{\SNR |h_{ii}|^{2}}{1+\SNR^{\alpha+p_{j}-1}|h_{ji}|^{2}}\right) < T_{i} \right\} \label{outage2} \\
&\mathcal{O}_{i3}(\vecr,\vecs,\vecp) \triangleq \notag \\& \lefto\{ \vech_{i}:  \log \lefto(1+ \frac{\SNR|h_{ii}|^{2}}{1+\SNR^{\alpha+p_{j}-1}|h_{ji}|^{2}}\right) < S_{i} +T_{i}\right\}  \label{outage3} \\
&\mathcal{O}_{i4}(\vecr,\vecs,\vecp) \triangleq \notag \\& \lefto\{ \vech_{i}:  \log \lefto(1+ \frac{\SNR^{p_{i}}|h_{ii}|^{2}+\SNR^{\alpha}|h_{ji}|^2}{1+\SNR^{\alpha+p_{j}-1}|h_{ji}|^{2}}\right)   < S_{i}+T_{j} \! \right\} \label{outage4} \\
&\mathcal{O}_{i5}(\vecr,\vecs,\vecp) \triangleq \notag \\& \lefto\{\vech_{i}: \log \lefto(1+ \frac{\SNR|h_{ii}|^{2}+\SNR^{\alpha}|h_{ji}|^2}{1+\SNR^{\alpha+p_{j}-1}|h_{ji}|^{2}}\right)   < T_{i}+T_{j} \right\}\label{outage5} \\
&\mathcal{O}_{i6}(\vecr,\vecs,\vecp) \triangleq \notag \\& \lefto\{ \vech_{i}:   \log \lefto(1+ \frac{\SNR|h_{ii}|^{2}+\SNR^{\alpha}|h_{ji}|^2}{1+\SNR^{\alpha+p_{j}-1}|h_{ji}|^{2}}\right) < S_{i}+T_{i}+ T_{j} \right\}  \label{outage6}
\end{align}
for $i,j=1,2$ and $i\neq j$.  We also define the complementary events $\bar{\mathcal{O}}_{ik}(\vecr,\vecs,\vecp)$ for $k=1,2,\ldots, 6$ as follows:
\begin{align}
&\bar{\mathcal{O}}_{i1}(\vecr,\vecs,\vecp) \triangleq \notag \\& \lefto\{ \vech_{i}:	   \log \lefto(1+ \frac{\SNR^{p_{i}}|h_{ii}|^{2}}{1+\SNR^{\alpha+p_{j}-1}|h_{ji}|^{2}}\right) \geq S_{i}  \right\} \label{noutage1} \\
&\bar{\mathcal{O}}_{i2}(\vecr,\vecs,\vecp) \triangleq \notag \\& \lefto\{\vech_{i}:   \log \lefto(1+ \frac{\SNR |h_{ii}|^{2}}{1+\SNR^{\alpha+p_{j}-1}|h_{ji}|^{2}}\right) \geq T_{i} \right\} \label{noutage2} \\
&\bar{\mathcal{O}}_{i3}(\vecr,\vecs,\vecp) \triangleq \notag \\& \lefto\{ \vech_{i}:  \log \lefto(1+ \frac{\SNR|h_{ii}|^{2}}{1+\SNR^{\alpha+p_{j}-1}|h_{ji}|^{2}}\right) \geq S_{i} +T_{i}\right\}  \label{noutage3} \\
&\bar{\mathcal{O}}_{i4}(\vecr,\vecs,\vecp) \triangleq \notag \\& \lefto\{ \vech_{i}:  \log \lefto(1+ \frac{\SNR^{p_{i}}|h_{ii}|^{2}+\SNR^{\alpha}|h_{ji}|^2}{1+\SNR^{\alpha+p_{j}-1}|h_{ji}|^{2}}\right)   \geq S_{i}+T_{j} \! \right\} \label{noutage4} \\
&\bar{\mathcal{O}}_{i5}(\vecr,\vecs,\vecp) \triangleq \notag \\& \lefto\{\vech_{i}: \log \lefto(1+ \frac{\SNR|h_{ii}|^{2}+\SNR^{\alpha}|h_{ji}|^2}{1+\SNR^{\alpha+p_{j}-1}|h_{ji}|^{2}}\right)  \geq T_{i}+T_{j} \right\}\label{noutage5} \\
&\bar{\mathcal{O}}_{i6}(\vecr,\vecs,\vecp) \triangleq \notag \\& \lefto\{ \vech_{i}:   \log \lefto(1+ \frac{\SNR|h_{ii}|^{2}+\SNR^{\alpha}|h_{ji}|^2}{1+\SNR^{\alpha+p_{j}-1}|h_{ji}|^{2}}\right) \geq S_{i}+T_{i}+ T_{j} \right\}  \label{noutage6}
\end{align}
for $i,j=1,2$ and $i\neq j$.  It is shown in \cite{Akuiyibo08} that $\prob{\mathcal{O}_{ik}(\vecr,\vecs,\vecp)} \doteq \SNR^{-d_{ik}(\vecr,\vecs,\vecp)}$, $i=1,2$, $k=1,2, \ldots, 6$, where 
\begin{align} 
d_{i1}(\vecr,\vecs,\vecp) &= 	\begin{cases} 
\label{cons1}	(p_{i}-s_{i})^{+},  \\ \ \ \ \ \   \text{if } \   p_{j}  <   1-\alpha \\
	(1-\alpha-p_{j}+p_{i}-s_{i})^{+}, \\ \ \ \ \ \  \text{if } \   p_{j} \geq  1 - \alpha
	\end{cases}  \\
d_{i2}(\vecr,\vecs,\vecp)  &= \begin{cases} 
\label{cons2}	(1-r_{i}+s_{i})^{+}, \\ \ \ \ \ \ \  \text{if } \   p_{j} <  1-\alpha   \\
	(2-\alpha-p_{j}-r_{i}+s_{i})^{+}, \\ \ \ \ \ \ \ \text{if } \   p_{j} \geq  1-\alpha  
	\end{cases}  \\
d_{i3}(\vecr,\vecs,\vecp) &= 	\begin{cases} 
\label{cons3}	(1-r_{i})^{+}, \\ \ \ \ \ \ \  \text{if } \   p_{j}  < 1-\alpha \\
	(2-\alpha-p_{j}-r_{i})^{+}, \\ \ \ \ \ \ \ \text{if } \   p_{j} \geq  1-\alpha 
	\end{cases} \\
d_{i4}(\vecr,\vecs,\vecp) &= 	\begin{cases} 
\label{cons4}	(p_{i}-s_{i}-r_{j}+s_{j})^{+} \! + \! (\alpha-s_{i}-r_{j}+s_{j})^{+}, \\  \ \ \ \ \ \ \ \ \ \ \ \ \  \text{if } \   p_{j} < 1-s_{i}-r_{j}+s_{j}  \\
	(p_{i}-s_{i}-r_{j}+s_{j})^{+}, \\ \ \ \ \text{if } \   p_{j} \geq  1-s_{i}-r_{j}+s_{j}  \ \text{and}  \ p_{j} < 1-\alpha \\
	(1-\alpha-p_{j}+p_{i}-s_{i}-r_{j}+s_{j})^{+}, \\ \ \ \  \text{if } \   p_{j} \geq  1-s_{i}-r_{j}+s_{j} \ \text{and} \ p_{j} \geq 1-\alpha  
	\end{cases} \\
d_{i5}(\vecr,\vecs,\vecp) &=  \begin{cases} 
\label{cons5}	\lefto(\! 1- \! \! \sum\limits_{k=1}^{2}r_{k}+\! \sum\limits_{l=1}^{2} s_{l} \! \right)^{\! \! +}\! \! +\! \lefto(\! \alpha -\! \sum\limits_{k=1}^{2}r_{k}+\! \sum\limits_{l=1}^{2} s_{l} \! \right)^{\! \! +}\! ,  \\  \ \  \text{if } \   p_{j} <  1-\sum\limits_{k=1}^{2}r_{k}+\sum\limits_{l=1}^{2} s_{l}, \\
	\lefto(\! 1-\sum\limits_{k=1}^{2}r_{k}+\sum\limits_{l=1}^{2} s_{l} \! \right)^{\! \! +},  \\  \ \  \text{if } \   p_{j} \geq  1-\sum\limits_{k=1}^{2}r_{k}+\sum\limits_{l=1}^{2} s_{l} \ \text{and} \ p_{j} < 1-\alpha \\
	\lefto(2-\alpha-p_{j}-\sum\limits_{k=1}^{2}r_{k}+\sum\limits_{l=1}^{2} s_{l}\right)^{\! \! +},  \\  \ \  \text{if } \   p_{j} \geq  1-\sum\limits_{k=1}^{2}r_{k}+\sum\limits_{l=1}^{2} s_{l} \ \text{and} \ p_{j} \geq 1-\alpha
	\end{cases} \\	
d_{i6}(\vecr,\vecs,\vecp) &=	\begin{cases} 
\label{cons6}	(1-r_{i}-r_{j}+s_{j})^{+}+(\alpha-r_{i}-r_{j}+s_{j})^{+},  \\ \ \ \ \ \ \ \ \ \ \ \ \ \ \ \    \text{if } \   p_{j} <  1-r_{i}-r_{j}+s_{j}  \\
	(1-r_{i}-r_{j}+s_{j})^{+},  \\ \ \ \ \ \ \text{if } \   p_{j} \geq  1-r_{i}-r_{j}+s_{j} \ \text{and} \ p_{j} < 1-\alpha  \\
	(2-\alpha-p_{j}-r_{i}-r_{j}+s_{j})^{+},  \\  \ \ \ \ \ \text{if } \   p_{j} \geq  1-r_{i}-r_{j}+s_{j} \ \text{and} \ p_{j} \geq 1-\alpha
	\end{cases}
\end{align} with $i,j=1,2$ and $i\neq j$. We define the total outage probability of the IC as the maximum of the probabilities of outage for the two receivers, that is,
\begin{align}
\prob{\mathcal{O}(\vecr,\vecs,\vecp)} \triangleq \max\left(\prob{\mathcal{O}_{1}(\vecr,\vecs,\vecp)}\! ,\prob{\mathcal{O}_{2}(\vecr,\vecs,\vecp)}\right).
\end{align} 
We note that this definition is compatible with our previous definitions. For a given rate tuple $\vecr$, we would like to minimize this probability over all choices of $\vecs$ and $\vecp$, i.e., 
\begin{align}
\prob{\mathcal{O}^{HK}(\vecr)} \triangleq \min_{\vecs, \vecp} \prob{\mathcal{O}(\vecr,\vecs,\vecp)} \label{gacirt}
\end{align} subject to
\begin{align}
&r_{i} = s_{i}+t_{i}\\
&s_{i},t_{i} \geq 0 \\
&0\leq p_{i} < 1, \ \ \text{for} \ i=1,2. 
\end{align}
We will next show that $\prob{\mathcal{O}^{HK}(\vecr)} $ obeys the following exponential behavior in $\SNR$
\begin{align}
\label{DMTcurve}
\prob{\mathcal{O}^{HK}(\vecr)}  \doteq \SNR^{-d^{HK}(\vecr)}
\end{align}
where 
\begin{align}
d^{HK}(\vecr) = \max_{\vecs,\vecp} \min\left(d_{1}(\vecr,\vecs,\vecp), d_{2}(\vecr,\vecs,\vecp)\right) 
\end{align} subject to \begin{align} 
&	s_{i}+t_{i} = r_{i} \notag \\  
&	0 \leq s_{i} \leq r_{i}\notag \\
&	0 \leq t_{i} \leq r_{i}\notag \\
&	0 \leq p_{i} < 1, \ \ \ \text{for} \ i=1,2, 
\end{align}
and where the $d_{i}(\vecr,\vecs,\vecp)$ are given by
\begin{align}
d_{1}(\vecr,\vecs, \vecp) = \min_{i=1,2, \ldots,6} d_{1i}(\vecr,\vecs,\vecp) \\
d_{2}(\vecr,\vecs, \vecp) = \min_{i=1,2, \ldots,6} d_{2i}(\vecr,\vecs,\vecp).
\end{align}
To see this, we note that $\prob{\mathcal{O}^{HK}(\vecr)}$ can be bounded as follows
\begin{align}
\notag &\min_{\vecs, \vecp} \max\left(\prob{\mathcal{O}_{1k}(\vecr,\vecs,\vecp)},\prob{\mathcal{O}_{2l}(\vecr,\vecs,\vecp)}\right)\leq \prob{\mathcal{O}^{HK}(\vecr)} \\  &\leq \min_{\vecs,\vecp} \max\lefto(\sum_{i=1}^{6}\prob{\mathcal{O}_{1i}(\vecr,\vecs,\vecp)},\sum_{j=1}^{6}\prob{\mathcal{O}_{2j}(\vecr,\vecs,\vecp)}\right) \label{glodg}
\end{align} where the inequality holds for all $k=1,2,\ldots, 6$ and $l=1,2,\ldots, 6$. In the high SNR limit the RHS of (\ref{glodg}) is dominated by the SNR exponent given by 
\begin{align}
	\max_{\vecs,\vecp} \min\lefto( \min_{i=1,2, \ldots,6}d_{1i}(\vecr,\vecs,\vecp),  \min_{j=1,2, \ldots,6} d_{2j}(\vecr,\vecs,\vecp)\right). 
\end{align}
The upper and lower bounds on $\prob{\mathcal{O}^{HK}(\vecr)}$ can be made to have the same SNR exponent upon selection of the appropriate values for $k$ and $l$ in the left-hand-side (LHS) of (\ref{glodg}).  We now arrived at a lower bound on the error probability of the \emph{joint ML decoder for the two-message, fixed-power-split HK scheme}. 

Following \cite{tse_mu}, we decompose the error probability of the \emph{joint ML decoder for the two-message, fixed-power-split HK scheme}  at $\mathcal{R}_{i}$ into seven disjoint error events.  As noted earlier, one of these events is irrelevant for the IC.  Denoting the decisions on the private and public message of $\mathcal{T}_{i}$ and the public message of $\mathcal{T}_{j}$ at $\mathcal{R}_{i}$ by $\vecu^{\cdindt{i}{u}}_{i}, \vecw^{\cdindt{i}{w}}_{i}$, and $\vecw^{\cdindt{j}{w}}_{j}$, respectively,  we end up with the following six error events when the transmitted codewords are $\vecu^{\cdind{i}{u}}_{i}, \vecw^{\cdind{i}{w}}_{i}$, and $\vecw^{\cdind{j}{w}}_{j}$ for $i,j=1,2$ and $i \neq j$: 
\begin{align}
\mathcal{E}^{HK}_{i1} &\triangleq \left\{ \cdindt{i}{u} \neq \cdind{i}{u}, \  \cdindt{i}{w} = \cdind{i}{w}, \  \cdindt{j}{w} = \cdind{j}{w}\right\}\\ 
\mathcal{E}^{HK}_{i2} &\triangleq \left\{  \cdindt{i}{u}= \cdind{i}{u}, \  \cdindt{i}{w} \neq \cdind{i}{w}, \  \cdindt{j}{w} = \cdind{j}{w}\right\}\\
\mathcal{E}^{HK}_{i3} &\triangleq \left\{  \cdindt{i}{u} \neq \cdind{i}{u}, \  \cdindt{i}{w} \neq\cdind{i}{w}, \  \cdindt{j}{w}  =\cdind{j}{w}\right\}\\
\mathcal{E}^{HK}_{i4} &\triangleq \left\{  \cdindt{i}{u}\neq \cdind{i}{u}, \  \cdindt{i}{w} = \cdind{i}{w}, \  \cdindt{j}{w}  \neq\cdind{j}{w}\right\}\\ 
\mathcal{E}^{HK}_{i5} &\triangleq \left\{ \cdindt{i}{u} =\cdind{i}{u}, \  \cdindt{i}{w} \neq\cdind{i}{w}, \  \cdindt{j}{w}  \neq \cdind{j}{w}\right\}\\
\mathcal{E}^{HK}_{i6} &\triangleq \left\{ \cdindt{i}{u} \neq \cdind{i}{u}, \ \cdindt{i}{w} \neq \cdind{i}{w}, \  \cdindt{j}{w}  \neq \cdind{j}{w}\right\}.
\end{align}
The total error event at $\mathcal{R}_{i}$ is simply the union of the above events, i.e.,
\begin{align}
\mathcal{E}^{HK}_{i} \triangleq \bigcup_{k=1}^{6} \mathcal{E}^{HK}_{ik}.
\end{align}

We let 
\begin{align} [k^{*} \ l^{*}] = \argmin_{\substack{ k=1,2 \\ l=1,2,3,4,5,6}}\left(d_{kl}(\vecr,\vecs, \vecp)\right).\end{align} 
Further, let the functions\footnote{We note that the functions $\Upsilon_{nm}(\vecr)$ and $\Psi_{nm}(\vecs^{*})$ might not be unique.} $\Upsilon_{nm}(\vecr)=[\upsilon_{nm}^{1}(\vecr) \ \upsilon_{nm}^{2}(\vecr)]^{T}$ and $\Psi_{nm}(\vecs^{*})=[\psi_{nm}^{1}(\vecs^{*}) \ \psi_{nm}^{2}(\vecs^{*})]^{T}$ be such that
\begin{align}
	d_{k^{*}l^{*}}(\vecr,\vecs^{*},\vecp^{*}) = d_{nm}(\Upsilon_{nm}(\vecr), \Psi_{nm}(\vecs^{*}), \vecp^{*}) \notag
\end{align}
for all $n=1,2$ and $m=1,2,\ldots,6$. 

Next, we derive an upper bound on $\mathcal{E}^{HK}_{i}$ and show that the SNR exponent of this bound matches the SNR exponent of the outage probability $\prob{\mathcal{O}^{HK}(\vecr)}$. We start by deriving an upper bound on  $\prob{\mathcal{E}^{HK}_{ik}}$ according to 
\begin{align}
\prob{\mathcal{E}^{HK}_{ik}}  &= \prob{\mathcal{E}^{HK}_{ik}, \mathcal{O}_{ik} \lefto(\Upsilon_{ik}(\vecr),\Psi_{ik}(\vecs^{*}), \vecp^{*}\right)} + \notag \\  \ \ \ \ \ \ \ \ & \ \ \ \  \ \ \ \ \ \ \ \ \  \prob{\mathcal{E}^{HK}_{ik}, \bar{\mathcal{O}}_{ik} \lefto(\Upsilon_{ik}(\vecr),\Psi_{ik}(\vecs^{*}),\vecp^{*}\right)}  \\
&\leq   \prob{ \mathcal{O}_{ik} \lefto(\Upsilon_{ik}(\vecr),\Psi_{ik}(\vecs^{*}), \vecp^{*}\right)}+ \\ \  \ \ \ \ \ & \ \ \ \ \ \  \ \ \ \ \ \ \   \prob{\mathcal{E}^{HK}_{ik}| \bar{\mathcal{O}}_{ik} \lefto(\Upsilon_{ik}(\vecr),\Psi_{ik}(\vecs^{*}), \vecp^{*}\right)},
\end{align}
for $i=1,2$ and $k=1,2, \ldots, 6$. Next, we derive an upper bound on $\prob{\mathcal{E}^{HK}_{ik}| \bar{\mathcal{O}}_{ik} \lefto(\Upsilon_{ik}(\vecr),\Psi_{ik}(\vecs^{*}), \vecp^{*}\right)}$ using the union bound and the PEP.  For the event $\mathcal{E}^{HK}_{i1}$, the receiver can cancel the contribution of $\vecw_{i}$ and  $\vecw_{j}$ out as they have been decoded correctly.  The resulting equivalent signal model is then
\begin{align}
\vecy = \sqrt{\SNR} h_{ii}\vecu_{i} + \sqrt{\SNR^{\alpha}}h_{ji}\vecu_{j}+\vecz.
\end{align}
Treating $\vecu_{j}$ as noise with $\vecu_{j} \sim \mathcal{CN}(\mathbf{0}, \SNR^{-(1-p_{j})}\matI_{N})$ results in an upper bound on the error probability as the worst noise under a covariance constraint is Gaussian \cite{Hassibi03}.   The equivalent noise $\vecn= \vecz+\sqrt{\SNR^{\alpha}}h_{ji}\vecu_{j}$ is therefore Gaussian with $\vecn \sim \mathcal{CN}(\mathbf{0}, (1+\SNR^{-(1-p_{j})+\alpha}|h_{ji}|^{2})\matI_{N})$.  Recall that we assumed that $\mathcal{R}_{j}$ knows $h_{ji}$ perfectly.  We are now in a position to upper-bound the PEP according to 
\begin{align}
&\mathbb{E}_{\vech_{i}}\lefto\{\prob{\vecu_{i}\rightarrow \tilde{\vecu}_{i}}\right\} \notag \\  & \ \ \ \ \ \leq \mathbb{E}_{\vech_{i}}\lefto\{\exp\left[-\frac{\|h_{ii} (\vecu_{i}-\tilde{\vecu}_{i})\|^{2}\SNR}{4(1+\SNR^{-(1-p_{j})+\alpha}|h_{ji}|^{2})}\right]\right\}. \notag
\end{align}
Since $\Delta\vecu_{i}={\sqrt{\SNR^{1-p_{i}}}}({\vecu}_{i}-\tilde{\vecu}_{i})$, we get
\begin{align}
&\mathbb{E}_{\vech_{i}}\lefto\{\prob{\vecu_{i}\rightarrow \tilde{\vecu}_{i}}\right\} \notag \\ &  \ \ \ \ \ \leq \mathbb{E}_{\vech_{i}}\lefto\{\exp\left[-\frac{\|h_{ii} (\Delta\vecu_{i})\|^{2}\SNR^{p_{i}}}{4(1+\SNR^{-(1-p_{j})+\alpha}|h_{ji}|^{2})}\right]\right\}. \label{gotinv}
\end{align}
Next, we use the fact that $\bar{\mathcal{O}}_{i1} \lefto(\Upsilon_{i1}(\vecr),\Psi_{i1}(\vecs^{*}), \vecp^{*}\right)$ entails $\frac{\SNR^{p_{i}}|h_{ii}|^{2}}{1+\SNR^{-(1-p_{j})+\alpha}|h_{ji}|^{2}} \geq \SNR^{\psi_{i1}^{i}(\vecs^{*})}$ where $i,j =1,2$ and $i\neq j$ and apply the union bound to upper-bound $\prob{\mathcal{E}^{HK}_{i1}| \bar{\mathcal{O}}_{i1} \lefto(\Upsilon_{i1}(\vecr),\Psi_{i1}(\vecs^{*}), \vecp^{*}\right)}$ according to 
\begin{align}
&\mathbb{E}_{\vech_{i}}\lefto\{\prob{\mathcal{E}_{i1}| \bar{\mathcal{O}}_{i1} \lefto(\Upsilon_{i1}(\vecr),\Psi_{i1}(\vecs^{*}), \vecp^{*}\right)}\right\} \leq \notag \\  & \ \ \ \ \ \ \ \ \ \SNR^{Ns_{i}}\exp\left[-\frac{\SNR^{\psi_{i1}^{i}(\vecs^{*})}\|\Delta\vecu_{i}\|^{2}}{4}\right] \label{eqsfrfegs}.
\end{align}
Since $\|\Delta\vecu_{i}\|^{2} \ \dotgeq \ \SNR^{-\psi_{i1}^{i}(\vecs^{*})+\epsilon}$, with $\epsilon > 0$, by assumption, we further have
\begin{align}
&\mathbb{E}_{\vech_{i}}\lefto\{ \prob{\mathcal{E}^{HK}_{i1}}\right\} \notag  \\  & \ \dotleq \  \prob{ \mathcal{O}_{i1} \lefto(\Upsilon_{i1}(\vecr),\Psi_{i1}(\vecs^{*}), \vecp^{*}\right)}+  \SNR^{Ns_{i}}\exp\left[-\SNR^{\epsilon}\right]  \label{ejfjjfjfalhfsak} \\
							& \ \dotleq \  \prob{ \mathcal{O}_{i1} \lefto(\Upsilon_{i1}(\vecr),\Psi_{i1}(\vecs^{*}), \vecp^{*}\right)}.
\end{align}

For the event $\mathcal{E}^{HK}_{i2}$, the receiver can cancel the contributions of the correctly decoded messages $\vecu_{i}$ and $\vecw_{j}$ out.  Following steps similar to those leading to \eqref{gotinv}, we obtain 
\begin{align}
&\mathbb{E}_{\vech_{i}}\lefto\{\prob{\vecw_{i}\rightarrow \tilde{\vecw}_{i}}\right\} \notag \\ & \ \ \ \ \ \leq \mathbb{E}_{\vech_{i}}\lefto\{\exp\left[-\frac{\|h_{ii}\Delta\vecw_{i}\|^{2}\SNR}{4(1+\SNR^{-(1-p_{j})+\alpha}|h_{ji}|^{2})}\right]\right\} \notag.
\end{align}
 Next, an application of  the union bound to $\prob{\mathcal{E}^{HK}_{i2}| \bar{\mathcal{O}}_{i2} \lefto(\Upsilon_{i2}(\vecr),\Psi_{i2}(\vecs^{*}), \vecp^{*}\right)}$ yields
\begin{align}
&\mathbb{E}_{\vech_{i}}\lefto\{\prob{\mathcal{E}^{HK}_{i2}| \bar{\mathcal{O}}_{i2} \lefto(\Upsilon_{i2}(\vecr),\Psi_{i2}(\vecs^{*}), \vecp^{*}\right)}\right\}  \leq \\ & \ \ \ \ \ \  \SNR^{Nt_{i}}\exp\left[-\frac{\SNR^{\upsilon_{i2}^{i}(\vecr)-\psi_{i2}^{i}(\vecs^{*})}\|\Delta\vecw_{i}\|^{2}}{4}\right] \notag 
\end{align}
as the event $\bar{\mathcal{O}}_{i2}\lefto(\Upsilon_{i2}(\vecr),\Psi_{i2}(\vecs^{*}), \vecp^{*}\right)$ entails \begin{align} \frac{\SNR|h_{ii}|^{2}}{1+\SNR^{\alpha+p_{j}-1}|h_{ji}|^{2}} \geq \SNR^{\upsilon_{i2}^{i}(\vecr)-\psi_{i2}^{i}(\vecs^{*})}.\end{align}  Since $\|\Delta\vecw_{i}\|^{2} \ \dotgeq \ \SNR^{-\upsilon_{i2}^{i}(\vecr)+\psi_{i2}^{i}(\vecs^{*})+\epsilon}$, with $\epsilon > 0$, by assumption, we further have
\begin{align}
&\mathbb{E}_{\vech_{i}}\lefto\{\prob{\mathcal{E}^{HK}_{i2}}\right\}  \notag \\
& \ \dotleq \ \prob{ \mathcal{O}_{i2}\lefto(\Upsilon_{i2}(\vecr),\Psi_{i2}(\vecs^{*}), \vecp^{*}\right)}+  \SNR^{Nt_{i}}\exp\left[-\SNR^{\epsilon}\right] \\
& \ \dotleq \  \prob{ \mathcal{O}_{i2}\lefto(\Upsilon_{i2}(\vecr),\Psi_{i2}(\vecs^{*}), \vecp^{*}\right)}.
\end{align}

For the event $\mathcal{E}^{HK}_{i3}$, the receiver can cancel the contribution of the correctly decoded message $\vecw_{j}$ out. We define $\vecx^{\cdindt{i}{x}}_{i}=\vecu^{\cdindt{i}{u}}_{i}+\vecw^{\cdindt{i}{w}}_{i}$, and recall that  $\vecx^{\cdind{i}{x}}_{i} =\vecu^{\cdind{i}{u}}_{i}+\vecw^{\cdind{i}{w}}_{i}$. The PEP of deciding in favor of $\vecx^{\cdindt{i}{x}}_{i}$ when $\vecx^{\cdind{i}{x}}_{i}$ was actually transmitted can be upper-bounded as  
\begin{align}
&\mathbb{E}_{\vech_{i}}\lefto\{\prob{\vecx_{i}\rightarrow \tilde{\vecx}_{i}} \right\} \notag \\ & \ \ \ \ \  \leq \mathbb{E}_{\vech_{i}}\lefto\{\exp\left[-\frac{\|h_{ii}\Delta\vecx_{i}\|^{2}\SNR}{4(1+\SNR^{-(1-p_{j})+\alpha}|h_{ji}|^{2})}\right]\right\} \notag
\end{align}
where $\Delta\vecx_{i}= \vecx^{\cdind{i}{x}}_{i}-\vecx^{\cdindt{i}{x}}_{i}$ (as defined before).   Next, applying the union bound, we get 
\begin{align}
& \mathbb{E}_{\vech_{i}}\lefto\{\prob{\mathcal{E}^{HK}_{i3}| \bar{\mathcal{O}}_{i3}\lefto(\Upsilon_{i3}(\vecr),\Psi_{i3}(\vecs^{*}), \vecp^{*}\right)} \right\} \leq \notag \\ & \ \ \ \ \ \  \ \ \ \ \  \SNR^{Nr_{i}}\exp\left[-\frac{\SNR^{\upsilon_{i3}^{i}(\vecr)}\|\Delta\vecx_{i}\|^{2}}{4}\right] \notag
\end{align}
since the event $\bar{\mathcal{O}}_{i3}\lefto(\Upsilon_{i3}(\vecr),\Psi_{i3}(\vecs^{*}),\vecp^{*}\right)$ entails \begin{align} \frac{\SNR|h_{ii}|^{2}}{1+\SNR^{\alpha+p_{j}-1}|h_{ji}|^{2}} \geq \SNR^{\upsilon_{i3}^{i}(\vecr)}.\end{align}  As $\|\Delta\vecx_{i}\|^{2} \ \dotgeq \ \SNR^{-\upsilon_{i3}^{i}(\vecr)+\epsilon}$, for $\epsilon > 0$, by assumption, we further have
\begin{align}
& \mathbb{E}_{\vech_{i}}\lefto\{\prob{\mathcal{E}^{HK}_{i3}}\right\} \notag \\ \ &\dotleq \ \prob{ \mathcal{O}_{i3}\lefto(\Upsilon_{i3}(\vecr),\Psi_{i3}(\vecs^{*}), \vecp^{*}\right)}+  \SNR^{Nr_{i}}\exp\left[-\SNR^{\epsilon}\right] \\
						\	& \dotleq \ \prob{ \mathcal{O}_{i3}\lefto(\Upsilon_{i3}(\vecr),\Psi_{i3}(\vecs^{*}), \vecp^{*}\right)}.
\end{align}

For the event $\mathcal{E}^{HK}_{i4}$, the receiver can cancel out the contribution of the correctly decoded message $\vecw_{i}$.  Denoting ${\matA}_{ij} = [\sqrt{\SNR^{1-p_{i}}} {\vecu}_{i}^{\cdind{i}{u}} \  \ {\vecw}_{j}^{\cdind{j}{w}}]$, $\tilde{\matA}_{ij} = [\sqrt{\SNR^{1-p_{i}}}\vecu_{i}^{\cdindt{i}{u}} \ \ \vecw^{\cdindt{j}{w}}_{j}]$, $\tilde{\vech} = [\sqrt{\SNR^{p_{i}}}h_{ii} \ \ \sqrt{\SNR^{\alpha}}h_{ji}]^{T}$, and recalling that $\Delta\matA_{ij} = \matA_{ij}-\tilde{\matA}_{ij}$,  the PEP corresponding to deciding in favor of $\tilde{\matA}_{ij}$ when $\matA_{ij}$ was actually transmitted is upper-bounded according to
\begin{align}
&\mathbb{E}_{\vech_{i}}\lefto\{	\prob{\matA_{ij} \rightarrow \tilde{\matA}_{ij}}\right\} \notag \\
     & \ \ \ \leq \mathbb{E}_{\vech_{i}}\lefto\{\exp\left[-\frac{\|\Delta\matA_{ij}\tilde{\vech}\|^{2}}{4(1+\SNR^{\alpha-(1-p_{j})}|h_{ji}|^{2})}\right]\right\} \notag \\
	& \ \ \ \leq  \mathbb{E}_{\vech_{i}} \lefto\{\exp\left[-\lambda_{\min} \frac{\SNR^{p_{i}}|h_{ii}|^{2}+\SNR^{\alpha}|h_{ji}|^{2}}{4(1+\SNR^{\alpha-(1-p_{j})}|h_{ji}|^{2})}\right]\right\} \notag \\
	&\ \ \  \leq  \exp\left[-\lambda_{\min} \SNR^{\psi_{i4}^{i}(\vecs^{*}) +\upsilon_{j4}^{j}(\vecr)-\psi_{j4}^{j}(\vecs^{*})}\right] \notag 
\end{align}
where $\lambda_{\min}$ is the smallest nonzero eigenvalue of $\Delta\matA_{ij}(\Delta\matA_{ij})^{H}$. As \begin{align} \lambda_{\min} \ \dotgeq \ \SNR^{-\psi_{i4}^{i}(\vecs^{*}) -\upsilon_{j4}^{j}(\vecr)+\psi_{j4}^{j}(\vecs^{*})+\epsilon} \end{align} with some $\epsilon > 0$, by assumption, we have
\begin{align}
& \mathbb{E}_{\vech_{i}}\lefto\{\prob{\mathcal{E}^{HK}_{i4}}\right\} \notag \\
					    \ & \dotleq \ \prob{ \mathcal{O}_{i4}\lefto(\Upsilon_{i4}(\vecr),\Psi_{i4}(\vecs^{*}), \vecp^{*}\right)}+  \SNR^{N (s_{i}+t_{j})}\exp\left[-\SNR^{\epsilon}\right] \notag \\
						\	& \dotleq \ \prob{ \mathcal{O}_{i4}\lefto(\Upsilon_{i4}(\vecr),\Psi_{i4}(\vecs^{*}), \vecp^{*}\right)} \notag .
\end{align}

For the event $\mathcal{E}^{HK}_{i5}$, the receiver cancels out the contributions of the correctly decoded $\vecu_{i}$.  Denoting $\matB_{ij} = [{\vecw}^{\cdind{i}{w}}_{i} \ \   \vecw^{\cdind{j}{w}}_{j}]$, $\tilde{\matB}_{ij} = [\vecw^{\cdindt{i}{w}}_{i} \ \ \vecw^{\cdindt{j}{w}}_{j}]$,  $\tilde{\vech} = [\sqrt{\SNR}h_{ii} \ \ \sqrt{\SNR^{\alpha}} h_{ji}]^{T}$, and recalling that $\Delta\matB_{ij} = \matB_{ij}-\tilde{\matB}_{ij}$, we have
\begin{align}
&\mathbb{E}_{\vech_{i}}\lefto\{	\prob{\matB_{ij}\rightarrow \tilde{\matB}_{ij}}\right\} \notag \\ 
& \ \ \ \leq \mathbb{E}_{\vech_{i}}\lefto\{\exp\left[-\frac{\|\Delta\matB_{ij}\tilde{\vech}\|^{2}}{4(1+\SNR^{-(1-p_{j})+\alpha}|h_{ji}|^{2})}\right]\right\} \notag \\
	& \ \ \ \leq  \mathbb{E}_{\vech_{i}} \lefto\{\exp\left[-\lambda_{\min} \frac{\SNR|h_{ii}|^{2}+\SNR^{\alpha}|h_{ji}|^{2}}{4(1+\SNR^{-(1-p_{j})+\alpha}|h_{ji}|^{2})}\right]\right\} \notag \\
	& \ \ \ \leq  \exp\left[-\lambda_{\min} \SNR^{\sum\limits_{k=1}^{2}\upsilon_{k5}^{k}(\vecr)-\sum\limits_{j=1}^{2}\psi_{j5}^{j}(\vecs^{*})}\right] \notag
\end{align}
where $\lambda_{\min}$ is the smallest nonzero eigenvalue of $\Delta\matB_{ij}(\Delta\matB_{ij})^{H}$. As \begin{align}\lambda_{\min} \ \dotgeq \ \SNR^{-\sum\limits_{k=1}^{2}\upsilon_{k5}^{k}(\vecr)+\sum\limits_{j=1}^{2}\psi_{j5}^{j}(\vecs^{*})+\epsilon} \end{align} with some $\epsilon > 0$, by assumption, we have
\begin{align}
& \mathbb{E}_{\vech_{i}}\lefto\{\prob{\mathcal{E}^{HK}_{i5}} \right\} \notag \\
\ &\dotleq \ \prob{ \mathcal{O}_{i5}\lefto(\Upsilon_{i5}(\vecr),\Psi_{i5}(\vecs^{*}), \vecp^{*}\right)}+  \SNR^{N (t_{1}+t_{2})}\exp\left[-\SNR^{\epsilon}\right] \notag \\ \notag
						\	& \dotleq \ \prob{ \mathcal{O}_{i5}\lefto(\Upsilon_{i5}(\vecr),\Psi_{i5}(\vecs^{*}), \vecp^{*}\right)} \notag.
\end{align}

Finally, for the event $\mathcal{E}^{HK}_{i6}$, all codewords are in error, so that there is nothing to cancel out.  Denoting $\matC_{ij} = [{\vecx}_{i}^{\cdind{i}{x}} \ \ {\vecw}_{j}^{\cdind{j}{w}}]$, $\widetilde{\matC}_{ij} = [\vecx^{\cdindt{i}{x}}_{i} \ \ \vecw^{\cdindt{j}{w}}_{j}]$, $\tilde{\vech} = [\sqrt{\SNR}h_{ii} \ \sqrt{\SNR^{\alpha}}h_{ji}]^T $, and recalling that $\Delta\matC_{ij} = \matC_{ij}-\tilde{\matC}_{ij}$, we obtain
\begin{align}
&\mathbb{E}_{\vech_{i}}\lefto\{	\prob{\matC_{ij}\rightarrow \widetilde{\matC}_{ij}}\right\} \notag \\ 
& \ \ \ \leq \mathbb{E}_{\vech_{i}}\lefto\{\exp\left[-\frac{\|\Delta\matC_{ij}\tilde{\vech}\|^{2}}{4(1+\SNR^{-(1-p_{j})+\alpha}|h_{ji}|^{2})}\right]\right\}  \notag \\
	& \ \ \ \leq  \mathbb{E}_{\vech_{i}} \lefto\{\exp\left[-\lambda_{\min} \frac{\SNR|h_{ii}|^{2}+\SNR^{\alpha}|h_{ji}|^{2}}{4(1+\SNR^{-(1-p_{j})+\alpha}|h_{ji}|^{2})}\right]\right\} \notag \\
	& \ \ \ \leq  \exp\left[-\lambda_{\min} \SNR^{\upsilon_{i6}^{i}(\vecr) +\upsilon_{j6}^{j}(\vecr)-\psi_{j6}^{j}(\vecs^{*})}\right] \notag
\end{align}
where $\lambda_{\min}$ is the smallest nonzero eigenvalue of $\Delta\matC_{ij}(\Delta\matC_{ij})^{H}$.  As \begin{align}\lambda_{\min} \ \dotgeq \ \SNR^{-\upsilon_{i6}^{i}(\vecr) -\upsilon_{j6}^{j}(\vecr)+\psi_{j6}^{j}(\vecs^{*})+\epsilon}\end{align} with some $\epsilon> 0$, by assumption, we have
\begin{align}
& \mathbb{E}_{\vech_{i}}\lefto\{\prob{\mathcal{E}^{HK}_{i6}}\right\} \notag \\
 \ &\dotleq \ \prob{ \mathcal{O}_{i6}\lefto(\Upsilon_{i6}(\vecr),\Psi_{i6}(\vecs^{*}), \vecp^{*}\right)}+  \SNR^{N (r_{i}+t_{j})}\exp\left[-\SNR^{\epsilon}\right]  \notag \\
						\	& \dotleq \ \prob{ \mathcal{O}_{i6}\lefto(\Upsilon_{i6}(\vecr),\Psi_{i6}(\vecs^{*}), \vecp^{*}\right)}. \notag 
\end{align}

Next, we upper-bound $\mathbb{E}_{\vech_{i}}\lefto\{\prob{\mathcal{E}^{HK}_{i}}\right\}$, $i=1,2$, as follows
\begin{align}
&\mathbb{E}_{\vech_{i}}\lefto\{\prob{\mathcal{E}^{HK}_{i}}\right\}  \leq \sum_{k=1}^{6}\mathbb{E}_{\vech_{i}}\lefto\{\prob{\mathcal{E}^{HK}_{ik}}\right\}  \\
&												 \ \dotleq \ \sum_{k=1}^{6} \prob{\mathcal{O}_{ik}\lefto(\Upsilon_{ik}(\vecr),\Psi_{ik}(\vecs^{*}), \vecp^{*}\right)} \\
& \doteq \max_{k=1,2,\dots,6} \prob{\mathcal{O}_{ik}\lefto(\Upsilon_{ik}(\vecr),\Psi_{ik}(\vecs^{*}), \vecp^{*}\right)} \\
															& \doteq \prob{\mathcal{O}^{HK}(\vecr)}.	
\end{align}
The error probability for the two-message, fixed power-split-HK scheme is given by 
\begin{align}
P\lefto(E^{HK}\right) &\doteq \max_{i=1,2} \mathbb{E}_{\vech_{i}}\lefto\{\prob{\mathcal{E}^{HK}_{i}}\right\}  \\
	  &\doteq \max\left\{\prob{\mathcal{O}^{HK}(\vecr)},\prob{\mathcal{O}^{HK}(\vecr)} \right\} \\
	 & \doteq \prob{\mathcal{O}^{HK}(\vecr)} \label{fdfgdsfgdsgxcvcbvbcb}
\end{align}
where \eqref{fdfgdsfgdsgxcvcbvbcb} follows from  the definition of $\prob{\mathcal{O}^{HK}(\vecr)}$.  From the outage lower bound \eqref{gacirt}, we have that 
\begin{align}
\prob{\mathcal{O}^{HK}(\vecr)} \ \dotleq \ P\lefto(E^{HK}\right) \ \dotleq \ \prob{\mathcal{O}^{HK}(\vecr)}
\end{align}
and therefore, 
\begin{align}
P\lefto(E^{HK}\right) \doteq \prob{\mathcal{O}^{HK}(\vecr)}.
\end{align}
\end{proof}

\begin{remark}
It turns out that the total outage probability can be described in a more simple fashion by recognizing that the constraints \eqref{cons2} and \eqref{cons5} are redundant.   An inspection of \eqref{cons2} and \eqref{cons3} immediately yields that 
$d_{i3}(\vecr,\vecs,\vecp) \leq d_{i2}(\vecr,\vecs,\vecp)$ so that \eqref{cons2} can be eliminated.   Finally, \eqref{cons5} can be eliminated as follows:
\begin{itemize}
\item{whenever $p_{j} <  1-\sum\limits_{k=1}^{2}r_{k}+s_{j}$, then \[ d_{i6}(\vecr,\vecs,\vecp) \leq d_{i5}(\vecr,\vecs,\vecp).\]}
\item{whenever $p_{j} \geq  1-\sum\limits_{k=1}^{2}r_{k}+s_{j}$ and \begin{itemize} \item[$\star$]{ $p_{j} \geq  1-\sum\limits_{k=1}^{2}r_{k}+\sum\limits_{l=1}^{2} s_{l}$, then \[d_{i6}(\vecr,\vecs,\vecp) \leq d_{i5}(\vecr,\vecs,\vecp).\]} \item[$\star$]{$p_{j} <  1-\sum\limits_{k=1}^{2}r_{k}+\sum\limits_{l=1}^{2} s_{l}$, then \[d_{j1}(\vecr,\vecs,\vecp) \leq d_{i5}(\vecr,\vecs,\vecp)\] with $i,j=1,2$ and $i \neq j$. }\end{itemize}}
\end{itemize}
It is interesting to observe that analogues of the eliminations carried out in the last step above were reported in \cite{Chong08}.  We note that the elimination of \eqref{cons2} and \eqref{cons5} is equivalent (in terms of DMT) to eliminating conditions \eqref{outage2} and \eqref{outage5} in the characterization of the total outage event in \eqref{totaloutage}. This,  in turn, is equivalent (in terms of DMT) to eliminating \eqref{unsz} and \eqref{unsz3} from the characterization of the achievable rate region  $\mathcal{R}^{*}$.    Now, it can be shown that the HK rate region described in \cite{Chong08}  evaluates precisely to the rate region $\mathcal{R}^{*}$ in \eqref{hkrater} without the constraints \eqref{unsz} and \eqref{unsz3} when the distributions of the inputs are assumed to be i.i.d. Gaussian in \cite{Chong08}.  
\end{remark}

\section{Achievable DMT of the Interference Channel}
We would like to recall that the joint decoder and the two-message, fixed-power-split HK scheme correspond to different power-splits between private and public messages (at the transmitters), different code design criteria, and different decoding algorithms.    As already mentioned, the joint decoder can be viewed as a special case of the two-message, fixed-power-split HK scheme where there are no private messages.  For a given rate tuple $\vecr$,  obviously, either $d^{HK}(\vecr)$ or $d^{JD}(\vecr)$ dominates.  Therefore, the maximum achievable DMT of the fixed-power-split HK scheme  is given by 
\begin{align} 
d(\vecr) = \max\lefto\{d^{HK}(\vecr), d^{JD}(\vecr)\right\} \label{maxref}
\end{align}
and can be achieved by using the appropriate power-split, code designs, and decoding algorithm as follows: 
\begin{itemize}
\item{If $d^{HK}(\vecr) \leq d^{JD}(\vecr)$, employ a family of codebooks satisfying the code design criteria in Theorem \ref{theoremJD},  and use the \ajml{IC}.}
\item{If $d^{HK}(\vecr) > d^{JD}(\vecr)$,  employ a family of codebooks satisfying the code design criteria, the power-split $\vecp^{*}$, and the \ajml{two-message, fixed-power-split HK scheme} in Theorem \ref{theoremHK}.}
\end{itemize} 
In the next section, we show that the fixed-power-split HK scheme is DMT-optimal for certain interference levels. Specifically, we call ICs with $ 1 >\alpha \geq 2/3 $, $ 2> \alpha \geq 1 $, and $ \alpha \geq 2 $ \emph{moderate}, \emph{strong} and \emph{very strong} ICs in the sense of \cite{Etkin08}, respectively.  Next, we will show that the fixed-power-split HK scheme  is DMT-optimal under \emph{moderate}, \emph{strong} and \emph{very strong interference} for symmetric multiplexing rates, i.e., for $r=r_{1}=r_{2}$.

\section{DMT-Optimality}
In this section, we derive an outer bound on the DMT region of the IC that is tighter than the outer bound derived in \cite{Akuiyibo08} for some interference levels.  It turns out that for symmetric multiplexing rates, i.e., when $r=r_{1}=r_{2}$, the two-message, fixed-power-split HK scheme achieves this outer bound for all $\alpha \geq 2/3$. Hence, for $\alpha \geq 2/3$, the two-message, fixed-power-split HK scheme is DMT-optimal for symmetric multiplexing rates.  For $\alpha < 2/3$, unfortunately,  the two-message, fixed-power-split HK scheme does not reach our outer bound.  For asymmetric rate requirements, i.e., when $r_{1}\neq r_{2}$,  we show that the two-message, fixed-power-split HK scheme is DMT-optimal for $\alpha \geq 1$.   We proceed by presenting our outer bound.

\subsection{Outer bound on DMT}
\label{laylaylom}
We consider outer-bounding the capacity region of the IC by providing $\mathcal{R}_{2}$ with the side information $\vecx_{1}$.  As $\mathcal{R}_{2}$ knows the fading coefficient $h_{12}$ perfectly (by assumption), it can cancel the interference out completely, leaving a one-sided IC as depicted in Fig. \ref{onesided}.  Further, we assume that a genie reveals the fading coefficient $h_{21}$ to $\mathcal{T}_{2}$.  It is shown in \cite{Etkin06, Etkin08} that the capacity region of the IC is contained in the following region 
\begin{align}
\mathcal{R}_{ETW}^{1} \triangleq 
\begin{cases}
	\mathcal{D}_{1}, & \SNR^{\alpha}|h_{21}|^{2} < 1 \\
	\mathcal{D}_{2}, & \SNR^{\alpha}|h_{21}|^{2} \geq 1
\end{cases}
\end{align}
where 
\begin{align}
	\mathcal{D}_{1} &\triangleq (S_{1}, T_{1}, S_{2}, T_{2}) :  \notag \\
	&S_{1}+T_{1} \leq \log\left(1+\SNR|h_{11}|^{2}\right) +1  \notag \\
	&S_{2}+T_{2} \leq \log\left(1+\SNR|h_{22}|^{2}\right) +1 \notag \\
\mathcal{D}_{2} &\triangleq (S_{1}, T_{1}, S_{2}, T_{2}) :  \notag \\
	&S_{1}+T_{1} \leq \log\left(1+\SNR|h_{11}|^{2}\right) +1  \notag \\
	&S_{1}+T_{1}+T_{2} \leq \log\left(1+\SNR|h_{11}|^{2}+\SNR^{\alpha}|h_{21}|^{2}\right) +1 \notag \\
	&S_{2} \leq \log\left(1+\SNR^{1-\alpha}\frac{|h_{22}|^{2}}{|h_{21}|^{2}}\right) +1 \notag \\
	&S_{2} + T_{2} \leq \log\left(1+\SNR|h_{22}|^{2}\right) +1 \notag.
\end{align}
For a set $\mathcal{S}$ of quadruples $\{S_{1}, T_{1}, S_{2}, T_{2}\}$, let $\prod(\mathcal{S})$ be the corresponding set of rate pairs such that $R_{1}=S_{1}+T_{1}$ and $R_{2} = S_{2}+T_{2}$.  We recall that $S_{i}=s_{i}\log\SNR$, $T_{i}=t_{i}\log\SNR$, and $R_{i}=r_{i}\log\SNR$ for $i=1,2$. Then, the set
\begin{align}
\mathcal{R}_{ETW}^{*} \triangleq \prod\left(\mathcal{R}_{ETW}^{1}\right)
\end{align}
is an outer bound on the achievable rate region for the IC, i.e., we have
\begin{align}
	\mathcal{R}_{ETW}^{*} \supseteq \mathcal{R}^{\dagger} \label{kitten}
\end{align}
where $\mathcal{R}^{\dagger}$ is any achievable rate region of the IC.  Next, we define the events 
\begin{align}
&   \mathcal{A} \triangleq \left\{h_{21}: \SNR^{\alpha}|h_{21}|^{2} <1 \right\} \notag \\
&   \bar{\mathcal{A}} \triangleq \left\{h_{21}: \SNR^{\alpha}|h_{21}|^{2} \geq 1 \right\} \notag
\end{align}
and 
\begin{align}
&	\mathcal{O}^{ETW}_{1i}(\vecr,\vecs) \triangleq \notag \\
&	\left\{\vech_{i} :  \log\left(1+\SNR|h_{ii}|^{2}\right) +1 < S_{i}+T_{i}\right\} \ \text{for} \ i=1,2\notag \\
&	\mathcal{O}^{ETW}_{13}(\vecr,\vecs) \triangleq \notag \\
&	\left\{\vech_{1} :  \log\left(1+\SNR|h_{11}|^{2}+\SNR^{\alpha}|h_{21}|^{2}\right) +1 < S_{1}+T_{1}+T_{2}   \right\} \notag \\
&   \mathcal{O}^{ETW}_{14}(\vecr, \vecs) \triangleq \notag \\
&	\left\{\vech_{2} :  \log\left(1+\SNR^{1-\alpha}\frac{|h_{22}|^{2}}{|h_{21}|^{2}}\right) +1 < S_{2} \right\} \notag.
\end{align}
Since any achievable rate region for the IC is contained in $\mathcal{R}_{ETW}^{*}$, it follows that the error probability of any scheme communicating over the IC is lower-bounded by
\begin{align} 
\prob{\mathcal{O}^{ETW}(\vecr)} \triangleq \min_{\vecs} \prob{\mathcal{O}^{ETW}_{1}(\vecr, \vecs)}
\end{align}
where the minimization is carried out subject to
\begin{align}
	r_{i} = s_{i} + t_{i} \\
	s_{i}, t_{i} \geq 0 \\
	s_{i}, t_{i} \leq r_{i}
\end{align}
for $i=1,2$ with 
\begin{align}
\mathcal{O}^{ETW}_{1}(\vecr, \vecs) \triangleq \mathcal{K}_{1}(\vecr,\vecs) \bigcup \mathcal{K}_{2}(\vecr,\vecs)
\end{align}
and 
\begin{align}
\mathcal{K}_{1}(\vecr,\vecs) &\triangleq \left(\bigcup_{i=1,2}\mathcal{O}^{ETW}_{1i}(\vecr,\vecs)\right) \bigcap \mathcal{A} \\
\mathcal{K}_{2}(\vecr,\vecs) &\triangleq \left(\bigcup_{i=1,2,3,4}\mathcal{O}^{ETW}_{1i}(\vecr,\vecs)\right) \bigcap \bar{\mathcal{A}}.
\end{align}
Next, we compute $\prob{\mathcal{O}^{ETW}_{1}(\vecr, \vecs)}$. We note that $\mathcal{O}^{ETW}_{1}(\vecr, \vecs) \triangleq \mathcal{K}_{1}(\vecr,\vecs) \bigcup \mathcal{K}_{2}(\vecr,\vecs)$ can equivalently be characterized as:
\begin{align}
&\mathcal{O}^{ETW}_{1}(\vecr, \vecs) = \notag \\ & \notag \left(\bigcup_{i=1,2} \mathcal{O}^{ETW}_{1i}(\vecr,\vecs)\right)\bigcup\left(\bigcup_{i=3,4} \mathcal{O}^{ETW}_{1i}(\vecr,\vecs)\bigcap \bar{\mathcal{A}}\right).
\end{align}
It follows that we can upper-bound $\prob{\mathcal{O}^{ETW}_{1}(\vecr, \vecs)}$ according to 
\begin{align}
& \prob{\mathcal{O}^{ETW}_{1}(\vecr, \vecs)} \leq  \notag \\ & \sum_{i=1}^{2}\prob{\mathcal{O}^{ETW}_{1i}(\vecr,\vecs)} + \sum_{i=3}^{4} \prob{\mathcal{O}^{ETW}_{1i}(\vecr,\vecs)\bigcap \bar{\mathcal{A}}} . \label{up}
\end{align}
We can also lower-bound $\prob{\mathcal{O}^{ETW}_{1}(\vecr, \vecs)}$ according to
\begin{align}
\prob{\mathcal{O}^{ETW}_{1i}(\vecr,\vecs)}  \leq \prob{\mathcal{O}^{ETW}_{1}(\vecr, \vecs)}  \label{low1}
\end{align}
for $i=1,2$.  Further, for $i=3,4$, we have
\begin{align}
\prob{\mathcal{O}^{ETW}_{1i}(\vecr,\vecs)\bigcap \bar{\mathcal{A}}} \leq \prob{\mathcal{O}^{ETW}_{1}(\vecr, \vecs)} \label{low2}.
\end{align}
 We only need to compute the SNR exponents of the upper and lower bounds to obtain the SNR exponent of $\prob{\mathcal{O}^{ETW}_{1}(\vecr, \vecs)}$.  It is shown in \cite{Akuiyibo08} that 
\begin{align}
\prob{\mathcal{O}^{ETW}_{1i}(\vecr,\vecs)} \doteq \SNR^{-d_{1i}^{ETW}(\vecr,\vecs)} \label{xx331}
\end{align}
where $d_{1i}^{ETW}(\vecr,\vecs) = (1-r_{i})^{+}$ for $i=1,2$, and 
\begin{align}
\prob{\mathcal{O}^{ETW}_{13}(\vecr,\vecs)} & \doteq \SNR^{-d_{13}^{ETW}(\vecr,\vecs)} \label{xx332} \\
\prob{\mathcal{O}^{ETW}_{14}(\vecr,\vecs)} & \doteq \SNR^{-d_{14}^{ETW}(\vecr,\vecs)}  \label{xx333} 
\end{align}
with
\begin{align} 
d_{13}^{ETW}(\vecr,\vecs) &= (1-r_{1}-r_{2}+s_{2})^{+}+(\alpha-r_{1}-r_{2}+s_{2})^{+} \notag \\
d_{14}^{ETW}(\vecr,\vecs) &= \begin{cases} 
									(1-\alpha-s_{2})^{+}, & \text{if} \ s_{2} >  0 \ {\rm and}\ \alpha <1 \label{userl} \\
									1, &  \text{if} \ s_{2} = 0  \  \\
									0, &   \text{if} \ s_{2} > 0  \ {\rm and} \ \alpha \geq 1. 
									\end{cases}
\end{align}
Combining \eqref{xx331}-\eqref{xx333} with \eqref{low1}-\eqref{low2} and \eqref{up}, it follows that 
 \begin{align}
 \prob{\mathcal{O}^{ETW}_{1}(\vecr, \vecs)} \doteq \SNR^{-d^{ETW}_{1}(\vecr,\vecs)}
 \end{align}
where 
\begin{align}
 d^{ETW}_{1} = \min_{i=1,2,3,4} d_{1i}^{ETW}(\vecr,\vecs).
\end{align}
The SNR exponent of $\prob{\mathcal{O}^{ETW}(\vecr)}$  is then obtained as
\begin{align} 
\prob{\mathcal{O}^{ETW}(\vecr)}= \min_{\vecs} \SNR^{-d^{ETW}_{1}(\vecr,\vecs)}    \\
									 =  \SNR^{-\max_{\vecs}d^{ETW}_{1}(\vecr,\vecs)}  \label{klo}
\end{align}
where the optimization is carried out subject to
\begin{align}
	r_{i} = s_{i} + t_{i} \\
	s_{i}, t_{i} \geq 0 \\
	s_{i}, t_{i} \leq r_{i}. 
\end{align}
The error probability lower bound \eqref{klo} is in general difficult to evaluate.  However, we show in the next subsection that in some cases, this bound can be evaluated very easily. 

\begin{figure}[htbp]
\begin{center}
\includegraphics[scale=0.4]{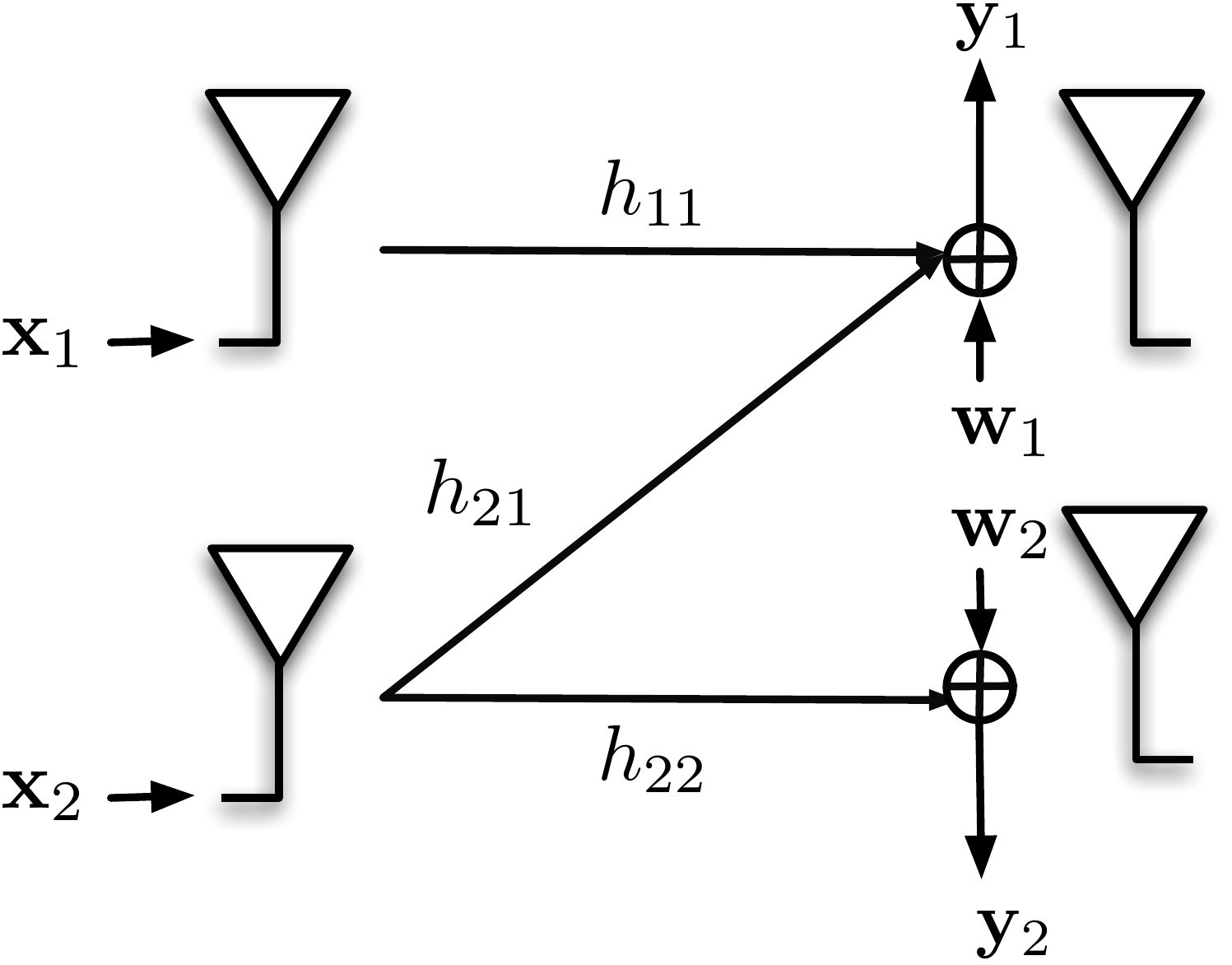}
\caption{One-sided interference channel}
\label{onesided}
\end{center}
\end{figure}

\subsection{The case $\alpha \geq 1$}
It follows immediately from the outer bound \eqref{klo} that the \ajmlong{IC} achieves the optimal DMT of the IC for all interference levels $\alpha \geq 1$. We denote the minimizing value of $\vecs$ in \eqref{klo} by $\vecs^{\dagger}$ and note that the DMT outer bound in Section \ref{laylaylom} can be simplified according to
\begin{align}
d^{ETW}_{1}(\vecr,\vecs^{\dagger}) = d^{JD}(\vecr).
\end{align}
Upon inspection of \eqref{userl}, we see that choosing any $s_{2} > 0$ results in $d_{i4}^{ETW}(\vecr,\vecs) = 0$ for $\alpha \geq 1$. Hence, for any $s_{2}> 0$, we have $d_{1}^{ETW}(\vecr,\vecs) = 0$.   For $s_{2} =0$, we get  
\begin{align}
d_{1i}^{ETW}(\vecr,\vecs) &= (1-r_{i})^{+} \ {\rm for} \ i=1,2 \label{opt1} \\
d_{13}^{ETW}(\vecr,\vecs) &= (1-r_{1}-r_{2})^{+}+(\alpha-r_{1}-r_{2})^{+} \label{opt2} \\
d_{14}^{ETW}(\vecr,\vecs) &= 1 \label{opt3}.
\end{align}
Therefore, $d^{ETW}_{1}(\vecr,\vecs^{\dagger})$ is equivalent to $d^{JD}(\vecr)$ by inspection of \eqref{jfafjfffffsgh} and \eqref{opt1}-\eqref{opt3}. 

\subsection{The case $1>\alpha \geq 2/3$}
For the case $1 > \alpha \geq 2/3$ and for general multiplexing rates for the two transmitters, proving optimality of the two-message, fixed-power-split HK scheme remains elusive.  However, we can show that the two-message, fixed-power-split HK scheme is DMT-optimal for $r_{1}=r_{2} = r$.  The maximum DMT of the two-message, fixed-power-split HK scheme is achieved for $1 \geq\alpha \geq 2/3$ as follows:
\begin{itemize}
	\item{for $r<\alpha/2$, use the \ajmlong{IC} according to Theorem \ref{theoremJD}}
		\item{for $r\geq \alpha/2$, use the \emph{joint ML decoder for the two-message, fixed-power-split HK scheme} according to Theorem \ref{theoremHK} with $p_{i} = 1-\alpha$ and $s_{i}=r-\alpha/2$ for $i=1,2$.}
\end{itemize}
We recall that in the case of symmetric multiplexing rates ($r_{1}=r_{2}=r$),  we have that $s = s_{i}$ for $i=1,2$.  It turns out that the DMT outer bound in \eqref{klo} can be maximized according to
\begin{itemize}
	\item{for $r<\alpha/2$, set $s= 0$.} 
	\item{for $r\geq \alpha/2$, set $s=r-\alpha/2$.}
\end{itemize}

With these choices of optimizing values, an inspection of the DMT outer bound in \eqref{klo} and the achievable region \eqref{maxref} yields that the two regions are equivalent.  Hence, for $1\geq \alpha \geq 2/3$ and $r_{1}=r_{2}=r$,  we have shown that the fixed-power-split HK scheme achieves the optimal DMT.


\section{Very Strong Interference} 
We recall that channels with $\alpha \geq 2$ are called \emph{very strong interference channels} in the sense of \cite{Etkin08}.  We shall see that the condition $\alpha \geq 2$ enables each transmitter-receiver pair to communicate as if the interference were not present.   In this section, we restrict to $\alpha \geq 2$ and show that the joint decoder and a \emph{stripping decoder},  which decodes interference while treating the intended signal as noise, subtracts the result out, and then decodes the desired signal,  are optimal for the IC under very strong interference.

\subsection{Joint decoder}
\label{VSint123255}
Consider the steps \eqref{fadgk} and \eqref{fadgfj} in the proof of the achievable DMT of joint decoding.  We can upper-bound $\prob{\mathcal{E}^{JD}_{ik}}$ as 
\begin{align}
&\prob{\mathcal{E}^{JD}_{ik}} = \prob{\mathcal{E}^{JD}_{ik}, \mathcal{O}^{JD}_{ik}(\vecr)} +\prob{\mathcal{E}^{JD}_{ik}, \bar{\mathcal{O}}^{JD}_{ik}(\vecr)}  \notag \\
&\leq   \prob{ \mathcal{O}^{JD}_{ik}(\vecr)}+ \prob{\mathcal{E}^{JD}_{ik}| \bar{\mathcal{O}}^{JD}_{ik}(\vecr)}  \label{ghas}
\end{align}
for $k=1,2$. We will see that this approach leads to stricter design criteria, but in exchange enables us to decouple the IC as we will  demonstrate shortly.  Using \eqref{stars} in \eqref{ghas} and noting that $\bar{\mathcal{O}}^{JD}_{ik}(\vecr)$ entails $\SNR|h_{ii}|^{2}+\SNR^{\alpha}|h_{ji}|^{2} \geq \SNR^{r_{1}+r_{2}}-1$, we can upper-bound $\mathbb{E}_{\vech_{i}} \lefto\{\prob{\mathcal{E}^{JD}_{i2}}\right\}$ according to 
\begin{align}
&\mathbb{E}_{\vech_{i}} \lefto\{\prob{\mathcal{E}^{JD}_{i2}}\right\} \ \dotleq \  \\&  \prob{ \mathcal{O}_{i2}^{JD}\left(\vecr\right)}+ \SNR^{N(r_{1}+r_{2})}\exp\left[-\frac{\lambda_{\min}\SNR^{r_{1}+r_{2}}}{4}\right]   \notag.
\end{align}
We recall that $\lambda_{\min}$ is the smallest eigenvalue of $\Delta\matX_{ij}(\Delta\matX_{ij})^{H}$.  Hence, if $\lambda_{\min} \ \dotgeq \ \SNR^{-r_{1}-r_{2}+\epsilon}$ for some $\epsilon>0$, we have that 
\begin{align}
	&\mathbb{E}_{\vech_{i}} \lefto\{\prob{\mathcal{E}^{JD}_{i2}}\right\} \ \dotleq \ \prob{\mathcal{O}^{JD}_{i2}(\vecr)} \label{abc}.
\end{align}
 
 Similarly, using \eqref{gasgafas} in \eqref{ghas} and noting that $\mathcal{O}_{i1}^{JD}$ entails $\SNR|h_{ii}|^{2} \geq \SNR^{r_{i}}-1$, we get 
 \begin{align}
&\mathbb{E}_{\vech_{i}}\lefto\{\prob{\mathcal{E}^{JD}_{i1}}\right\} \notag \\ 
& \ \ \ \dotleq \ \prob{ \mathcal{O}_{i1}^{JD}\left(\vecr\right)}+ \SNR^{Nr_{i}}\exp\left[-\frac{\SNR^{r_{i}}\|\Delta\vecx_{i}\|^2}{4}\right] \label{rhsboy}.
 \end{align}
 If $\|\Delta\vecx_{i}\|^{2} \ \dotgeq \ \SNR^{-r_{i}+\epsilon}$ for some $\epsilon  > 0 $ for every pair of codewords, the second term on the RHS of \eqref{rhsboy} decays exponentially, leaving the polynomially decaying term, according to
 \begin{align}
 \mathbb{E}_{\vech_{i}}\lefto\{\prob{\mathcal{E}^{JD}_{i1}}\right\} \ \dotleq \ \prob{ \mathcal{O}_{i1}^{JD}\left(\vecr\right)} \label{bcg}.
 \end{align}
 Inserting \eqref{abc} and \eqref{bcg} into \eqref{hoijd}, we get 
 \begin{align}
& \mathbb{E}_{\vech_{i}}\lefto\{\prob{\mathcal{E}^{JD}_{i}} \right\}  \leq \sum_{k=1}^{2}\mathbb{E}_{\vech_{i}}\lefto\{ \prob{\mathcal{E}_{ik}^{JD}} \right\} \\
										&\ \dotleq \  \prob{ \mathcal{O}_{i1}^{JD}\left(\vecr\right)} +\prob{ \mathcal{O}_{i2}^{JD}\left(\vecr\right)} \\  
										& \doteq \ \SNR^{-(1-r_{i})^{+}} + \SNR^{-(1-r_{1}-r_{2})^{+} - (\alpha-r_{1}-r_{2})^{+}} \label{ffsazez}
\end{align}
for $i=1,2$.   We simplify \eqref{ffsazez} for $\alpha \geq 2$ to get 
\begin{align}
	\mathbb{E}_{\vech_{i}}\lefto\{\prob{\mathcal{E}^{JD}_{i}} \right\} \ \dotleq \ \SNR^{-(1-r_{i})^{+}}.
\end{align}
We recall that $P(E_{ii})$ is the average ML error probability under the assumption that the perfectly decoded interference has been removed.  We note that $P(E_{ii})$ is a lower bound on $\mathbb{E}_{\vech_{i}}\lefto\{\prob{\mathcal{E}^{JD}_{i}} \right\}$.  Further, by the outage bound on $P(E_{ii})$ \cite{zheng_tradeoff}, $P(E_{ii})$ is lower-bounded according to
\begin{align}
\SNR^{-(1-r_{i})^{+}}\  \dotleq \ P(E_{ii}) \ \dotleq \ \mathbb{E}_{\vech_{i}}\lefto\{\prob{\mathcal{E}^{JD}_{i}} \right\} \ \dotleq \  \SNR^{-(1-r_{i})^{+}} \notag.
\end{align}
Hence, we get
\begin{align}
\ P(E_{ii}) \doteq  \mathbb{E}_{\vech_{i}}\lefto\{\prob{\mathcal{E}^{JD}_{i}} \right\}  \doteq   \SNR^{-(1-r_{i})^{+}}.
\end{align} This shows that under very strong interference, the IC  is effectively \emph{decoupled}, in the sense that, it is possible to achieve the performance of two point-to-point SISO systems without interference, provided that we employ a family of codebooks that satisfy
\begin{align}
\|\Delta\vecx_{i}\|^{2} \ \dotgeq \ \SNR^{-r_{i}+\epsilon} \\
\lambda_{\min}\left( \Delta\matX_{ij}(\Delta\matX_{ij})^{H}\right) \ \dotgeq \ \SNR^{-r_{1}-r_{2}+\epsilon}
\end{align}
for all pairs of codewords $\vecx_{i}^{n_{i}},  \vecx_{i}^{\tilde{n}_{i}} \in \mathcal{C}_{i}(\SNR,r_{i})$ s.t. $\vecx_{i}^{n_{i}} \neq \vecx^{\tilde{n}_{i}}_{i}$, $\vecx_{j}^{n_{j}},  \vecx_{j}^{\tilde{n}_{j}} \in \mathcal{C}_{j}(\SNR,r_{j})$ s.t. $\vecx_{j}^{n_{j}} \neq \vecx^{\tilde{n}_{j}}_{j}$  for $i,j = 1,2$ and $i \neq j$,  where  $\Delta \vecx_{i} = \vecx_{i}^{n_{i}} - \vecx^{\tilde{n}_{i}}_{i}$, $\Delta \vecx_{j} = \vecx_{j}^{n_{j}} - \vecx^{\tilde{n}_{j}}_{j}$,   and $\Delta\matX_{ij} = [\Delta\vecx_{i} \ \Delta\vecx_{j}]$, and $\lambda_{\min}(\Delta\matX_{ij} (\Delta\matX_{ij})^{H})$ denotes the smallest nonzero eigenvalue of $\Delta\matX_{ij} (\Delta\matX_{ij})^{H}$, for some\footnote{We note that $\epsilon$ is allowed to be different in \eqref{xxczzxvz} and \eqref{xxczzxvz2}. } $\epsilon >  0$, with a power-split according to $p_{i}=-\infty$ for $i=1,2$ and the receiver algorithm corresponding to the joint decoder described earlier. Hence, the joint decoder is DMT-optimal under very strong interference.  What is more, as shown next, a stripping decoder achieves the DMT performance of the joint decoder, and therefore, is also DMT-optimal.

\subsection{Stripping decoder}
In this section, we take $N=1$; we will see later that optimal performance can be achieved for $N\geq 1$, in contrast to the fixed-power-split HK scheme.  In the following, we use the short-hand $x_{i}$ for the first element of the transmit signal vector $\vecx_{i}$, $y_{i}$ for the first element of the receive signal vector $\vecy_{i}$, and $\mathcal{X}_{i}$ for $\mathcal{C}_{i}(\SNR,r_{i})$. 

We write  $\prob{E_{ij}|\vech_{j}}$ for $i,j=1,2$ and $i\neq j$ for the ML decoding error probability of decoding $ \mathcal{T}_{i}$ at receiver $\mathcal{R}_{j}$ under the assumption that $ \mathcal{T}_{j}$ is treated as noise. We define the respective average ML decoding error probability as $P(E_{ij}) = \mathbb{E}_{\vech_{j}}\lefto\{\prob{E_{ij}|\vech_{j}}\right\}$.  We assume throughout that the transmit symbols are equally likely for both transmitters, and hence $\prob{x_{i}} = \frac{1}{|\mathcal{X}_{i}|}$ for $i=1,2$.



In the following, we show that a stripping decoder achieves the DMT outer bound in \cite{Akuiyibo08} given by 
\begin{align}
\label{Leveq}d(r) \leq \min\{(1-r_{1})^{+},(1-r_{2})^{+}\}.
  \end{align}
\begin{theorem}
\label{VSint1232}
For the fading IC with I/O relation (\ref{intc1s})-(\ref{intc2s}), we have   
\begin{align}
		P(E) \doteq \SNR^{- \min\{(1-r_{1})^{+},(1-r_{2})^{+}\}}
\end{align}
provided that $\Delta x_{i} = x^{j}_{i}-x_{i}^{k}$ satisfies $|\Delta x_{i}|^{2} \ \dotgeq \ \SNR^{-r_{i}+\epsilon}$ for every pair $x^{j}_{i},x^{k}_{i}$ in each codebook $\mathcal{X}_{i}$, $i=1,2$, and  for some $\epsilon > 0 $. 
\end{theorem}
\begin{proof}
In the following, we show that a stripping decoder achieves the optimal DMT region.  We start by decoding $ \mathcal{T}_{2}$ at $\mathcal{R}_{1}$ while treating $ \mathcal{T}_{1}$ as noise,  i.e., we have the effective I/O relation 
\begin{align}
\label{xzzxe}
	y_{1} & = \sqrt{ \SNR^{\alpha} } h_{21}x_{2} + \tilde{z}
\end{align}
where	$\tilde{z}$ is the effective noise term with variance $1+ \SNR|h_{11}|^{2}$. We next note that the worst case (in terms of mutual information and hence outage probability) uncorrelated (with the transmit signal) additive noise under a variance constraint is Gaussian \cite[Theorem 1]{Hassibi03}.  In the following, we use the corresponding worst case outage probability to exponentially upper-bound $P(E_{21})$, i.e., we set $\tilde{z}\sim \mathcal{CN}(0, 1+ \SNR|h_{11}|^{2})$.   We start by normalizing the received signal according to
\begin{align}
	\frac{y_{1}}{\sqrt{1+\SNR|h_{11}|^2}} & = \sqrt{\frac{\SNR^{\alpha}}{1+\SNR|h_{11}|^2}}h_{21}x_{2}+ z
\end{align}	
where $ z \sim \mathcal{CN}(0,1)$.  We can now upper-bound $\prob{E_{21}|\vech_{1}}$ as
\begin{align}
&	\prob{E_{21}|\vech_{1}} = \sum_{x_{2} \in \mathcal{X}_{2}} \prob{x_{2}}\prob{E_{21}| \vech_{1}, x_{2} } \\
\label{equipp}							    & \ \ \ = \frac{1}{|\mathcal{X}_{2}|} \sum_{i=1}^{|\mathcal{X}_{2}|} \prob{\bigcup\limits_{\substack{ j =1    \\ {j  \neq i}}}^{|\mathcal{X}_{2}|}  \! x^{i}_{2} \rightarrow x^{j}_{2}\left.\right| \vech_{1}} \\
\label{unionbound}							    & \ \ \ \leq |\mathcal{X}_{2}| \prob{x^{\tilde{i}}_{2} \rightarrow x^{\tilde{j}}_{2}\left.\right| \vech_{1}} \\
							 \label{cortez}   &\ \ \ \leq |\mathcal{X}_{2}| {\rm Q} \left(\sqrt{\frac{\SNR^{\alpha}|h_{21}|^{2}|\Delta x_{2}|^{2}}{2(1+\SNR|h_{11}|^{2})}}\right),
\end{align}
where $\left\{x_{2}^{\tilde{i}}, x_{2}^{\tilde{j}}\right\}$ denotes the (or ``a'' in the case of multiple pairs with the same distance) pair of symbols with minimum Euclidean distance among all possible pairs of different symbols.  We now define the outage event $\mathcal{O}_{ii}$ associated with decoding $ \mathcal{T}_{i}$ at $\mathcal{R}_{i}$ ($i=1,2$) in the absence of  interference and its complementary event $\bar{\mathcal{O}}_{ii}$  as follows
\begin{align}
\label{fax}	\mathcal{O}_{ii} & \triangleq \left\{h_{ii}: \log\left(1+\SNR|h_{ii}|^{2}\right) < R_{i} \right\} \\
\label{faxoz}	\bar{\mathcal{O}}_{ii} &\triangleq \left\{h_{ii}: \log\left(1+\SNR|h_{ii}|^{2}\right) \geq R_{i} \right\}.
\end{align}
We note that this definition is consistent with the definition of $P(E_{ii})$. Similarly, we define the event $\mathcal{O}_{ij}$ associated with decoding $ \mathcal{T}_{i}$ at $\mathcal{R}_{j}$ while treating $ \mathcal{T}_{j}$ as noise ($i,j=1,2$ and $i\neq j$) and its complementary event $\bar{\mathcal{O}}_{ij}$  as follows
\begin{align}
	\mathcal{O}_{ij} & \triangleq \left\{\vech_{j}: \log\left(1+\frac{\SNR^{\alpha}|h_{ij}|^{2}}{1+\SNR|h_{jj}|^2}\right) < R_{i} \right\} \notag \\
	\bar{\mathcal{O}}_{ij} &\triangleq  \left\{\vech_{j}: \log\left(1+\frac{\SNR^{\alpha}|h_{ij}|^{2}}{1+\SNR|h_{jj}|^2}\right) \geq R_{i} \right\} \notag. 
\end{align}
Next, we upper-bound $P(E_{21})$ according to
\begin{align}
&P(E_{21}) = \Exop_{\vech_{1}} \! \lefto\{\prob{E_{21}|\vech_{1}}\right\}  =  \notag \\[0.1cm]
& \label{molt1}  \Exop_{\vech_{1}}\! \lefto\{\prob{\mathcal{O}_{21}}\prob{E_{21}|\vech_{1}, \mathcal{O}_{21}} \! + \!  \prob{\bar{\mathcal{O}}_{21}}\prob{E_{21}|\vech_{1},\bar{\mathcal{O}}_{21}} \! \right\}   \\[0.1cm]
 &\label{molt2}  \ \ \ \ \ \leq  \prob{\mathcal{O}_{21}}+ \Exop_{\vech_{1}}\! \lefto\{\prob{E_{21}|\vech_{1},\bar{\mathcal{O}}_{21}} \right\}  \\[0.1cm]
 &\label{molt3} \ \ \ \  \ \leq \prob{\mathcal{O}_{21}} + \SNR^{r_{2}} {\rm Q} \left(\sqrt{\frac{\SNR^{r_{2}}|\Delta x_{2}|^{2}}{2}}\right)
\end{align} where (\ref{molt1}) follows from Bayes's rule and (\ref{molt2}) is obtained by upper-bounding $\prob{E_{21}|\vech_{1}, \mathcal{O}_{21}}$ and $\prob{\bar{\mathcal{O}}_{21}}$ by $1$.  Finally, (\ref{molt3}) follows by using the fact that $\bar{\mathcal{O}}_{21}$ entails $\frac{\SNR^{\alpha}|h_{21}|^{2}}{1+\SNR|h_{11}|^2} \geq 2^{R_{2}}-1$, and invoking $R_{2}=r_{2}\log\SNR$, $|\mathcal{X}_{2}| = \SNR^{r_{2}}$, and $\SNR \gg 1$ in (\ref{cortez}).  It can be shown that $\prob{\mathcal{O}_{21}} \doteq \SNR^{-(\alpha-1-r_{2})^{+}}$ for $\alpha \geq  2$ \cite{Akuiyibo08}.  Further, since $|\Delta x_{2}|^{2} \ \dotgeq \ \SNR^{-r_{2}+\epsilon}$, for $\epsilon > 0$, by assumption,  we can further simplify the above as the second term in (\ref{molt3}) decays exponentially in SNR whereas the first term decays polynomially, i.e., we get \begin{align}
\Exop_{\vech_{1}}\! \left\{\prob{E_{21}|\vech_{1}}\right\}  \ \dotleq \  \prob{\mathcal{O}_{21}} \doteq \SNR^{-(\alpha-1-r_{2})^{+}}.
\end{align}
We proceed to analyze decoding of $ \mathcal{T}_{1}$ at $\mathcal{R}_{1}$ and start by defining $\bar{x}_{2}$ as the result of decoding $ \mathcal{T}_{2}$ at $\mathcal{R}_{1}$.  Note that we do not need to assume that $ \mathcal{T}_{2}$ was decoded correctly at $\mathcal{R}_{1}$. We begin by upper-bounding $\prob{E_{11}|\vech_{1}}$ given $\bar{x}_{2}$: 
\begin{align}
&\prob{E_{11}|\vech_{1},\bar{x}_{2}}   \notag \\
\label{xxza2}& = \sum_{x_{1} \in \mathcal{X}_{1}} \sum_{x_{2}\in \mathcal{X}_{2}} \prob{x_{1}} \prob{x_{2}} \prob{E_{1}| \vech_{1},x_{1}, x_{2}, \bar{x}_{2}} \\ 
\label{xxza3}& = \frac{1}{|\mathcal{X}_{1}||\mathcal{X}_{2}|} \sum_{i=1}^{|\mathcal{X}_{1}|}   \sum_{k=1}^{|\mathcal{X}_{2}|}  \prob{\!  \bigcup\limits_{\substack{j =1     \\ {j\neq i}}}^{|\mathcal{X}_{1}|}  x^{i}_{1} \! \rightarrow x^{j}_{1}\left.\right|\vech_{1},x_{2}^{k},\bar{x}_{2}}  \\
\label{xxza4}& \leq \frac{|\mathcal{X}_{1}|}{|\mathcal{X}_{2}|}\sum_{k=1}^{|\mathcal{X}_{2}|}\prob{x^{\tilde{i}}_{1} \rightarrow x^{\tilde{j}}_{1}\!\left.\right|\vech_{1},x_{2}^{k},\bar{x}_{2}}\! ,
\end{align}
where $\left\{x_{1}^{\tilde{i}}, x_{1}^{\tilde{j}}\right\}$ denotes the (or ``a'' in the case of multiple pairs with the same distance) pair of symbols with minimum Euclidean distance among all possible pairs of different symbols. Next, we further upper-bound $\prob{E_{11}|\vech_{1},\bar{x}_{2}}$ by considering two events; namely, when $\mathcal{R}_{1}$ decodes $ \mathcal{T}_{2}$ correctly and when it does not:
\begin{align}
& \prob{E_{11}|\vech_{1},\bar{x}_{2}}  \leq  \notag \\ 
& \frac{|\mathcal{X}_{1}|}{|\mathcal{X}_{2}|} \sum_{k=1}^{|\mathcal{X}_{2}|} \left(\prob{\bar{x}_{2}\!=x_{2}^{k}| \vech_{1}, x_{2}^{k}}\!\prob{x^{\tilde{i}}_{1} \!  \rightarrow x^{\tilde{j}}_{1}  \!\left.\right| \! \vech_{1},x_{2}^{k}, \bar{x}_{2} ,\bar{x}_{2} \! = \! x_{2}^{k}} \! \right.  \notag \\
&\left. + \prob{\bar{x}_{2}\! \neq x_{2}^{k}|\vech_{1}, x_{2}^{k}}\!\prob{x^{\tilde{i}}_{1} \!  \rightarrow x^{\tilde{j}}_{1}\!  \left. \right| \! \vech_{1},x_{2}^{k}, \bar{x}_{2}, \bar{x}_{2}\! \neq x_{2}^{k}} \!  \right) \! , \label{junkf}
\end{align}
where $\prob{x^{\tilde{i}}_{1} \!  \rightarrow x^{\tilde{j}}_{1}\! \left.\right| \! \vech_{1},x_{2}^{k}, \bar{x}_{2} ,\bar{x}_{2} \! = \! x_{2}^{k}}$ is the probability of mistakenly decoding $x_{1}^{\tilde{i}}$ for $x_{1}^{\tilde{j}}$ given that $ \mathcal{T}_{2}$ transmitted $x_{2}^{k}$ and $\mathcal{R}_{1}$ decoded $ \mathcal{T}_{2}$ correctly, i.e.,  $\bar{x}_{2} =x_{2}^{k}$.  The quantity $\prob{\bar{x}_{2} = x_{2}^{k}|\vech_{1}, x_{2}^{k}}$ is the probability of decoding $ \mathcal{T}_{2}$ correctly given that $x_{2}^{k}$ was transmitted. By upper-bounding $\prob{\bar{x}_{2}=x_{2}^{k}|\vech_{1},x_{2}^{k}}$ and $\prob{x^{\tilde{i}}_{1} \! \rightarrow x^{\tilde{j}}_{1}\left.\right| \! \vech_{1},x_{2}^{k}, \bar{x}_{2}, \bar{x}_{2} \! \neq x_{2}^{k}}$ in (\ref{junkf}) by $1$, we arrive at 
\begin{align}
 \prob{E_{11}|\vech_{1},\bar{x}_{2}} &\leq \frac{|\mathcal{X}_{1}|}{|\mathcal{X}_{2}|} \sum_{k=1}^{|\mathcal{X}_{2}|} \prob{x^{\tilde{i}}_{1} \rightarrow x^{\tilde{j}}_{1} \! \left.\right|\vech_{1},x_{2}^{k},\bar{x}_{2},\bar{x}_{2}=x_{2}^{k}} + \notag \\   
 &\label{georg}  \ \ \ \ \ \ \ \  \ \  \frac{|\mathcal{X}_{1}|}{|\mathcal{X}_{2}|}\sum_{k=1}^{|\mathcal{X}_{2}|} \prob{\bar{x}_{2}\neq x_{2}^{k}|\vech_{1}, x_{2}^{k}}. 
\end{align}
Next, noting that $\frac{1}{|\mathcal{X}_{2}|}\sum\limits_{k=1}^{|\mathcal{X}_{2}|}\prob{\bar{x}_{2}\neq x_{2}^{k}|\vech_{1}, x_{2}^{k}}\leq \prob{E_{21}|\vech_{1}}$ and invoking the corresponding upper bound (\ref{cortez}) in (\ref{georg}), we get
 \begin{align}
 \prob{E_{1}|\vech_{1},\bar{x}_{2}} & \leq |\mathcal{X}_{1}|{\rm Q}\left(\sqrt{\frac{\SNR|h_{11}|^{2}|\Delta  x_{1}|^{2}}{2}}\right)+ \notag \\ & \ \ \ \ \ \ \ 
\label{georg2}|\mathcal{X}_{1}||\mathcal{X}_{2}|{\rm Q}\left(\sqrt{\frac{\SNR^{\alpha}|h_{21}|^{2}|\Delta  x_{2}|^{2}}{2(1+\SNR|h_{11}|^{2})}}\right).
\end{align} The first term on the RHS of (\ref{georg2}) follows from the first term on the RHS of (\ref{georg}), since given $\bar{x}_{2} =x_{2}^{k}$, the interference is subtracted out perfectly, leaving an effective SISO channel without interference.  We are now in a position to upper-bound $P(E_{11})$:
\begin{align}
\label{consxzzz}	& P(E_{11}) = \Exop_{\vech_{1}}\!\lefto\{\prob{E_{11}|\vech_{1}}\right\} \leq   \Exop_{\vech_{1}}\!\lefto\{\prob{E_{11}|\vech_{1},\bar{x}_{2}}\right\} \\[0.12cm] 
	& \leq \Exop_{\vech_{1}}\lefto\{ |\mathcal{X}_{1}|{\rm Q}\left(\sqrt{\frac{\SNR|h_{11}|^{2}|\Delta  x_{1}|^{2}}{2}}\right)\right\}+\notag \\[0.12cm] 
\label{xxzassr}	& \Exop_{\vech_{1}}\lefto\{|\mathcal{X}_{1}||\mathcal{X}_{2}|{\rm Q}\left(\sqrt{\frac{\SNR^{\alpha}|h_{21}|^{2}|\Delta  x_{2}|^{2}}{2(1+\SNR|h_{11}|^{2})}}\right) \right\}.
\end{align}
Here, (\ref{consxzzz}) follows since the error probability incurred by using the stripping decoder constitutes a natural upper bound on $\Exop_{\vech_{1}}\!\lefto\{\prob{E_{11}|\vech_{1}}\right\} $. We upper-bound (\ref{xxzassr}) by splitting each of the two terms into outage and no outage sets using Bayes's rule to arrive at
\begin{align}
	&  P(E_{11}) = \Exop_{\vech_{1}}\lefto\{\prob{E_{11}|\vech_{1}}\right\} \leq  \notag \\[0.08cm] 
 	& \prob{\mathcal{O}_{11}} + \SNR^{r_{1}}{ \rm Q}\lefto(\frac{\SNR^{r_1}|\Delta x_{1}|^{2}}{2}\right) + \prob{\mathcal{O}_{21}} + \notag 
	\\ 	& \ \ \ \ \ \ \ \ \ \ \ \ \ \ \ \ \ \ \ \ \ \ \ \ \ \ \ \ \ \ \ \ \   
\label{barbu}	\SNR^{r_{1}+r_{2}}{ \rm Q}\lefto(\frac{\SNR^{r_2}|\Delta x_{2}|^{2}}{2}\right). 
 \end{align}
The second and fourth terms on the RHS of (\ref{barbu}) follow from (\ref{xxzassr}) since $\bar{\mathcal{O}}_{11}$ and $\bar{\mathcal{O}}_{21}$ entail $\SNR |h_{11}|^{2} \geq 2^{R_{1}}-1$ and $\frac{\SNR^{\alpha}|h_{21}|^{2}}{1+\SNR|h_{11}|^2} \geq 2^{R_{2}}-1$, respectively, and since $R_{i}=r_{i}\log\SNR$, $|\mathcal{X}_{i}| = \SNR^{r_{i}}$ for $i=1,2$, and $\SNR \gg 1$.  Given that the minimum Euclidean distances in each codebook, $|\Delta  x_{1}|^{2}$ and $|\Delta  x_{2}|^{2}$,  obey $|\Delta  x_{1}|^{2} \ \dotgeq \ \SNR^{-r_{1}+\epsilon}$ and $|\Delta  x_{2}|^{2} \ \dotgeq \ \SNR^{-r_{2}+\epsilon}$, for some $\epsilon > 0$, by assumption, we get
\begin{align}
	P(E_{11}) &= \Exop_{\vech_{1}}\!\left\{\prob{E_{11}|\vech_{1}}\right\}  \ \dotleq  \ \ \prob{\mathcal{O}_{11}} +  \prob{\mathcal{O}_{21}} \\ 
											                  & \ \doteq  \  \SNR^{-\maxo{1-r_{1}}} + \SNR^{-\maxo{\alpha-1-r_{2}}} \\
\label{xxxaszadv1}											                  &  \ \doteq  \  \SNR^{-\min\{\maxo{1-r_{1}},\maxo{\alpha-1-r_{2}}\}}.
\end{align}	
Similar derivations for decoding at $\mathcal{R}_{2}$ lead to  
\begin{align}
	\label{xxxaszadv2}											P(E_{22})            \ \dotleq  \  \SNR^{-\min\{\maxo{1-r_{2}},\maxo{\alpha-1-r_{1}}\}}.
\end{align}
We note that the error probability of decoding $ \mathcal{T}_{i}$ at $\mathcal{R}_{i}$ is exponentially lower-bounded by $\prob{\mathcal{O}_{ii}}$ for $i=1,2$ \cite{zheng_tradeoff}.  Hence, $P(E_{ii})$ is  sandwiched according to
\begin{align}
\label{sand1} &\!\!\! \SNR^{-\maxo{1-r_{i}}}  \dotleq \  P(E_{ii})  \ \dotleq   \  \SNR^{-\min\{\maxo{1-r_{i}},\maxo{\alpha-1-r_{j}}\}} 
\end{align}
for $i,j=1,2$ and $i\neq j$.  The proof is concluded by first upper-bounding \begin{align} P(E) =\max\{P(E_{11}), P(E_{22})\} \end{align} as
\begin{align}
P(E)  \ \dotleq  \ & \max  \lefto\{\SNR^{-\min\{\maxo{1-r_{1}},\maxo{\alpha-1-r_{2}}\}}, \right. \notag  \\ 
 & \ \ \ \left.   \SNR^{-\min\{\maxo{1-r_{2}},\maxo{\alpha-1-r_{1}}\}}\right\}  \notag \\
\label{sensib1} \doteq & \ \SNR^{-\min\lefto\{\maxo{1-r_{1}},\maxo{1-r_{2}}\right\}}
\end{align}
where  (\ref{sensib1}) is a consequence of the assumption $\alpha \geq 2$.   Secondly, $P(E)$ can be lower-bounded using the outage bounds on the individual error probabilities:
\begin{align}
\label{sensib2}	\SNR^{-\min\lefto\{\maxo{1-r_{1}},\maxo{1-r_{2}}\right\}} \ &\dotleq  \ P(E).	
\end{align}
Since the $\SNR$ exponents in the upper bound (\ref{sensib1}) and the lower bound (\ref{sensib2}) match, we can conclude that  
 \begin{align}
P(E) \doteq \SNR^{-\min\lefto\{\maxo{1-r_{1}},\maxo{1-r_{2}}\right\}}
 \end{align}
which establishes the desired result. \end{proof}
\begin{remark} 
We can immediately conclude from Theorem \ref{VSint1232} that using a sequence of codebooks that is DMT-optimal for the SISO channel for both users results in DMT-optimality for the IC under very strong interference. 
\end{remark}
\begin{remark} 
If $r_{1}=r_{2}=r$ and we use sequences of codebooks $\mathcal{C}(\SNR,r)$ satisfying the conditions of Theorem \ref{VSint1232} for both users,  then we have 
\begin{align}
P(E_{11}) \doteq P(E_{22}) \doteq  \SNR^{-(1-r)^{+}}
\end{align}
as a simple consequence of (\ref{sand1}). This means that in the special case, where each $ \mathcal{T}_{i}$ transmits at the same multiplexing rate $r$, we have the stronger result that the single user DMT, i.e., the DMT that is achievable for a SISO channel in the absence of any interferers,  is achievable for both users.  In effect, under very strong interference and when the two users operate at the same multiplexing rate, the interference channel effectively gets  \emph{decoupled}. For a stripping decoder and $r_{1}\neq r_{2}$, we can, in general, not arrive at the same conclusion as the SNR exponents in (\ref{sand1}) do not necessarily match. 
\end{remark}



\section{Suboptimal Strategies}
In the following, we investigate the DMT performance of treating the IC as a combination of two MACs and sharing transmission time between the two transmitters. These strategies are suboptimal; in fact, it can be shown that the two-message, fixed-power-split HK scheme always outperforms these schemes.  Nevertheless, we analyze these two schemes as they are of some practical importance.  

\subsection{Achievable DMT for treating the IC as a combination of two MACs}
\label{jointMAC}
A simple achievable rate region for the IC is obtained by treating the IC as a MAC at each receiver $\mathcal{R}_{j}$ for $j=1,2$.  Next,  we formally define the strategy of treating the IC as a combination of two MACs.
\begin{definer}
\label{MACdefinition}
A MAC at $\mathcal{R}_{i}$ is obtained by requiring the messages from both transmitters $\mathcal{T}_{j}$, $j=1,2$, to be decoded at $\mathcal{R}_{i}$ for $i=1,2$.  
\end{definer}
\begin{definer}
A \ajmlong{MAC} at $\mathcal{R}_{j }$ ($j=1,2$)  carries out joint ML detection on the messages from both transmitters ($\mathcal{T}_{i}$ for $i=1,2$).    The ML error probability and the average ML error probability of this receiver are denoted by $\prob{\mathcal{E}^{MAC}_{j}}$ and $P\lefto(E^{MAC}_{j}\right)\triangleq\mathbb{E}_{\vech_{j}}\left\{\prob{\mathcal{E}^{MAC}_{j}}\right\}$, respectively.
\end{definer}

The following theorem provides the achievable  DMT for the strategy of treating the IC as a combination of two MACs.
\begin{theorem}
\label{theoremMAC}
The DMT corresponding to treating the IC as a MAC at each receiver is given by
\begin{align}
d^{MAC}(\vecr) =   \min\limits_{\substack{i=1,2\\ k=1,2,3}}\left\{d^{MAC}_{ik}(\vecr)\right\}
\end{align}
where
 \begin{align}
	&d^{MAC}_{i1}(\vecr)= (1-r_{i})^{+} \notag  \\
	&d^{MAC}_{i2}(\vecr)= (\alpha-r_{j})^{+},  \ \ \ \text{for} \ \  i,j=1,2  \ \text{and} \ i\neq j \notag \\
	& d^{MAC}_{i3}(\vecr) = \lefto(1-r_{1}-r_{2}\right)^{+}+ \lefto(\alpha-r_{1}-r_{2}\right)^{+}. \notag
\end{align}
Denote 
\begin{align} 
	[i^{*} \ k^{*}]= \arg \min_{\substack{i=1,2 \\ k= 1,2,3}} d^{MAC}_{ik}(\vecr). 
\end{align} Let $\Xi_{ik}(\vecr) =[\xi_{ik}^{1}(\vecr) \ \xi_{ik}^{2}(\vecr)]^{T}$ be functions\footnote{We note that the functions $\Xi_{ik}(\vecr)$ might not be unique.} such that 
\begin{align}
d^{MAC}_{i^{*}k^{*}}(\vecr) = d^{MAC}_{ik}(\Xi_{ik}(\vecr))
\end{align}  for $i=1,2$, $k=1,2,3$. 
If a sequence (in SNR) of codebooks with block length $N \geq 2$ satisfies
\begin{align}
	\| \Delta \vecx_{i}\|^{2} \ & \dotgeq \ \SNR^{-\min\lefto\{\xi_{i1}^{i}(\vecr), \xi_{j2}^{i}(\vecr)\right\}+\epsilon} \label{cdcmac1}\\
	\lambda_{\min}\lefto(\Delta\matX_{ij} (\Delta\matX_{ij})^{H}\right) \ & \dotgeq \ \SNR^{-\xi_{i3}^{1}(\vecr)-\xi_{i3}^{2}(\vecr)+\epsilon} \label{cdcmac2}
\end{align}
for all pairs of codewords $\vecx_{i}^{n_{i}},  \vecx_{i}^{\tilde{n}_{i}} \in \mathcal{C}_{i}(\SNR,r_{i})$ s.t. $\vecx_{i}^{n_{i}} \neq \vecx^{\tilde{n}_{i}}_{i}$, $\vecx_{j}^{n_{j}},  \vecx_{j}^{\tilde{n}_{j}} \in \mathcal{C}_{j}(\SNR,r_{j})$ s.t. $\vecx_{j}^{n_{j}} \neq \vecx^{\tilde{n}_{j}}_{j}$ for $i,j = 1,2$ and $i \neq j$,  where  $\Delta \vecx_{i} = \vecx_{i}^{n_{i}} - \vecx^{\tilde{n}_{i}}_{i}$, $\Delta \vecx_{j} = \vecx_{j}^{n_{j}} - \vecx^{\tilde{n}_{j}}_{j}$,   and $\Delta\matX_{ij} = [\Delta\vecx_{i} \ \Delta\vecx_{j}]$, and $\lambda_{\min}(\Delta\matX_{ij} (\Delta\matX_{ij})^{H})$ denotes the smallest nonzero eigenvalue of $\Delta\matX_{ij} (\Delta\matX_{ij})^{H}$, for some\footnote{We note that $\epsilon$ is allowed to be different in \eqref{cdcmac1} and \eqref{cdcmac2}. } $\epsilon >  0$, then $P(E)$ obeys
\begin{align}
P(E) \doteq \SNR^{-d^{MAC}(\vecr)}.
\end{align}
\end{theorem}

\begin{proof}
We first identify an upper bound on the DMT and then show, using an appropriate lower bound, that this DMT is, indeed, achievable.  We define the outage events corresponding to decoding of $\mathcal{T}_{i}$, decoding of $\mathcal{T}_{j}$, and  jointly decoding of $\mathcal{T}_{i}$ \emph{and} $\mathcal{T}_{j}$ at $\mathcal{R}_{i}$ for $i,j=1,2$ and $i\neq j$ by 
\begin{align}
	\mathcal{O}_{i1}^{MAC} &\triangleq \lefto\{ \vech_{i} : I(\vecx_{i}; \vecy_{i} | \vecx_{j}, \vech_{i}) < R_{i}  \right\} \label{acharnomac1} \\
	\mathcal{O}_{i2}^{MAC} &\triangleq \lefto\{ \vech_{i} : I(\vecx_{j}; \vecy_{i} | \vecx_{i}, \vech_{i}) < R_{j}  \right\} \label{acharnomac2} \\
	\mathcal{O}_{i3}^{MAC} &\triangleq \lefto\{ \vech_{i} : I(\vecx_{i},\vecx_{j} ; \vecy_{i} | \vech_{i}) < R_{1}+R_{2} \right\}  \label{acharnomac3}.
\end{align}
We define an outage event for the MAC at $\mathcal{R}_{i}$  as
\begin{align}
	\mathcal{O}^{MAC}_{i} \triangleq \bigcup_{k=1}^{3} \mathcal{O}^{MAC}_{ik} \label{aouteven}.
\end{align}
 We define the total outage probability for treating the IC as a combination of MACs as  
\begin{align}
\prob{\mathcal{O}^{MAC}} \triangleq  \max\lefto\{\prob{	\mathcal{O}^{MAC}_{1}}\! , 	\prob{\mathcal{O}^{MAC}_{2} }\right\}.
\end{align} 
Using a standard argument along the lines of \cite{tse_mu, isit08_cgb}, we can see that assuming that both transmitters employ i.i.d. Gaussian codebooks results in no loss of optimality in terms of DMT performance.  We can therefore evaluate \eqref{acharnomac1}-\eqref{acharnomac3} as 
\begin{align}
	&\mathcal{O}_{i1}^{MAC}(\vecr) \triangleq \lefto\{ \vech_{i}:  \log\left(1+\SNR|h_{ii}|^2\right) < R_{i}  \right\} \notag \\
	&\mathcal{O}_{i2}^{MAC}(\vecr) \triangleq \lefto\{ \vech_{i}:  \log\left(1+\SNR^{\alpha}|h_{ji}|^2\right) < R_{j}  \right\} \notag \\
	&\mathcal{O}_{i3}^{MAC}(\vecr) \triangleq \notag \\
	& \lefto\{ \vech_{i} :\log\left(1+\SNR^{\alpha}|h_{ji}|^{2}+\SNR|h_{ii}|^{2}\right) < R_{1}+R_{2} \right\} \notag
\end{align}
with $i,j=1,2$ and $i\neq j$. In the following, we will also need the definitions of the no-outage events, according to
\begin{align}
	&\bar{\mathcal{O}}_{i1}^{MAC}(\vecr) \triangleq \lefto\{ \vech_{i}:  \log\left(1+\SNR|h_{ii}|^2\right) \geq R_{i}  \right\} \notag \\
	&\bar{\mathcal{O}}_{i2}^{MAC}(\vecr) \triangleq \lefto\{ \vech_{i}:  \log\left(1+\SNR^{\alpha}|h_{ji}|^2\right) \geq R_{j}  \right\} \notag \\
	&\bar{\mathcal{O}}_{i3}^{MAC}(\vecr) \triangleq \notag \\
	&\lefto\{ \vech_{i} :\log\left(1+\SNR^{\alpha}|h_{ji}|^{2}+\SNR|h_{ii}|^{2}\right) \geq R_{1}+R_{2} \right\} \notag
\end{align}
with $i,j=1,2$ and $i\neq j$.  We can now establish the asymptotic behavior of $\mathcal{O}^{MAC}_{i}$. By the union bound, we have 
\begin{align}
	\prob{\mathcal{O}^{MAC}_{i}}  & \leq  \sum_{k=1}^{3}  \prob{\mathcal{O}^{MAC}_{ik}(\vecr)} \\
									    & \doteq \max_{k=1,2,3} \prob{\mathcal{O}^{MAC}_{ik}(\vecr)}.  \label{amaxomizc}
\end{align}
It is shown in \cite{zheng_tradeoff} and \cite{Akuiyibo08} that
\begin{align}
\prob{\mathcal{O}_{i1}^{MAC}(\vecr)} &\doteq \SNR^{-d_{i1}^{MAC}(\vecr)} \\
\prob{\mathcal{O}_{i2}^{MAC}(\vecr)} &\doteq \SNR^{-d_{i2}^{MAC}(\vecr)} \\
\prob{\mathcal{O}_{i3}^{MAC}(\vecr)} &\doteq \SNR^{-d_{i3}^{MAC}(\vecr)} 
\end{align}
with 
\begin{align}
d_{i1}^{MAC}(\vecr) &= (1-r_{i})^{+}  \label{ajod1}\\
d_{i2}^{MAC}(\vecr) &= (\alpha-r_{j})^{+}\label{ajod2} \\
d_{i3}^{MAC}(\vecr) &= (1-r_{1}-r_{2})^{+}+(\alpha-r_{1}-r_{2})^{+} \label{ajod3}
\end{align}
for $i, j=1,2$ and $i \neq j$.  We point out that \eqref{ajod1} and \eqref{ajod2} define six SNR exponents $d_{ik}^{MAC}(\vecr)$, i.e., for $i=1,2$ and $k=1,2,3$.  The outage event corresponding to jointly decoding the signals from both transmitters at $\mathcal{R}_{1}$ is identical to the outage event corresponding to jointly decoding the signals from both transmitters at $\mathcal{R}_{2}$. Hence,  the corresponding SNR exponents of the outage probabilities of these events, namely, $d_{13}^{MAC}(\vecr)$ and $d_{23}^{MAC}(\vecr)$, are exactly the same. The total outage probability corresponding to treating the IC as a combination of MACs then satisfies 
\begin{align}
\prob{\mathcal{O}^{MAC}} &=  \max\lefto\{\prob{	\mathcal{O}^{MAC}_{1} } \!, 	\prob{\mathcal{O}^{MAC}_{2} }\right\}.  \label{aafasfasazzzzkmf}
\end{align}
From \eqref{amaxomizc}, it follows that
\begin{align}
	\prob{\mathcal{O}^{MAC}_{i}} &\doteq \max_{k=1,2,3} \prob{\mathcal{O}^{MAC}_{ik}(\vecr)}  \notag \\
									   &\doteq \SNR^{-\min\limits_{k=1,2,3} d_{ik}^{MAC}(\vecr)} \label{axxazsfafgafg}.
\end{align}
Hence, combining \eqref{aafasfasazzzzkmf} and \eqref{axxazsfafgafg}, we get 
 \begin{align}
\prob{\mathcal{O}^{MAC}}	  &  \doteq   \max_{i=1,2}\SNR^{-\min_{k=1,2,3} d_{ik}^{MAC}(\vecr)} \label{affffazcvgn} \\
							  & \doteq  \SNR^{-d^{MAC}(\vecr)} \label{atwentyone}\end{align}
where 
\begin{align}
d^{MAC}(\vecr) =   \min\limits_{\substack{i=1,2\\ k=1,2,3}}\left\{d^{MAC}_{ik}(\vecr)\right\}. \end{align}   We note that \eqref{affffazcvgn} can be simplified by eliminating either $d_{13}^{MAC}(\vecr)$ or $d_{23}^{MAC}(\vecr)$ as explained earlier.

With \eqref{axxazsfafgafg} we arrived at a lower bound on the error probability of the \ajml{MAC} at $\mathcal{R}_{i}$.  This lower bound, by definition, gives an upper bound on the DMT region.  We next try to find an upper bound on the error probability that has the same exponential behavior as this lower bound.  To this end, consider next the error probability corresponding to the \ajml{MAC}. We first define the relevant error events.  Let $\vecx_{i}^{n_{i}}$ and $\vecx_{j}^{n_{j}}$ with $n_{i} \in \{1,2,\ldots, 2^{NR_{i}}\}$, $n_{j} \in \{1,2,\ldots, 2^{NR_{j}}\}$ ($i,j=1,2$ and $i\neq j$) be the  codewords transmitted by $\mathcal{T}_{i}$ and $\mathcal{T}_{j}$, respectively. The results of (joint ML) decoding of $\mathcal{T}_{i}$ and $\mathcal{T}_{j}$ at $\mathcal{R}_{i}$ are denoted by $\vecx_{i}^{\tilde{n}_{i}}$ and $\vecx_{j}^{\tilde{n}_{j}}$, respectively, with $\tilde{n}_{i} \in \{1,2,\ldots, 2^{NR_{i}}\}$, $\tilde{n}_{j} \in \{1,2,\ldots, 2^{NR_{j}}\}$ for $i,j=1,2$ and $i\neq j$.   We have the error events corresponding to $\mathcal{T}_{i}$ only, $\mathcal{T}_{j}$ only, and $\mathcal{T}_{i}$ and $\mathcal{T}_{j}$ being decoded in error at $\mathcal{R}_{i}$ as
\begin{align}
\mathcal{E}_{i1}^{MAC} & \triangleq \left\{ \tilde{n}_{i}\neq n_{i}, \  \tilde{n}_{j} = n_{j}\right\} \label{macerr1}\\ 
\mathcal{E}_{i2}^{MAC} & \triangleq \left\{ \tilde{n}_{i}= n_{i}, \  \tilde{n}_{j} \neq n_{j}\right\} \label{macerr2}\\ 
\mathcal{E}_{i3}^{MAC} & \triangleq \left\{ \tilde{n}_{i}\neq n_{i}, \  \tilde{n}_{j} \neq n_{j}\right\} \label{macerr3} 
\end{align}
for $i,j=1,2$ and $i\neq j$.  We will also need the total error probability defined as 
\begin{align}
	\mathcal{E}^{MAC}_{i} \triangleq \bigcup_{k=1,2,3}\mathcal{E}_{ik}^{MAC}.
\end{align}
We denote 
\begin{align} 
	[i^{*} \ k^{*}]= \arg \min_{\substack{i=1,2 \\ k= 1,2,3}} d^{MAC}_{ik}(\vecr). 
\end{align} Let $\Xi_{ik}(\vecr) =[\xi_{ik}^{1}(\vecr) \ \xi_{ik}^{2}(\vecr)]^{T}$ be functions\footnote{We note that the functions $\Xi_{ik}(\vecr)$ might not be unique.} such that 
\begin{align}
d^{MAC}_{i^{*}k^{*}}(\vecr) = d^{MAC}_{ik}(\Xi_{ik}(\vecr))
\end{align}  for $i=1,2$, $k=1,2,3$. 

We next find an upper bound on the probability of the events $\mathcal{E}_{ik}^{MAC}$ as follows:
\begin{align}
&\prob{\mathcal{E}^{MAC}_{ik}} \notag \\ 
& = \prob{\mathcal{E}^{MAC}_{ik}, \mathcal{O}^{MAC}_{ik}(\Xi_{ik}(\vecr))} +\prob{\mathcal{E}^{MAC}_{ik}, \bar{\mathcal{O}}^{MAC}_{ik}(\Xi_{ik}(\vecr))}  \notag \\
&\leq   \prob{ \mathcal{O}^{MAC}_{ik}(\Xi_{i}(\vecr))}+ \prob{\mathcal{E}^{MAC}_{ik}| \bar{\mathcal{O}}^{MAC}_{ik}(\Xi_{i}(\vecr))} \label{afadgk}.
\end{align}
We start by deriving an upper bound on the average (w.r.t. the random channel) pairwise error probability (PEP) of each error event $\mathcal{E}^{MAC}_{ik}$ for $i=1,2$ and $k=1,2,3$.   Assuming, without loss of generality,  that we have an $\mathcal{E}^{MAC}_{i3}$ type error event, the probability of the ML decoder mistakenly deciding in favor of the codeword $\matX_{ij}^{\tilde{n}_{i}\tilde{n}_{j}}=[\vecx_{i}^{\tilde{n}_{i}} \ \vecx_{j}^{\tilde{n}_{j}}]$ when $\matX^{n_{i}n_{j}}_{ij} =[\vecx^{n_{i}}_{i} \ \vecx^{n_{j}}_{j}]$ (with $\vecx_{i}^{n_{i}}, \vecx^{\tilde{n}_{i}}_{i} \in \mathcal{C}_{i}(\SNR,r_{i})$ and $\vecx_{j}^{n_{j}}, \vecx^{\tilde{n}_{j}}_{j} \in \mathcal{C}_{j}(\SNR,r_{j})$, $i,j=1,2$ and $i\neq j$) was actually transmitted, can be upper-bounded according to
\begin{align}
&	 \mathbb{E}_{\vech_{i}}\! \lefto\{\prob{\matX_{ij}^{n_{i}n_{j}} \rightarrow \matX^{\tilde{n}_{i}\tilde{n}_{j}}_{ij}}\right\} \\ &  \leq  \mathbb{E}_{\vech_{i}}\lefto\{\exp\left[- \frac{\|\Delta\matX_{ij} \tilde{\vech}_{i} \|^2}{4} \right]\right\} \\ 
											  & \leq  \mathbb{E}_{\vech_{i}}\lefto\{\exp\left[- \frac{\lambda_{\min}\|\tilde{\vech}_{i} \|^2}{4} \right]\right\} \\ 													    
											  &  = \mathbb{E}_{\vech_{i}}\lefto\{\exp\left[- \lambda_{\min} \frac{\SNR|h_{ii}|^2+\SNR^{\alpha}|h_{ji}|^{2}}{4} \right]\right\} \label{oup}
\end{align}
where $\tilde{\vech}_{i} = [\sqrt{\SNR}h_{ii} \ \sqrt{\SNR^{\alpha}}h_{ji}]^T $  for $i,j=1,2$ and $i\neq j$ and $\lambda_{\min}$ is the smallest nonzero eigenvalue of $\Delta\matX_{ij}(\Delta\matX_{ij})^{H}$.  Noting that the no outage event  $\bar{\mathcal{O}}^{MAC}_{i3}\left(\Xi_{i3}(\vecr)\right)$ entails $\SNR|h_{ii}|^2+\SNR^{\alpha}|h_{ji}|^{2} \geq \SNR^{\xi_{i3}^{1}(\vecr)+\xi_{i3}^{2}(\vecr)}-1$, \eqref{afadgk} implies  an upper bound on $\prob{\mathcal{E}^{MAC}_{i3}}$ according to:
\begin{align}
&\mathbb{E}_{\vech_{i}} \! \lefto\{\prob{\mathcal{E}^{MAC}_{i3}}\right\} \ \dotleq \ \label{axxvafa} \\&  \prob{ \mathcal{O}_{i3}^{MAC}\left(\Xi_{3}(\vecr)\right)}+ \SNR^{N(r_{1}+r_{2})}\exp\left[-\frac{\lambda_{\min}\SNR^{\xi_{i3}^{1}(\vecr)+\xi_{i3}^{2}(\vecr)}}{4}\right]   \notag.
\end{align}
Here, we used the definitions $R_{i}=r_{i}\log\SNR$ for $i=1,2$ and $\exp[-\frac{\lambda_{\min}}{4} (\SNR^{\xi_{i3}^{1}(\vecr)+\xi_{i3}^{2}(\vecr)}-1) ] \doteq \exp[-\frac{\lambda_{\min}}{4}\SNR^{\xi_{i3}^{1}(\vecr)+\xi_{i3}^{2}(\vecr)}]$. Given that $\lambda_{\min} \ \dotgeq \ \SNR^{-\xi_{i3}^{1}(\vecr)-\xi_{i3}^{2}(\vecr)+\epsilon}$ with $\epsilon > 0$, by assumption, we obtain
\begin{align}
&\mathbb{E}_{\vech_{i}}\!  \left\{\prob{\mathcal{E}^{MAC}_{i3}}\right\}  \notag \\ 
& \ \dotleq \ \prob{ \mathcal{O}_{i3}^{MAC}\left(\Xi_{i3}(\vecr)\right)}+ \SNR^{N(r_{1}+r_{2})}\exp\left[-\frac{\SNR^{\epsilon}}{4}\right]   \label{axafkasfkaskfsfk} \\
& \doteq  \prob{ \mathcal{O}_{i3}^{MAC}\left(\Xi_{i3}(\vecr)\right)} \notag \\
&\doteq \SNR^{-d_{i^{*}k^{*}}^{MAC}(\vecr)}\label{afinzcuoqw}
\end{align}
as the second term on the RHS of \eqref{axafkasfkaskfsfk} decays exponentially in SNR whereas the first term decays polynomially. Eq. \eqref{afinzcuoqw} is a consequence of the definition of the function $\Xi_{i3}(\vecr)$. 

A similar analysis for the $\mathcal{E}^{MAC}_{i1}$-type error event results in  
 \begin{align}
 & \mathbb{E}_{\vech_{i}}\lefto\{\prob{\vecx_{i}^{n_{i}} \rightarrow \vecx^{\tilde{n}_{i}}_{i}}\right\}  \leq \notag \\ 
 & \hspace{2cm} \mathbb{E}_{\vech_{i}}\lefto\{\exp\left[-  \frac{\SNR|h_{ii}|^2\|\Delta\vecx_{i}\|^{2}}{4} \right]\right\} \label{agasgafas}
 \end{align} 
which, upon invoking \begin{align} 	\notag \| \Delta \vecx_{i}\|^{2} \ & \dotgeq \ \SNR^{-\min\lefto\{\xi_{i1}^{i}(\vecr), \xi_{j2}^{i}(\vecr)\right\}+\epsilon} \end{align} and using the fact that $\bar{\mathcal{O}}_{i1}^{MAC}\left(\Xi_{i1}(\vecr)\right)$ entails $\SNR|h_{ii}|^2 \geq \SNR^{\xi_{i1}^{i}(\vecr)}-1$, yields
\begin{align}
&\mathbb{E}_{\vech_{i}}\lefto\{\prob{\mathcal{E}^{MAC}_{i1}}\right\}   \ \dotleq \ \prob{ \mathcal{O}_{i1}^{MAC}\left(\Xi_{i1}(\vecr)\right)}+\notag \\ 
&\ \ \ \ \  \SNR^{Nr_{i}}\exp\left[-\frac{\SNR^{\xi_{i1}^{i}(\vecr)-\min\lefto\{\xi_{i1}^{i}(\vecr), \xi_{j2}^{i}(\vecr)\right\}+\epsilon}}{4}\right]  \label{amosnzzcP} \\
& \  \doteq \prob{ \mathcal{O}_{i1}^{MAC}\left(\Xi_{i1}(\vecr)\right)} \doteq \SNR^{-d_{i^{*}k^{*}}^{MAC}(\vecr)}
 \end{align}
for $i=1,2$. 

A similar analysis for the $\mathcal{E}^{MAC}_{i2}$-type error event results in  
 \begin{align}
 & \mathbb{E}_{\vech_{i}}\lefto\{\prob{\vecx_{j}^{n_{j}} \rightarrow \vecx^{\tilde{n}_{j}}_{j}}\right\}  \leq \notag \\ 
 & \hspace{2cm} \mathbb{E}_{\vech_{i}}\lefto\{\exp\left[-  \frac{\SNR^{\alpha}|h_{ji}|^2\|\Delta\vecx_{j}\|^{2}}{4} \right]\right\} \label{fagasgafas}
 \end{align} 
which, upon invoking \begin{align} 	\notag \| \Delta \vecx_{j}\|^{2} \ & \dotgeq \ \SNR^{-\min\lefto\{\xi_{j1}^{j}(\vecr), \xi_{i2}^{j}(\vecr)\right\}+\epsilon} \end{align}and using the fact that $\bar{\mathcal{O}}_{i2}^{MAC}\left(\Xi_{i2}(\vecr)\right)$ entails $\SNR^{\alpha}|h_{ji}|^2 \geq \SNR^{\xi_{i2}^{j}(\vecr)}-1$, yields
\begin{align}
&\mathbb{E}_{\vech_{i}}\lefto\{\prob{\mathcal{E}^{MAC}_{i2}}\right\}  \ \dotleq \ \prob{ \mathcal{O}_{i2}^{MAC}\left(\Xi_{i2}(\vecr)\right)}+ \notag \\ 
& \ \ \ \ \ \SNR^{Nr_{j}}\exp\left[-\frac{\SNR^{\xi_{i2}^{j}(\vecr)-\min\lefto\{\xi_{j1}^{j}(\vecr), \xi_{i2}^{j}(\vecr)\right\}+\epsilon}}{4}\right]  \label{famosnzzcP} \\
& \ \ \doteq \prob{ \mathcal{O}_{i2}^{MAC}\left(\Xi_{i2}(\vecr)\right)} \doteq \SNR^{-d_{i^{*}k^{*}}^{MAC}(\vecr)}
 \end{align}
for $i,j=1,2$ and $i\neq j$.  To complete the proof, we note that
\begin{align}
\mathbb{E}_{\vech_{i}}\lefto\{\prob{\mathcal{E}^{MAC}_{i}} \right\} & \leq \sum_{k=1}^{3}\mathbb{E}_{\vech_{i}}\lefto\{ \prob{\mathcal{E}_{ik}^{MAC}} \right\} \\
										&\ \dotleq \   \sum_{k=1}^{3} \prob{ \mathcal{O}_{ik}^{MAC}\left(\Xi_{ik}(\vecr)\right)}  \\  
										& = 3 \SNR^{-d_{i^{*}k^{*}}^{MAC}(\vecr)} \doteq \SNR^{-d^{MAC}(\vecr)} \notag.
										\end{align}
We finally get 
\begin{align}
P\lefto(E^{MAC}\right) & = \max_{i=1,2} \mathbb{E}_{\vech_{i}}\lefto\{	\prob{\mathcal{E}^{MAC}_{i}} \right\} \\
	  & \ \dotleq \ \SNR^{-d^{MAC}(\vecr)} \label{afortysizx}.
\end{align}
Since \eqref{afortysizx} gives an upper bound that matches the  lower bound in \eqref{atwentyone}, the proof is complete.
\end{proof}

\subsection{Time sharing}
We assume that the transmitters are orthogonalized in time or frequency such that each $\mathcal{T}_{i}$ ($i=1,2$) uses a  fraction $\theta_{i}$ of the channel resources with $\theta_{1}+\theta_{2} =1$ and $0\leq \theta_{i} \leq 1$.   Then, $\mathcal{T}_{i}$  enjoys an effective SISO channel $\theta_{i}$ fraction of time or frequency, and the effective transmission rate of $\mathcal{T}_{i}$ is given by $R_{i}/\theta_{i} = (r_{i}/\theta_{i})\log\SNR$.  Let $P\lefto(E^{TS}_{i}\right)$ be the average ML error probability for decoding $\mathcal{T}_{i}$ at $\mathcal{R}_{i}$ for the time sharing system.  It is shown in \cite{zheng_tradeoff}  that 
\begin{align}
\label{tianeq}	P\lefto(E^{TS}_{i}\right) \doteq \begin{cases} 
											\SNR^{-(1-r_{i}/\theta_{i})^{+}}, & \ \text{if} \ \theta_{i} > 0  \\
											1, & \ \text{if} \ \theta_{i} = 0  \\ 
											\end{cases}
\end{align}
for $i=1,2$. The achievable DMT of this strategy is then
\begin{align}
P\lefto(E^{TS}\right) & = \max\lefto\{P\lefto(E^{TS}_{1}\right) \!,  P\lefto(E^{TS}_{2}\right)\right\}. \notag 			
\end{align}
We can optimize over the parameters $\theta_{i}$ to get the best possible DMT of this strategy according to
\begin{align}
P\lefto(E^{OTS}\right) \triangleq \min_{\theta_{1}, \theta_{2}} \max\lefto\{P\lefto(E^{TS}_{1}\right)\!, P\lefto(E^{TS}_{2}\right)\right\} 
\end{align}
subject to 
\begin{align}
\theta_{1}+\theta_{2} =1 \notag \\
0\leq \theta_{i} \leq 1 \notag 
\end{align}
for $i=1,2$.


\section*{Numerical results} Figs. \ref{div05}-\ref{div15} show the DMT achieved by the fixed-power-split HK scheme (HK) in comparison to the outer bound we derived in \eqref{klo} (ETW), the outer bound in \cite{Akuiyibo08} (AL08), to treating interference as noise (TIAN), and to time-sharing (TS) for symmetric rates $r=r_{1}=r_{2}$ and for $\alpha=1/2$, $\alpha=2/3$, $\alpha=1$, and $\alpha=1.5$, respectively.  

Fig. \ref{div05} shows the achievable DMT regions and the outer bounds for $\alpha = 0.5$.  In this case, we see that the two-message, fixed-power-split HK scheme (HK) is only DMT-optimal for multiplexing rates $r < 1/4$, and falls short of achieving the outer bound \eqref{klo} (ETW) and the outer bound in \cite{Akuiyibo08} (AL08) for multiplexing rates $r \geq 1/4$.  It is interesting to note that the outer bound \eqref{klo} is better than the outer bound  in \cite{Akuiyibo08} for multiplexing rates $r < 0.45$, whereas for $r > 0.45$ the opposite is true, i.e., the outer bound \cite{Akuiyibo08} is tighter than the outer bound  \eqref{klo}.

Figs. \ref{div066}-\ref{div1} depict the achievable DMT regions and the outer bounds for $\alpha = 2/3$ and $\alpha=1$, respectively.  In these cases, we see that the two-message, fixed-power-split HK scheme (HK) is DMT-optimal and achieves the DMT outer bound in \eqref{klo}.  We also observe that the outer bound \eqref{klo} is tighter than the outer bound in \cite{Akuiyibo08} for all multiplexing rates.

In Fig. \ref{div15}, we plot the outer bounds and the achievable DMT regions for the interference level $\alpha=1.5$.   The two-message, fixed-power-split HK scheme achieves the DMT outer bound \eqref{klo}, and therefore, is DMT-optimal for $\alpha=1.5$.  We note that for $\alpha=1.5$, the outer bound \eqref{klo} and the outer bound in \cite{Akuiyibo08} are identical.
 
\begin{figure}[htbp]
\begin{center}
\includegraphics[scale=0.6]{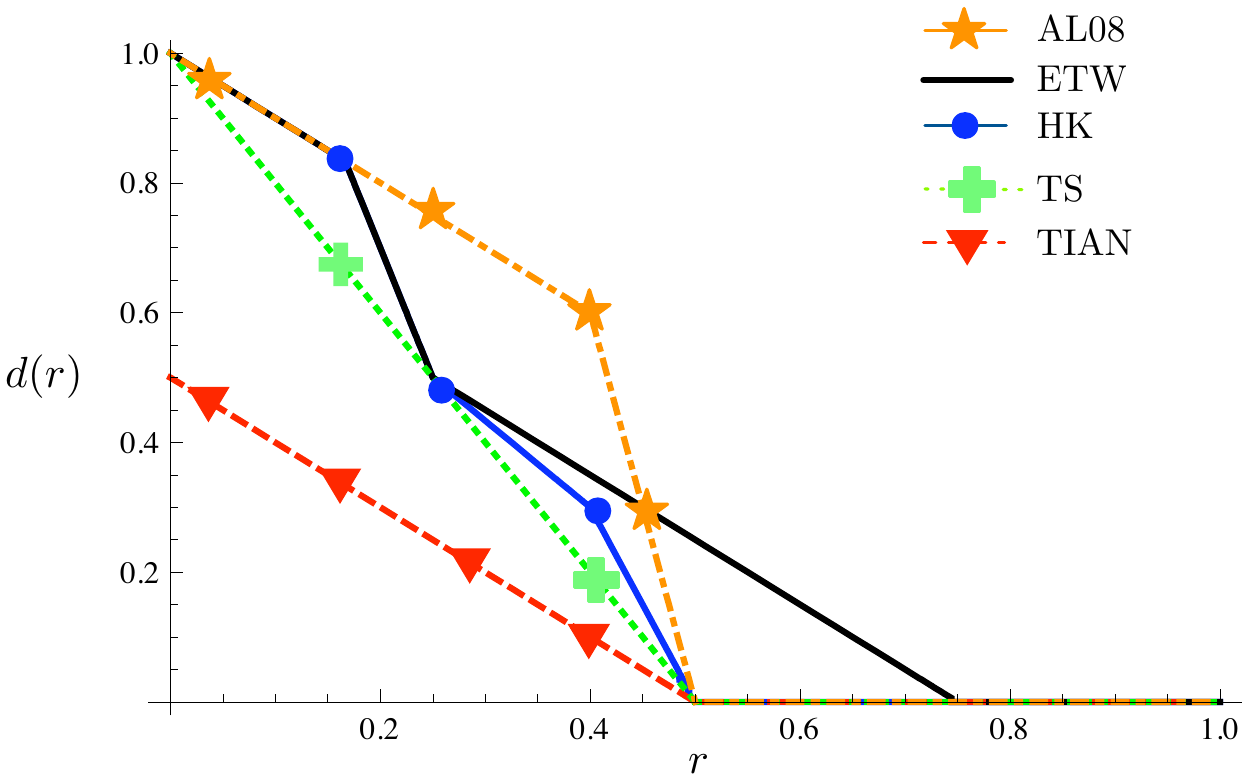}
\caption{Symmetric rate DMT for $\alpha=1/2$ and for various schemes.}
\label{div05}
\end{center}
\end{figure}

\begin{figure}[htbp]
\begin{center}
\includegraphics[scale=0.6]{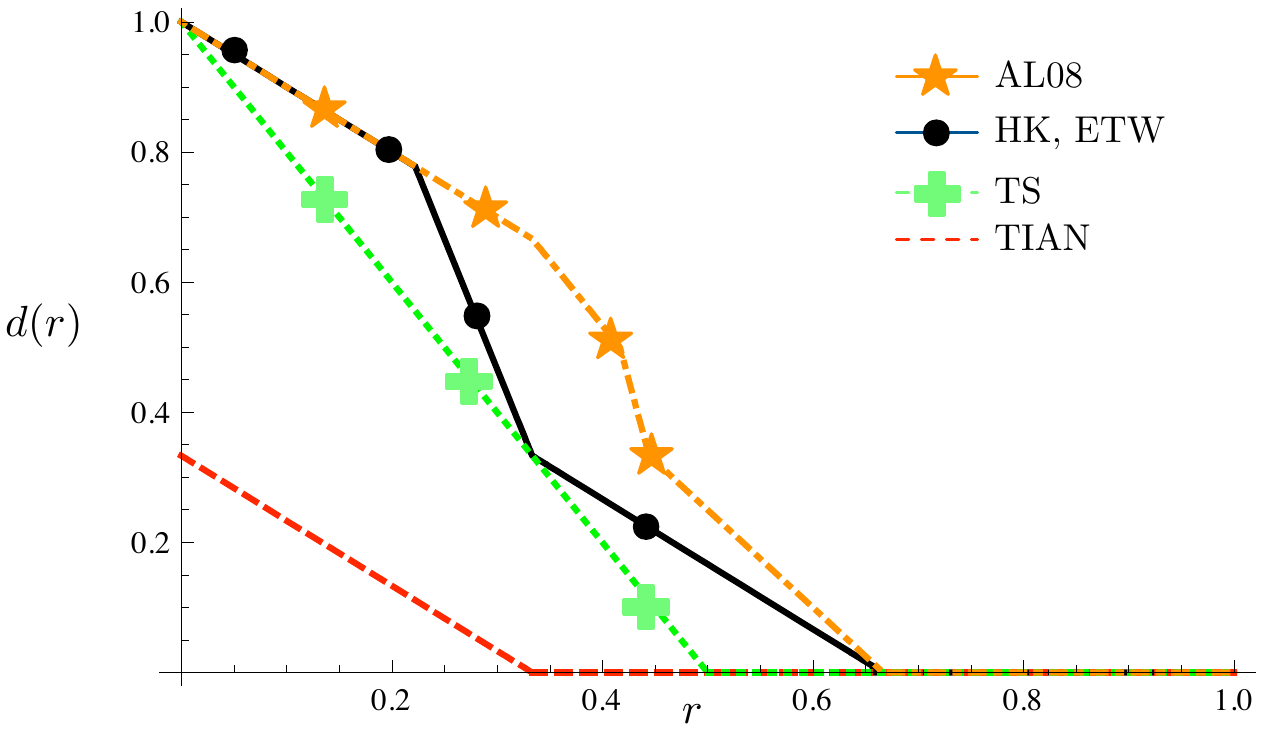}
\caption{Symmetric rate DMT for $\alpha=2/3$ and for various schemes.}
\label{div066}
\end{center}
\end{figure}

\begin{figure}[htbp]
\begin{center}
\includegraphics[scale=0.6]{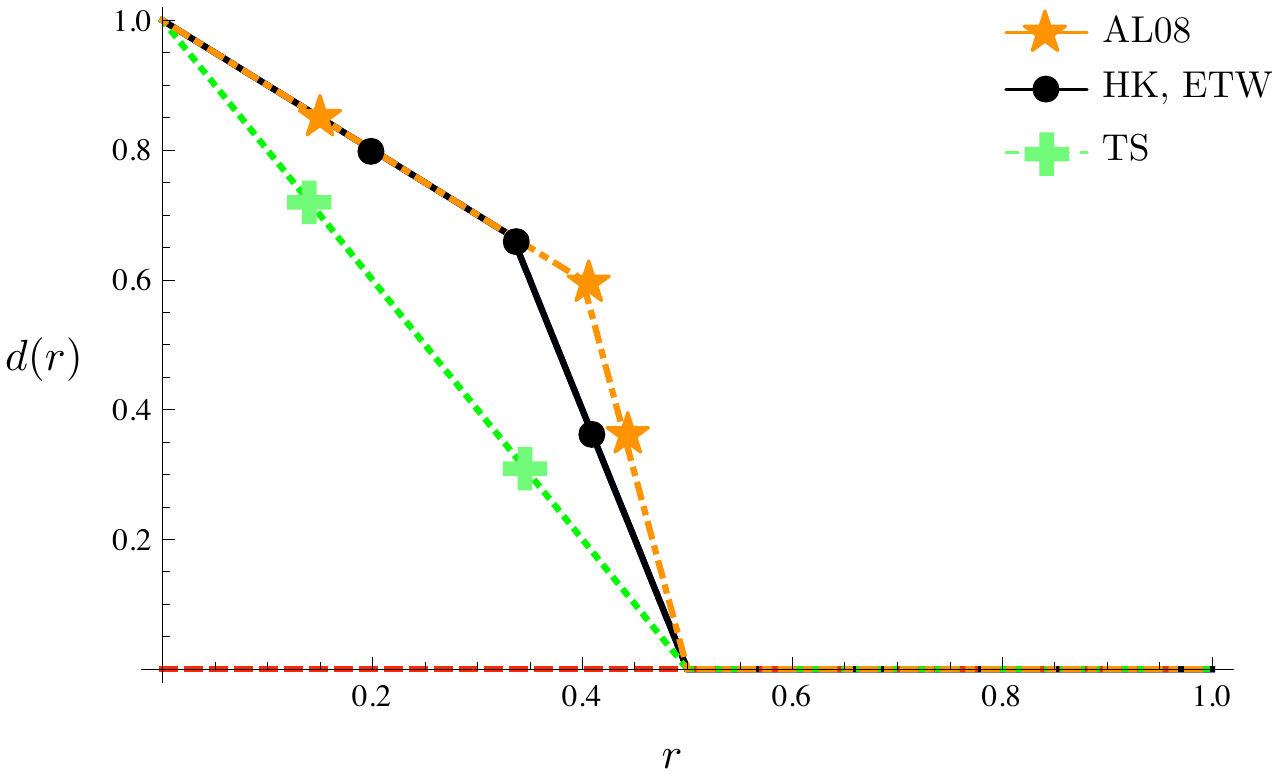}
\caption{Symmetric rate DMT for $\alpha=1$ and for various schemes.}
\label{div1}
\end{center}
\end{figure}

\begin{figure}[htbp]
\begin{center}
\includegraphics[scale=0.6]{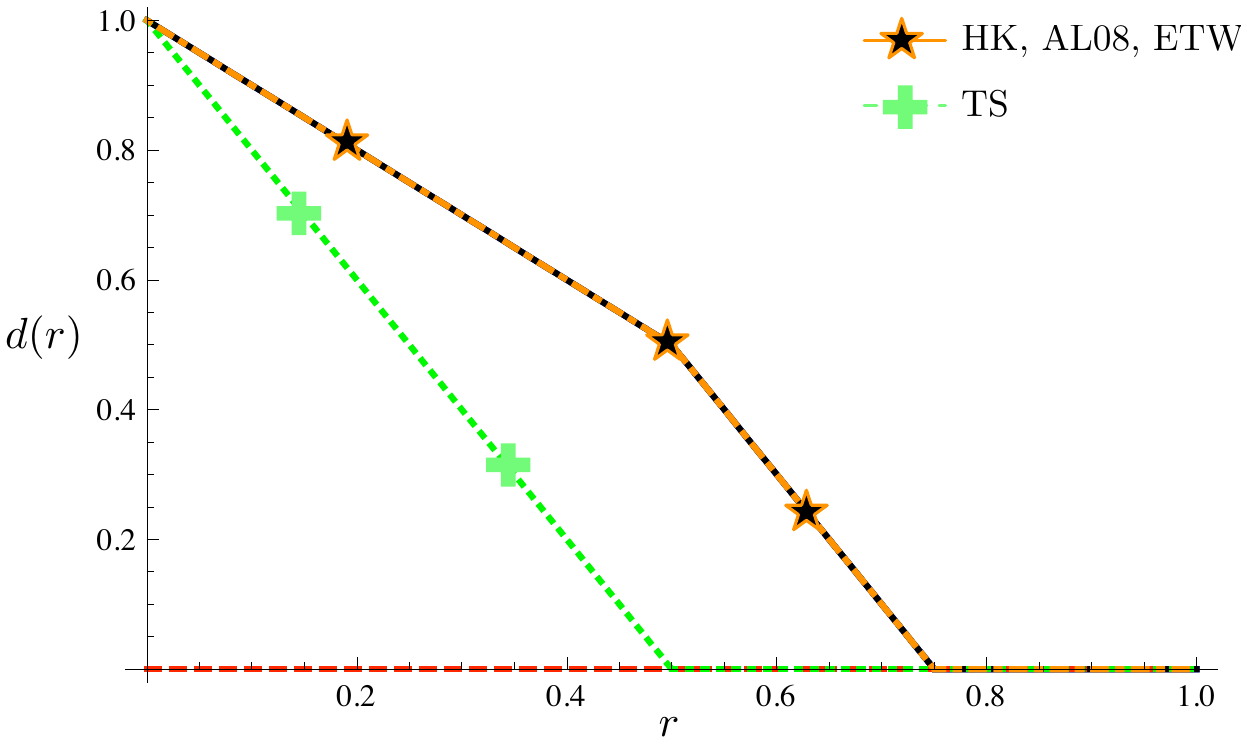}
\caption{Symmetric rate DMT for $\alpha=1.5$ and for various schemes.}
\label{div15}
\end{center}
\end{figure}

\section{Conclusions}
We characterized the optimal DMT of the two-user fading IC for the cases of  \emph{moderate}, \emph{strong}, and \emph{very strong} interference. Further,  we proved that a two-message, fixed power-split HK superposition coding scheme achieves the optimal DMT of the two-user fading IC under \emph{moderate}, \emph{strong}, and \emph{very strong} interference. We provided code design criteria for the corresponding superposition codes.    A complete characterization of the optimal DMT of the two-user fading IC under \emph{weak} interference remains an open question.
\bibliography{IEEEabrv,confs-jrnls,publishers,cebib}
\bibliographystyle{IEEEtran}

\end{document}